\documentclass[11pt,letterpaper]{article}
\pdfoutput=1

\usepackage{amsmath,amsfonts,amssymb,amsthm}
\usepackage{amsthm}
\usepackage{graphicx,color}
\usepackage{boxedminipage}
\usepackage[ruled,boxed,linesnumbered]{algorithm2e}
\usepackage{framed}
\usepackage{thmtools}
\usepackage{thm-restate}
\usepackage{xspace}
\usepackage{todonotes}
\usepackage{footnote}
\usepackage{xcolor}
\usepackage{enumitem}
\usepackage{caption}
\usepackage{subcaption}

\usepackage[margin=1in,nohead]{geometry}

\usepackage{todonotes}
\usepackage{footnote}
 \usepackage[pdftex, plainpages = false, pdfpagelabels, 
                 bookmarks=false,
                 bookmarksopen = true,
                 bookmarksnumbered = true,
                 breaklinks = true,
                 linktocpage,
                 pagebackref,
                 colorlinks = true,  
                 linkcolor = blue,
                 urlcolor  = blue,
                 citecolor = red,
                 anchorcolor = green,
                 hyperindex = true,
                 hyperfigures
                 ]{hyperref} 
 \usepackage[nameinlink]{cleveref}
 \usepackage{xifthen}
 \usepackage{tabularx}
\usepackage{tikz} 
 \usetikzlibrary{calc,decorations.pathreplacing}

\newtheorem{theorem}{Theorem}
\newtheorem{lemma}{Lemma}
\newtheorem{claim}{Claim}[section]

\newtheorem{definition}{Definition}

\newtheorem{proposition}{Proposition}

\theoremstyle{definition}

\DeclareMathOperator{\operatorClassNP}{NP}
\newcommand{\classNP}{\ensuremath{\operatorClassNP}}

\DeclareMathOperator{\operatorClassFPT}{FPT\xspace}
\newcommand{\classFPT}{\ensuremath{\operatorClassFPT}\xspace}
\DeclareMathOperator{\operatorClassW}{W}
\newcommand{\classW}[1]{\ensuremath{\operatorClassW[#1]}}

\newcommand{\ps}{\mathcal{P}}
\newcommand{\fw}{\mathcal{W}}
\newcommand{\fc}{\mathcal{C}}
\newcommand{\fcf}{\fc_\ell}
\newcommand{\fcfs}{\fc_\ell^*}
\newcommand{\col}{R}
\newcommand{\ev}{\mathcal{T}}
\newcommand{\svs}{\mathtt{start}}
\newcommand{\ordv}{\mathtt{ordv}}
\newcommand{\lkg}{linkage\xspace}
\newcommand{\wlkg}{walkage\xspace}
\newcommand{\lkgs}{linkages\xspace}
\newcommand{\wlkgs}{walkages\xspace}
\newcommand{\concs}{\circ}
\newcommand{\concm}{\diamond}
\newcommand{\bls}{\mathcal{B}}
\newcommand{\we}{\mathtt{we}}

\newlength{\RoundedBoxWidth}
\newsavebox{\GrayRoundedBox}
\newenvironment{GrayBox}[1]%
   {\setlength{\RoundedBoxWidth}{.93\textwidth}
    \def\boxheading{#1}
    \begin{lrbox}{\GrayRoundedBox}
       \begin{minipage}{\RoundedBoxWidth}}%
   {   \end{minipage}
    \end{lrbox}
    \begin{center}
    \begin{tikzpicture}%
       \node(Text)[draw=black!20,fill=white,rounded corners,%
             inner sep=2ex,text width=\RoundedBoxWidth]%
             {\usebox{\GrayRoundedBox}};
        \coordinate(x) at (current bounding box.north west);
        \node [draw=white,rectangle,inner sep=3pt,anchor=north west,fill=white] 
        at ($(x)+(6pt,.75em)$) {\boxheading};
    \end{tikzpicture}
    \end{center}}

\newenvironment{defproblemx}[2][]{\noindent\ignorespaces%
                                \FrameSep=6pt%
                                \parindent=0pt%
                \vspace*{-1.5em}
                \ifthenelse{\isempty{#1}}{%
                  \begin{GrayBox}{\textsc{#2}}%
                }{%
                  \begin{GrayBox}{\textsc{#2} parameterized by~{#1}}%
                }
                \begin{tabular*}{\textwidth}{@{\hspace{.1em}} >{\itshape} p{1.8cm} p{0.8\textwidth} @{}}%
            }{
                \end{tabular*}%
                \end{GrayBox}%
                \ignorespacesafterend
            }

\newcommand{\Oh}{\mathcal{O}}

\newcommand{\pname}{\textsc}
\newcommand{\ProblemFormat}[1]{\pname{#1}}
\newcommand{\ProblemIndex}[1]{\index{problem!\ProblemFormat{#1}}}
\newcommand{\ProblemName}[1]{\ProblemFormat{#1}\ProblemIndex{#1}{}\xspace}

\newcommand{\probkCycle}{\ProblemName{$k$-Cycle}}
\newcommand{\probkPath}{\ProblemName{$k$-Path}}
\newcommand{\probLongCycle}{\ProblemName{Longest Cycle}}
\newcommand{\probLongSTP}{\ProblemName{Longest $(s,t)$-Path}}
\newcommand{\probLongSTC}{\ProblemName{Longest $(s,t)$-Cycle}}
\newcommand{\probTCycle}{\ProblemName{$T$-Cycle}}
\newcommand{\probMaxColP}{\ProblemName{Maximum Colored $(s,t)$-Path}}
\newcommand{\probMaxColCycle}{\ProblemName{Maximum Colored Cycle}}
\newcommand{\probLongTCycle}{\ProblemName{Longest $T$-Cycle}}

\newif\iflong
\begin{document}
\longtrue

\title{Fixed-Parameter Tractability of Maximum  Colored Path and Beyond\thanks{The research leading to these results has received funding from the Research Council of Norway via the project  BWCA (grant no. 314528). Kirill Simonov acknowledges support by the Austrian Science Fund (FWF, project Y1329). Giannos Stamoulis acknowledges support by the ANR project ESIGMA (ANR-17-CE23-0010) and the French-German Collaboration ANR/DFG Project UTMA (ANR-20-CE92-0027).}
}

\author{
Fedor V. Fomin\thanks{
Department of Informatics, University of Bergen, Norway.}
\and
Petr A. Golovach\addtocounter{footnote}{-1}\footnotemark{}
\and
Tuukka Korhonen\addtocounter{footnote}{-1}\footnotemark{}
\and
Kirill Simonov\thanks{Algorithms and Complexity Group, TU Wien, Austria.}
\and 
Giannos Stamoulis\thanks{LIRMM, Universit\'e de Montpellier, CNRS, France.}
}

\date{}

\maketitle

\thispagestyle{empty}
\begin{abstract}
We introduce a general method for obtaining fixed-parameter algorithms for problems about finding paths in undirected graphs, where the length of the path could be unbounded in the parameter.
The first application of our method is as follows.

We give a randomized algorithm, that given a colored $n$-vertex undirected graph, vertices $s$ and $t$, and an integer $k$, finds an $(s,t)$-path containing at least $k$ different colors in time $2^k n^{\Oh(1)}$.
This is the first FPT algorithm for this problem, and it generalizes the algorithm of Bj\"orklund, Husfeldt, and Taslaman~[SODA~2012] on finding a path through $k$ specified vertices.
It also implies the first $2^k n^{\Oh(1)}$ time algorithm for finding an $(s,t)$-path of length at least $k$.

Our method yields FPT algorithms for even more general problems.
For example, we consider the problem where the input consists of an $n$-vertex undirected graph $G$, a matroid $M$ whose elements correspond to the vertices of $G$ and which is represented over a finite field of order $q$, a positive integer weight function on the vertices of $G$, two sets of vertices $S,T \subseteq V(G)$, and integers $p,k,w$, and the task is to find $p$ vertex-disjoint paths from $S$ to $T$ so that the union of the vertices of these paths contains an independent set of $M$ of cardinality $k$ and weight $w$, while minimizing the sum of the lengths of the paths.
We give a $2^{p+\Oh(k^2 \log (q+k))} n^{\Oh(1)} w$ time randomized algorithm for this problem.
\end{abstract}

\newpage
\pagestyle{plain}
\setcounter{page}{1}


\section{Introduction}\label{sec:intro}
The study of long cycles and paths in graphs is a popular research direction in parameterized algorithms.
Starting from the color-coding  of  Alon, Yuster, and Zwick \cite{AlonYZ95}, powerful algorithmic techniques have been developed \cite{DBLP:journals/siamcomp/Bjorklund14,BjorklundHKK17,rank-treewidth2,cut-and-count,FominLPS16,Koutis08,Williams09,Zehavi16}, see also 
\cite[Chapter~10]{cygan2015parameterized}, for finding long cycles and paths in graphs. However, most of the known methods are applicable only in the scenario when the size of the solution is bounded by the parameter. 
Let us explain what we mean by that by the following example. 

Consider two very related problems,  \probkCycle and \probLongCycle. In both problems, we are given a graph\footnote{In this paper, graphs are assumed to be undirected if it is not explicitly mentioned to be otherwise.} $G$ and an integer parameter $k$. In \probkCycle we ask whether $G$ has a cycle of length \emph{exactly} $k$. In \probLongCycle, we ask whether $G$ contains  a cycle of length \emph{at least} $k$. While in the first problem any solution should have exactly $k$ vertices, in the second problem the solution could be even a Hamiltonian cycle on $n$ vertices. The essential difference in applying color-coding (and other methods) to these problems is that for \probkCycle, a random coloring of the vertices of $G$ in $k$ colors will color the vertices of a solution cycle with different colors  with probability $e^{-k}$.
Such information about colorful solutions allows dynamic programming to solve \probkCycle (as well as the related \probkPath problem, the problem of finding a path of length exactly $k$).  However, since a  solution cycle for \probLongCycle is not upper-bounded by a function of $k$, the coloring argument falls apart. 
As Fomin et al. write  in \cite{FominLPS16} ``This is why color-coding and other techniques applicable to \probkPath do not seem to work here.'' 
Sometimes, like in the case of \probLongCycle, a simple ``edge contraction'' trick, see  \cite[Exercise~5.8]{cygan2015parameterized}, allows reducing the problem to \probkCycle.
We are not aware of general methods for solving problems related to cycles and paths when the size of the solution is not upper-bounded by the parameter.

The main result of this paper is a theorem that allows deriving algorithms for various parameterized problems about paths, cycles, and beyond, in the scenario when the size of the solution is not upper-bounded by the parameter. We discuss numerous applications of the theorem in the next subsection.

Our theorem is about finding a \emph{$k$-colored $(S,T)$-\lkg} in a colored graph.
Let $G$ be a graph, $S$ and $T$ be sets of vertices of $G$, and $p$ be a positive integer.
An {\em $(S,T)$-\lkg} of order $p$ is a set $\ps$ of $p = |\ps|$ vertex-disjoint paths, each starting in $S$ and ending in $T$.
The set of vertices in the paths of $\ps$ is denoted by $V(\ps)$.
The \emph{total length} (or often just the length) of an $(S,T)$-\lkg is the total number of vertices in its paths, i.e., $|V(\ps)|$.
For a coloring  $c : V(G) \rightarrow [n]$ of $G$,  an $(S,T)$-\lkg $\ps$ is called \emph{$k$-colored} if $V(\ps)$ contains at least $k$ different colors, i.e., $|c(V(\ps))| \ge k$.
Let us note that in the above definition the sets $S$ and $T$ are not necessarily disjoint and that the coloring $c$ is not necessarily a proper coloring in the graph-coloring sense.
We also note that for vertices $s,t \in V(G)$, an $(\{s\},\{t\})$-\lkg of order $1$ corresponds to an $(s,t)$-path.

\begin{theorem}\label{thm:maintheorem}
There is a randomized algorithm, that given as an input an $n$-vertex graph $G$, a coloring $c : V(G) \rightarrow [n]$ of $G$, two sets of vertices $S,T \subseteq V(G)$, and integers $p,k$, in time $2^{k+p} n^{\Oh(1)}$ either returns a  $k$-colored $(S,T)$-\lkg of order $p$ and of the minimum total length, or determines that $G$ has  no $k$-colored $(S,T)$-\lkg of order $p$.
\end{theorem}

Few remarks are in order. First, \Cref{thm:maintheorem} cannot be extended to directed graphs. It is easy to show, see   
\Cref{prop:hard-directed}, that finding a $2$-colored $(s,t)$-path in a $2$-colored directed graph is already 
\classNP-hard. Second, 
by another simple reduction, see \Cref{prop:hard-secoco}, it can also be observed that if the time complexity of \Cref{thm:maintheorem} could be improved to $(2-\varepsilon)^{k+p} n^{\Oh(1)}$ for $\varepsilon>0$, even in the case when $p=1$, $G$ is colored with $k$ colors, and $S=T=V(G)$, then \textsc{Set Cover} would admit a $(2-\varepsilon)^n (mn)^{\Oh(1)}$ time algorithm, contradicting the Set Cover Conjecture (SeCoCo) of Cygan~et~al.~\cite{CyganDLMNOPSW16}.
We also remark that actually we prove an even more general result than \Cref{thm:maintheorem}, our result in full generality will be stated as \Cref{thm:weightedmain}.
It can be also observed that by a simple reduction that subdivides edges, the coloring could be on the edges of $G$ instead of vertices (or on both vertices and edges).

The algorithm in  \Cref{thm:maintheorem}   invokes DeMillo-Lipton-Schwartz-Zippel lemma  for  polynomial identity testing and thus is ``heavily'' randomized.   
We do not know whether \Cref{thm:maintheorem}  could be derandomized. 
The special case of \Cref{thm:maintheorem} when the coloring is a bijection, the problem of finding an $(S,T)$-\lkg of order $p$ and of length at least $k$, can be reduced to the (rooted) topological minor containment.
To see why, observe that if we enumerate all possible collections $\ps$ of $p$ paths of total length $k$, then we can check for each collection $\ps$  if it is contained as a rooted topological minor in $G$.
The topological minor containment admits a deterministic \classFPT algorithm parameterized by the size of the pattern graph~\cite{DGroheKMW11}.
However the running time of the algorithm of Grohe et al.~\cite{DGroheKMW11} is bounded by a tower of exponents in $k$ and $p$.
Our next theorem gives a deterministic algorithm for computing an $(S,T)$-\lkg of order $p$ and of length at least $k$ 
whose running time is single-exponential in the the parameter $k$ for any fixed value of $p$.
The other advantage of the algorithm in \Cref{thm:detmaintheorem} is that it works on directed graphs too. In the following statement, a \emph{directed} $(S,T)$-\lkg is defined analogously to an $(S,T)$-\lkg, but is composed of directed paths from $S$ to $T$.

\begin{restatable}{theorem}{detmaintheorem}\label{thm:detmaintheorem}
There is a deterministic algorithm that, given an $n$-vertex digraph $G$, two sets of vertices $S,T \subseteq V(G)$, an integer $p$, and an integer $k$, in time $p^{\Oh(kp)} n^{\Oh(1)}$ either returns a directed   $(S,T)$-\lkg of order $p$ and of total length at least $k$, or determines that $G$ has  no directed  $(S,T)$-\lkg of order $p$ and total length at least $k$.
\end{restatable}

\subsection{Applications of \Cref{thm:maintheorem}}

\Cref{thm:maintheorem} implies FPT algorithms for several problems.
It encompasses a number of fixed-parameter-tractability results and improves the running times for several fundamental well-studied problems. 

\medskip
\noindent\emph{Longest path/cycle}. 
When the coloring $c : V(G) \rightarrow [n]$ is a bijection, and thus all vertices of $G$ are colored in different colors, then an  $(S,T)$-\lkg  is  {$k$-colored} if and only if its length is at least $k$.
In this case, \Cref{thm:maintheorem} outputs an $(S,T)$-\lkg  of order $p$ with at least $k$ vertices in time $2^{k+p} n^{\Oh(1)}$.
In particular, for $p=1$ it implies that \probLongSTP (i.e., for $s,t\in V(G)$ and $k\geq 0$, to decide whether there is an $(s,t)$-path of length at least $k$) is solvable in time $2^{k} n^{\Oh(1)}$.
Since  one can solve \probLongCycle  (to decide whether $G$ contains  a cycle of length at least $k$)  by solving for every edge $st\in E(G)$ the \probLongSTP problem, \Cref{thm:maintheorem} also yields an algorithm solving \probLongCycle in time $2^{k} n^{\Oh(1)}$.
To the best of our knowledge, the previous best known algorithm for \probLongSTP runs in time $4.884^k n^{\Oh(1)}$~\cite{DBLP:journals/ipl/FominLPSZ18} and the previous best known algorithm for \probLongCycle runs in time $1.66^{2k}n^{\Oh(1)}=2.76^kn^{\Oh(1)}$~\cite{BjorklundHKK17,Zehavi16}. The latter algorithm  follows by combining the result of 
 Zehavi  \cite{Zehavi16} stating that \probLongCycle is solvable in time $t(G,2k) n^{\Oh(1)}$, where $t(G,k)$ is the best known running time for solving \probkPath, with  the fastest algorithm for \probkPath of Bj{\"{o}}rklund et al.~\cite{BjorklundHKK17}.
 
 For $p=2$, the problem of finding an $(S,T)$-\lkg of length at least $k$ is equivalent to the problem of finding a cycle of length at least $k$ passing through a given pair of vertices $s,t$.
 A randomized algorithm of running time $(2e)^k n^{\Oh(1)}$ for this problem,  known as \probLongSTC, was given by Fomin et al. in \cite[Theorem~4]{FominGSS20} (see also \cite{FominGSS22}). 
 
 As we already have mentioned the problem of finding an $(S,T)$-\lkg of order $p$ and of length at least $k$ can be reduced to the (rooted) topological minor containment. 
 For $p\geq 3$, \Cref{thm:maintheorem,thm:detmaintheorem} provide the {\sl first} (randomized and deterministic) single-exponential in $k + p$ and single-exponential in $k$ for constant $p$, respectively, algorithms for computing an $(S,T)$-\lkg of order $p$ and of length at least $k$.
 For directed graphs,  \Cref{thm:detmaintheorem}  gives the first FPT algorithm for the problem parameterized by $k + p$.

 \medskip
\noindent\emph{$T$-cycle}.  In the \probTCycle problem, we are given a graph $G$ and a set $T\subseteq V(G)$ of terminals. The task is to decide whether there is a cycle passing through all terminals \cite{BjorklundHT12,Kawarabayashi08,Wahlstrom13}.
By the celebrated result of Bj{\"{o}}rklund, Husfeldt, and  Taslaman \cite{BjorklundHT12}, \probTCycle is solvable in time $2^{|T|} n^{\Oh(1)}$, and their algorithm in fact returns the shortest such cycle.
To solve \probTCycle as an application of \Cref{thm:maintheorem}, we do the following.
We pick a terminal vertex $t \in T$, create a twin vertex $s$ of $t$ (i.e., a vertex $s$ with $N(s) = N(t)$), and color $s$ and $t$ with color $1$.
We then color all non-terminal vertices of $G$ with color $1$ too.
The remaining terminal vertices $T\setminus \{t\}$ we color in $|T|-1$ colors from $2$ to $|T|$, such that no color repeats twice.
Then $G$ has a $T$-cycle if and only if there is a  $|T|$-colored $(\{s\},\{t\})$-\lkg of order $1$.
Therefore, using the algorithm of \Cref{thm:maintheorem}, we can also find the shortest $T$-cycle in time $2^{|T|} n^{\Oh(1)}$.
One could use  \Cref{thm:maintheorem} to generalize the algorithmic result of  Bj{\"{o}}rklund, Husfeldt, and  Taslaman in different settings.
For example, instead of a cycle passing through all terminal vertices, we can ask for a cycle containing at least $k$ terminals from a set $T$ of unbounded size, in time $2^k n^{\Oh(1)}$.

Another generalization of  \probTCycle comes from covering terminal vertices by at most $p$ disjoint cycles. 
For example, in the basic VRP (vehicle routing problem) one wants to route $p$ vehicles, one route per vehicle, starting and finishing at the depot so that all the customers are supplied with their demands and the total travel cost is minimized   \cite{Christofides85}. In the simplified situation when the clients are viewed as terminal vertices $T$ of a graph and routes in VRP are required to be disjoint, this problem turns into the problem of finding a ``$p$-flower'' of minimum total length containing all vertices of $T$. By $p$-flower we mean a family of $p$ cycles that intersect only in one (depot) vertex $s$. To see this problem as a problem of finding a colored $(S,T)$-\lkg, we replace the depot $s$ by a set $S$ of $2p$ vertices whose neighbors are identical to the neighbors of $s$. 
Then similar to  \probTCycle,  this variant of VRP reduces to computing a minimum length $(|T|+1)$-colored $(S,S)$-\lkg of order $p$; thus it is solvable in time $2^{|T|+p} n^{\Oh(1)}$ by \Cref{thm:maintheorem}.

 \medskip
\noindent\emph{Colored paths and cycles}. 
The problems of finding a path, cycle, or another specific subgraph in a colored graph with the maximum or the minimum number of different colors appear in different subfields of algorithms, graph theory, optimization, and operations research \cite{BroersmaLWZ05,CerulliDGR06,CohenIMTP21,CouetouxNV17,HassinMS07,KowalikL16,KumarLSS21,PanolanSZ19,wirth2001multicriteria}. In particular, the seminal color-coding technique of Alon, Yuster, and Zwick
\cite{AlonYZ95}, builds on an algorithm finding a colorful path in a $k$-colored graph, that is, a path of $k$ vertices and $k$ colors, in time $\Oh(2^{k} n)$.

 In the \probMaxColP problem, we are given  a  graph $G$ with a coloring  $c : V(G) \rightarrow [n]$ and integer $k$. The task is to identify whether $G$ contains a $k$-colored $(s,t)$-path, i.e., an $(s,t)$-path with at least $k$ different colors.
In the literature, this problem   
is also known as   \textsc{Maximum Labeled Path}~\cite{CouetouxNV17} and   \textsc{Maximum Tropical Path}~\cite{CohenIMTP21}.
\Cref{thm:maintheorem} yields the {\sl first} \classFPT algorithm for  \probMaxColP, as well as for \probMaxColCycle (decide whether $G$ contains a $k$-colored cycle).
It is also the first \classFPT algorithm for the even more restricted variant of deciding if a given $k$-colored graph contains any $k$-colored path.
A recent paper of Cohen et al.~\cite{CohenIMTP21} claims a $\Oh(2^{k} n^2)$ time  deterministic algorithm for computing a shortest $k$-colored path in a given $k$-colored graph.
Unfortunately, a closer inspection of the algorithm of Cohen et al. reveals that it computes a $k$-colored walk instead of a $k$-colored path.\footnote{
The error in~\cite{CohenIMTP21} occurs on p. 478. It is claimed that if $P$ is a shortest $(u,v)$-path that  uses the set $C$  of colors  and  $P'$ is a $(w,t)$-sub-path of $P$ using colors $C'\subseteq C$,  then $P'$ must be a shortest  $(w,t)$-path among all $(w,t)$-paths using colors $C'$.  This claim is correct for walks but not for paths.} 
  
It is interesting to note that the minimization version of the colored $(s,t)$-path, i.e., to decide whether there is an $(s,t)$-path containing at most $k$ different colors, is  \classW{1}-hard even on very restricted classes of graphs \cite{EibenK20}.

\medskip
\noindent\emph{Beyond graphs: frameworks}.  Frameworks provide a natural generalization of colored graphs. 
Following Lov{\'{a}}sz~\cite{Lovasz19}, we say that a pair $(G,M)$, where $G$ is a graph and $M=(V(G),\mathcal{I})$ is a matroid on the vertex set of $G$, is a \emph{framework}. Then we seek for  a path, cycle, or   $(S,T)$-\lkg in $G$ maximizing the rank function of $M$.
Note that frameworks $(G,M)$ where $M$ is a partition matroid generalize colored graphs.
Indeed, the universe $V(G)$ of $M$ is partitioned into color classes $L_1, \dots, L_n$ and a set $I$ is independent if $|I\cap L_i|\leq 1$ for every color $i\in [n]$. 
However, by plugging different types of matroids into the definition of the framework, we obtain problems that cannot be captured by colored graphs.

Frameworks, under the name \emph{pregeometric graphs}, were used by  Lov{\'{a}}sz in his influential work 
on representative families of linear matroids \cite{Lovasz77}. The problem of computing maximum matching in frameworks is strongly related to the matchoid,  the matroid parity, and polymatroid matching problems. See the  Matching Theory book of 
  Lov{\'{a}}sz and Plummer \cite{LovaszPlummerbook876} for an overview. In their book,  Lov{\'{a}}sz and Plummer use the term  \emph{matroid graph} for frameworks. In his most recent monograph, \cite{Lovasz19},   Lov{\'{a}}sz introduces the term frameworks, and this is the term we adopt in our work. 
 More generally,  the problems of computing specific subgraphs of large ranks in a framework, belong to the broad class of problems about submodular function optimization under combinatorial constraints \cite{CalinescuCPV11,ChekuriP05,NemhauserWF78}.

Let  $(G,M)$ be a framework and 
let $r\colon 2^{V(G)}\rightarrow\mathbb{Z}_{\geq 0}$ be the rank function of the matroid $M$.
The \emph{rank of a subgraph} $H$ of $G$ is $r(V(H))$ and we denote it by $r(H)$. We say that an $(S,T)$-\lkg  $\ps$ in a framework  $(G,M)$  is \emph{$k$-ranked} if the rank of $\ps$, that is the rank in $M$ of the elements corresponding to the vertices of the paths of $\ps$,  is at least $k$.  
With additional work involving 
 (lossy) randomized truncation of the matroid, it is possible to extend  \Cref{thm:maintheorem} from colored graphs to 
 frameworks over a general class of representable matroids.
   
\begin{restatable}{theorem}{thmalgebraic}\label{thm:main-frameworks}
There is a randomized algorithm that, given a framework $(G,M)$, where $G$ is an $n$-vertex graph and $M$ is represented as a matrix  over a finite field of order $q$, sets of vertices $S,T\subseteq V(G)$, and an integer $k$, in time $2^{p + \Oh(k^2 \log (q+k))} n^{\Oh(1)}$ 
either finds a   $k$-ranked $(S,T)$-\lkg of order $p$  and of minimum total length, 
or determines that $(G,M)$ has no  $k$-ranked $(S,T)$-\lkg of order $p$.
\end{restatable}

With minor adjustments, \Cref{thm:main-frameworks} can be adapted for frameworks with matroids
that are in general not representable over a finite field of small order.
For example, uniform matroids, and more generally transversal matroids, are representable over a finite field, but the field of representation must be large enough.
Despite this, we can apply \Cref{thm:main-frameworks} to transversal matroids.
Similarly, it is possible to apply  \Cref{thm:main-frameworks} in the situation when $M$ is represented by 
an integer matrix over rationals with entries bounded by $n^{\Oh(k)}$.

\medskip
\noindent\emph{Weighted extensions}.
\Cref{thm:maintheorem} can be extended into a weighted version in two different settings.
The first setting is to have weights on edges that affect the length of the $(S,T)$-\lkg.
It is easy to see that by subdividing edges, coloring the subdivision vertices with a new ``dummy color'', and increasing $k$ by one, all our algorithms work in the setting when the edges have polynomially-bounded positive integer weights.

The second weighted extension is more interesting.
It is to have weights on vertices, and asking for an $(S,T)$-\lkg containing a combination of weights and colors in a specific sense.
In this setting, we have in addition to the coloring $c : V(G) \rightarrow [n]$ a weight function $\we : V(G) \rightarrow \mathbb{Z}_{\ge 1}$.
For integers $k,w$, we say that an $(S,T)$-\lkg $\ps$ is $(k,w)$-colored if its vertices $V(\ps)$ contain a set $X \subseteq V(\ps)$ so that $|X|=k$, all vertices of $X$ have different colors, and the total weight of $X$ is exactly $\we(X) = \sum_{v \in X} \we(v) = w$.
This weighted version does not follow by direct reductions, but instead by a modification of \Cref{thm:maintheorem} (in our main proof, we will directly prove \Cref{thm:weightedmain} instead of \Cref{thm:maintheorem}).

\begin{theorem}\label{thm:weightedmain}
There is a randomized algorithm that, given as an input an $n$-vertex graph $G$, a coloring $c : V(G) \rightarrow [n]$ of $G$, a weight function $\we : V(G) \rightarrow \mathbb{Z}_{\ge 1}$, two sets of vertices $S,T$, and three integers $p,k,w$, in time $2^{k+p} n^{\Oh(1)} w$ either returns a $(k,w)$-colored $(S,T)$-\lkg of order $p$ and of minimum total length, or determines that no $(k,w)$-colored $(S,T)$-\lkgs of order $p$ exist.
\end{theorem}

Note that \Cref{thm:weightedmain} implies \Cref{thm:maintheorem} by setting all vertex weights to $1$ and $w=k$.
\Cref{thm:weightedmain} allows to derive some applications of our technique that do not directly follow from \Cref{thm:maintheorem}, which we proceed to describe.

 \medskip
\noindent\emph{Longest $T$-cycle}.
Recall that in the \probTCycle problem the task is to find a cycle passing through a given set $T$ of terminal vertices.
Both the algorithm of Bj\"orklund, Husfeldt, and Taslaman~\cite{BjorklundHT12}, and the application of the algorithm of~\Cref{thm:maintheorem} find in fact the shortest $T$-cycle.
A natural generalization of the \probTCycle problem is the \probLongTCycle problem, where in addition to the set $T$  we are given an integer $k$ and the task is to find a cycle of length at least $k$ passing through the terminals $T$.
\Cref{thm:weightedmain} can be used to solve \probLongTCycle in time $2^{\max(|T|,k)} n^{\Oh(1)}$ as follows.
First, if $|T| \ge k$, any $T$-cycle has length at least $k$ and we just use the algorithm for \probTCycle.
Otherwise, like in the reduction for \probTCycle, we first pick a terminal $t \in T$ and create a twin $s$ of it.
Then,
we color $s$ and $t$ with color $1$, and all the other vertices with different colors from $2$ to $n$.
We also assign weight $3$ to the terminal vertices $T$, weight $1$ to the vertex $s$, and weight $2$ to all other vertices.
We invoke \Cref{thm:weightedmain} to find an $(\{s\},\{t\})$-\lkg of order $1$ that contains a set $X$ of vertices with distinct colors, size $|X|=k$, and weight $\we(X) = 2k+|T|$.
Any such set $X$ must be a superset of $T$ and not contain $s$, and therefore the found path must correspond to a cycle of length at least $k$ passing through the terminals $T$.

\medskip
\noindent\emph{Vehicle routing with profits}.
With \Cref{thm:weightedmain}, we can give an algorithm for the vehicle routing problem in a bit more general setting.
In particular, we consider the situation where the depot has $k$ parcels, $p$ vehicles, and for each vertex $v$ we know that we obtain a profit $\we(v)$ for delivering a parcel to that vertex.
We can use \Cref{thm:weightedmain} with the same reduction as used for VRP earlier, but instead letting the coloring of the vertices to be a bijection, to obtain a $2^{k+p} n^{\Oh(1)} w$ time algorithm for determining the shortest routing by cycles intersecting only at the depot that yields a total profit of $w$.

\medskip
\noindent\emph{Longest $k$-colored $(S,T)$-\lkg}.
\Cref{thm:weightedmain} can be also used to derive a longest path version of \Cref{thm:maintheorem}, in particular an algorithm that given a graph $G$, a coloring $c : V(G) \rightarrow [n]$, two sets of vertices $S,T \subseteq V(G)$, three integers $k,p,\ell$, in time $2^{p+\ell+k} n^{\Oh(1)}$ outputs a $k$-colored $(S,T)$-\lkg of order $p$ and length at least $\ell$.
The reduction is as follows.
First, if $p \ge \ell$, then any $(S,T)$-\lkg of order $p$ has length at least $\ell$, so we can use \Cref{thm:maintheorem}.
Otherwise, we are looking for a $k$-colored $(S,T)$-\lkg that contains at least $\ell-p$ edges.
We subdivide every edge, and for each created subdivision vertex we assign a new color and weight $2k$.
For the original vertices we keep their colors and assign weight $1$.
Now, any $k$-colored $(S,T)$-\lkg of order $p$ and length at least $\ell$ corresponds to an $(S,T)$-\lkg of order $p$ that contains a set $X$ of vertices with distinct colors, size $|X|=k+\ell-p$, and weight exactly $\we(X) = (\ell-p)\cdot 2k + k$ (note that here we use the property that we are looking for an exact weight instead of maximum weight).

\medskip
\noindent\emph{Weighted frameworks}.
We consider a generalization of frameworks into weighted frameworks.
In particular, we say that a triple $(G,M,\we)$, where $G$ is a graph, $M = (V(G), \mathcal{I})$ is a matroid, and $\we : V(G) \rightarrow \mathbb{Z}_{\ge 1}$ is a weight function, is a \emph{weighted framework}.
Now we can say that an $(S,T)$-\lkg $\ps$ in a weighted framework $(G,M,\we)$ is {\em $(k,w)$-ranked} if $V(\ps)$ contains a set $X$ of vertices with $X \in \mathcal{I}$, size $|X|=k$, and weight $\we(X) = w$.
By using the same reduction as from \Cref{thm:maintheorem} to \Cref{thm:main-frameworks}, we obtain the following theorem.

\begin{restatable}{theorem}{thmweightedframeworks}\label{thm:weightedframeworkthm}
There is a randomized algorithm that given a weighted framework $(G,M,\we)$, where $G$ is an $n$-vertex graph and $M$ is represented as a matrix  over a finite field of order $q$, sets of vertices $S,T\subseteq V(G)$, and integers $p,k,w$, in time $2^{p+\Oh(k^2 \log (q+k))} n^{\Oh(1)} w$ 
either finds a $(k,w)$-ranked $(S,T)$-\lkg of order $p$ and of minimum total length, 
or determines that $(G,M,\we)$ has no $(k,w)$-ranked $(S,T)$-\lkgs of order $p$.
\end{restatable}

Note that \Cref{thm:weightedframeworkthm} implies \Cref{thm:main-frameworks} by setting all vertex weights to $1$ and $w=k$.

Finally, we remark that even though the correctness argument of our algorithm is technical, the algorithm itself is simple and practical, consisting of only simple dynamic programming over walks in the graph.
In particular, the observed practicality of the algorithm of Bj\"orklund, Husfeldt, and Taslaman~\cite{BjorklundHT12} for \probTCycle on graphs with thousands of vertices holds also for our algorithm.

\medskip
\noindent\emph{Organization of the paper.}
The rest of the paper is organized as follows.
In \Cref{section:overview} we overview our techniques and outline our algorithms.
In \Cref{sec:prelim} we recall definitions and preliminary results.
In \Cref{sec:theta} we prove the main result, i.e., \Cref{thm:weightedmain} (recall that \Cref{thm:weightedmain} implies \Cref{thm:maintheorem}).
In \Cref{sec:extframeworks} we give the extensions of our results from colored graphs to frameworks, i.e., \Cref{thm:weightedframeworkthm}.
In \Cref{sec:detKiril} we prove \Cref{thm:detmaintheorem}.
Finally, we conclude in \Cref{sec:concl}.


\section{Techniques and outline}\label{section:overview}
The techniques behind \Cref{thm:maintheorem} build on the idea of exploiting cancellation of monomials, a fundamental tool in the area \cite{DBLP:journals/siamcomp/Bjorklund14,BjorklundHKK17,BjorklundHT12,BjorklundKK16,Koutis08,KoutisW16,LingasP15,PanolanSZ19,Williams09}.
In particular, we build on the cycle-reversal-based cancellation for \probTCycle introduced by Bj\"orklund, Husfeldt, and Taslaman~\cite{BjorklundHT12}, and on the bijective labeling-based cancellation introduced by Bj\"orklund~\cite{DBLP:journals/siamcomp/Bjorklund14} (see also~\cite{BjorklundHKK17}).
The algorithm of \Cref{thm:detmaintheorem} builds on color-coding~\cite{AlonYZ95}, generalizing ideas that appeared in~\cite{FominGSS20} for finding an $(s,t)$-cycle of length at least $k$.

In \Cref{subsec:newtech} we explain the new ideas of the techniques behind \Cref{thm:maintheorem} in comparison to the earlier works, in \Cref{subsec:outline} we give a more detailed outline of the proof of \Cref{thm:maintheorem}, and in \Cref{subsec:detoutline} we give an outline of the proof of \Cref{thm:detmaintheorem}.

\subsection{New techniques for \Cref{thm:maintheorem}\label{subsec:newtech}}
Let us first focus on the single path case of \Cref{thm:maintheorem}, i.e., $p = |S| = |T| = 1$, corresponding to the question of finding a $k$-colored $(s,t)$-path.
Our algorithm is analogous to the algorithm of Bj\"orklund, Husfeldt, and Taslaman~\cite{BjorklundHT12} for \probTCycle, but instead of having the ``interesting set'' of vertices $T$ fixed in advance, our algorithm can choose any interesting set $X \subseteq V(G)$ of vertices of size $|X| = k$ included in the path ``on the fly'' in the dynamic programming over the walks.
In particular, our dynamic programming over walks can choose whether it gives a label to a vertex or not. 
This is the crucial difference to the earlier works where there would be a set of vertices $Y \subseteq V(G)$ fixed in advance so that a vertex of $Y$ would always be given a label if encountered in the walk and the vertices $V(G) \setminus Y$ would never be given labels~\cite{DBLP:journals/siamcomp/Bjorklund14,BjorklundHKK17,BjorklundHT12,PanolanSZ19}.
This would impose a limitation that because these algorithms work in time exponential in the number of labels used (i.e. $2^k$, where $k$ is the number of labels), the intersection of the found path with the set $Y$ would have to be bounded in the parameter.
This explains why the previous techniques could not yield an FPT-algorithm for \ProblemName{Maximum Colored Path}, as no such suitable set $Y$ can be fixed in advance.

Our on the fly labeling of vertices allows our algorithm to find paths that visit the same color multiple times, while still making sure that at least $k$ different colors are visited.
In particular, the interesting set $X \subseteq V(G)$ of vertices in the path that we want to label is any set of size $k$ that contains $k$ different colors.
While our dynamic programming is still a straightforward dynamic programming over walks, the main difficulty over previous works is the argument that if no solution exists, then the polynomial that we compute is zero, i.e., all unwanted walks cancel out.

First, the argument of cancellation in the case when two vertices of the same color are given a label is a now-standard application of the bijective labeling based cancellation of Bj\"orklund~\cite{DBLP:journals/siamcomp/Bjorklund14}.
Therefore, our main focus is on a cancellation argument for walks where $k$ vertices of different colors have been labeled.
Here, our starting point is the cycle reversal based cancellation argument for \probTCycle~\cite{BjorklundHT12}, but in our case significantly more arguments are needed.
In particular, the main difference to earlier works caused by the introduction of the on the fly labeling is that a vertex can occur in a walk as both labeled and unlabeled.
Very much oversimplified, this case is handled by a new label-swap cancellation argument, where a label is moved from a labeled occurrence of a vertex into an unlabeled occurrence of the vertex.
While in isolation this argument is simple, it causes significant complications when combining with the cycle reversal based cancellation, in particular because of the ``no labeled digons'' property we have to impose to the labeled walk.
However, we manage to combine these two arguments into a one very technical cancellation argument.

Then, let us move from one $(s,t)$-path to an $(S,T)$-\lkg.
This generalization of using cancellation of monomials to find multiple paths is foreshadowed by an algorithm for minimum cost flow by Lingas and Persson~\cite{LingasP15}.
However, their arguments are considerably simpler due to not having labels on the walks.

To find $(S,T)$-\lkgs, we use a similar dynamic programming to the one path case, extending the set of walks from $S$ to $T$ one walk at a time.
Here, we must introduce a new cancellation argument for the case when two different walks intersect.
This argument is again simple in isolation: take the intersection point of the two intersecting walks and swap the suffixes of them starting from this point.
First, to make sure that this operation does anything we need to make sure that the suffixes are not equal.
We do this by enforcing that the ending vertices of the walks are different already in the dynamic programming, which adds the extra $2^p$ factor to the time complexity.
The second complication is that again, this suffix swap operation does not play well together with the other cancellation arguments, and we need to again significantly increase the complexity of the combination of the three cancellation arguments.
In the end, we have to consider 18 different cases in our cancellation argument, see \Cref{def:invophi}.

The extension from \Cref{thm:maintheorem} to the weighted setting of \Cref{thm:weightedmain} is a simple modification of the dynamic programming so that also the weight of the labeled vertices $X$ is stored.
Interestingly, this argument could be extended to look for paths containing a set of vertices $X$ with any property of $X$ that could be efficiently evaluated in dynamic programming.

\subsection{Outline of \Cref{thm:maintheorem}\label{subsec:outline}}
We first give the outline of the algorithm for the single path case of finding a $k$-colored $(s,t)$-path, and then discuss the generalization to $(S,T)$-\lkg.

Superficially, our approach follows the one of Bj{\"{o}}rklund, Husfeldt, and Taslaman \cite{BjorklundHT12} developed for the \textsc{$T$-cycle} problem.
Similar to  Bj{\"{o}}rklund et al., for every length $\ell\geq 1$, we define a certain family of walks $\fcf$ and a polynomial $f(\fcf)$ so that over a finite field of characteristic 2, the polynomial $f(\fcf)$ is non-zero if the graph contains a $k$-colored $(s,t)$-path of length $\ell$, and the polynomial $f(\fcf)$ is the identically zero polynomial if the graph does not contain any $k$-colored $(s,t)$-path of length $\le \ell$. 
Then by making use of the DeMillo-Lipton-Schwartz-Zippel lemma~\cite{DBLP:journals/jacm/Schwartz80,DBLP:conf/eurosam/Zippel79}, finding a $k$-colored $(s,t)$-path of minimum length boils down to evaluating the polynomial $f(\fcf)$ at a random point for increasing values of $\ell$.

With $k$-colored path, the role similar to the role of terminal vertices in \textsc{$T$-cycle} is played by a subset $X$ of $k$ vertices of the path with $k$ different colors.
However, a priori we do not know this set $X$, and there could be $n^k$ candidates so we cannot enumerate them.
Because of that, we define the polynomial $f$ on families of \emph{labeled} $(s,t)$-walks in the graph $G$.
A labeled $(s,t)$-walk of length $\ell$ is a pair of sequences $W = ((v_1, \ldots, v_\ell), (r_1, \ldots, r_\ell))$, where $v_1, \ldots, v_\ell$ is an $(s,t)$-walk of length $\ell$, and $r_1, \ldots, r_\ell$ is a sequence of numbers from $[0, k]$ indicating a labeling.
The interpretation of the labeling is that $r_i = 0$ indicates that the index $i$ of the walk is not labeled, and $r_i \ge 1$ indicates that the index $i$ is labeled with the label $r_i$, with the interpretation that the vertex $v_i$ at this index is selected to the set $X$.

Next we present the definition of the polynomial $f$.
The polynomial $f$ is over GF($2^{3 + \lceil \log_2 n \rceil}$), which is a field of characteristic 2 and order $\ge 8n$.
With every edge $uv \in E(G)$ we associate a variable $f_e(uv)$, with every vertex $v \in V(G)$ we associate a variable $f_v(v)$, and with every color-label pair $(x,y) \in [n] \times [k]$ we associate a variable $f_c(x,y)$.
For a labeled walk $W = ((v_1, \ldots, v_\ell), (r_1, \ldots, r_\ell))$ we associate the monomial 
\[f(W) = \prod_{i=1}^{\ell-1} f_e(v_{i}v_{i+1}) \cdot \prod_{i \in [\ell] \mid r_i \neq 0} f_v(v_i) \cdot f_c(c(v_i), r_i).\]
For the family of walks ${\fcf}$, which we will define immediately, we are interested in the polynomial
\[f({\fcf}) = \sum_{W \in {\fcf}} f(W).\]

For vertices $s,t$ and integers $k,\ell$, the family $\fcf$ is the  family of all labeled $(s,t)$-walks \[W = ((v_1 = s, v_2, \ldots, v_\ell = t), (r_1, \ldots, r_\ell))\] of length $\ell$ that satisfy the following two properties. 
The first property is that the labeling $(r_1, \ldots, r_\ell)$ is bijective, meaning that every label from $[k]$ is used exactly once.
Note that this implies that every monomial of $f(\fcf)$ has degree $\ell-1+2k$, being a product of $\ell-1$ edge variables, $k$ vertex variables, and $k$ color-label pair variables.
The second property is that the labeled walk $W$ has no \emph{labeled digons}.
By that we mean that $W$ cannot have a subwalk $v_{i-1} v_i v_{i+1}$ with $v_i$ being a labeled vertex (with $r_i \neq 0$) and $v_{i-1}=v_{i+1}$.
It is not immediately clear that having no labeled digons is useful, but this will turn out to be crucial similarly to the property of having no $T$-digons in the algorithm for $T$-cycle~\cite{BjorklundHT12}.

It is not difficult to prove that when a graph has a $k$-colored $(s,t)$-path of length $\ell$, then $f(\fcf)$ is a non-zero polynomial.
Indeed, a path has no repeated vertices and thus has no labeled digons, so if we take a $k$-colored $(s,t)$-path $v_1, \ldots, v_\ell$ and let the labels $r_1, \ldots, r_\ell$ take the values from $[k]$ on $k$ vertices with $k$ different colors, then the labeled walk $W = ((v_1, \ldots, v_\ell), (r_1, \ldots, r_\ell))$ appears in $\fcf$, and thus a corresponding monomial $f(W)$ appears in $f(\fcf)$.
Because $v_1, \ldots, v_\ell$ is a path and the labeled vertices have different colors, we can recover the labeled walk $W$ uniquely from the monomial $f(W)$, and therefore the monomial $f(W)$ must occur exactly once in the polynomial $f(\fcf)$ (i.e. with coefficient $1$), and therefore $f(\fcf)$ is non-zero.

The proof of the opposite statement---absence of a $k$-colored $(s,t)$-path of length $\le \ell$ implies that  $f(\fcf)$ is zero---is more complicated.
We have to show that in this case each monomial $f(W)$ for labeled walks $W \in \fcf$ occurs an even number of times in the polynomial $f(\fcf)$, in particular that there is an even number of labeled walks $W \in \fcf$ for every monomial $f(W)$.
The proof is based on constructing an \emph{$f$-invariant fixed-point-free involution} $\phi$  on $\fcf$, that is a function $\phi : \fcf \rightarrow \fcf$ such that for every $W \in \fcf$ it holds that $f(W) = f(\phi(W))$, $\phi(W) \neq W$, and $\phi(\phi(W)) = W$.

Let us start with the easy part of the proof, that is, constructing such $\phi$ for labeled walks where two vertices of the same color are labeled (which could be two different occurrences of the same vertex).
In this case, let $1 \le i<j \le \ell$ be the lexicographically smallest pair of indices so that $c(v_i) = c(v_j)$, $r_i \neq 0$, and $r_j \neq 0$.
The function $\phi$ works by swapping $r_i$ with $r_j$.
Because each label from $[k]$ occurs in $r_1, \ldots, r_\ell$ exactly once, this results in a different labeled walk $\phi(W)$ with the same monomial $f(\phi(W)) = f(W)$, and moreover $W = \phi(\phi(W))$ holds.
After this argument, we can let $\fcfs \subseteq \fcf$ be the family of labeled walks in $\fcf$ where all labeled vertices have different colors, and we know that $f(\fcfs) = f(\fcf)$.
Therefore, we can focus only to constructing $\phi : \fcfs \rightarrow \fcfs$.

Now, the first approach (which does not work) would be to adapt the strategy of Bj{\"{o}}rklund, Husfeldt,  and Taslaman  for our purposes.
The essence of their strategy is the following. 
Since walks from $\fcfs$ do not have labeled digons and because there is no $k$-colored $(s,t)$-path of length $\le \ell$, it is possible to show that every walk $W \in \fcfs$ has a ``loop'', that is a subwalk $vUv$ starting and ending in the same vertex $v$, and so that $U$ is not a palindrome.
Then $\phi(W)$ is the walk $W'$ obtained from $W$ by reversing $U$.
This approach does not work directly in our case.
The reason is that a labeled vertex could also occur several times in a walk as unlabeled.  
Because of that, reversing a subwalk can result in a walk with a labeled digon, and thus $\phi$ could map $W$ outside the family $\fcfs$. 
For example, for a walk $abcd\hat{a}b$ (here $\hat{a}$ is a labeled vertex),   reversing  $a\overleftarrow{bcd}\hat{a}b$   results in walk  $adcb\hat{a}b$ with labeled digon $b\hat{a}b$.
A natural ``patch'' for that type of walks is to not reverse but to apply a new type of operation of swapping a label from one occurrence of a vertex to another occurrence of it. 
For example, swapping a label for $abcd\hat{a}b$ would result in $\hat{a}bcd{a}b$. 
This results in a different labeled walk contributing the same monomial $f(W)$ to the polynomial.
See~\autoref{fig:example1} for an illustration of the above examples.

\begin{figure}[ht]
\centering
\includegraphics[width=10cm]{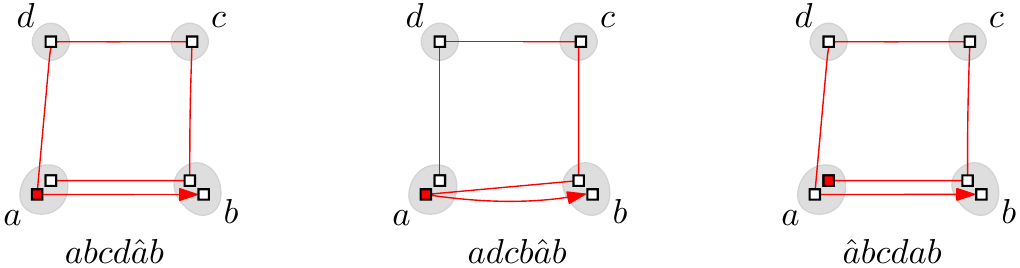}
\caption{An illustration of the walks $abcd\hat{a}b$, $a\protect\overleftarrow{bcd}\hat{a}b =adcb\hat{a}b$, and $\hat{a}bcd{a}b$.
The grey bags correspond to vertices of the graph. The squares are copies of the corresponding bag-vertex and these together with the red path illustrate the order the vertices appear in the walk. Red squares correspond to labeled vertices.}
\label{fig:example1}
\end{figure}

However,  the new operation of swapping a label brings us new problems. 
First of all, swapping a label could again result in a labeled digon. 
For example, swapping a label for  walk $\hat{a}bcac$ results in walk $abc\hat{a}c$ with labeled digon $c\hat{a}c$. 
An attempt to ``patch'' this by using a ``mixed'' strategy---when possible, swap a label, otherwise reverse---does not work either.
For example, for walk $W=\hat{a}bcac$ we cannot label swap (that will result in a labeled digon $c\hat{a}c$), hence we reverse. 
Thus we obtain walk   
$W'=\phi(W)=\hat{a}\overleftarrow{bc}ac=\hat{a}{cb}ac$. 
For $W'$,  swapping a label for $a$ is a valid operation, thus 
$\phi(W')=a{cb}\hat{a}c$, but then we would have that $\phi(\phi(W))\neq W$. See~\autoref{fig:example2} for an illustration of the above example.

\begin{figure}[ht]
\centering
\includegraphics[width=11cm]{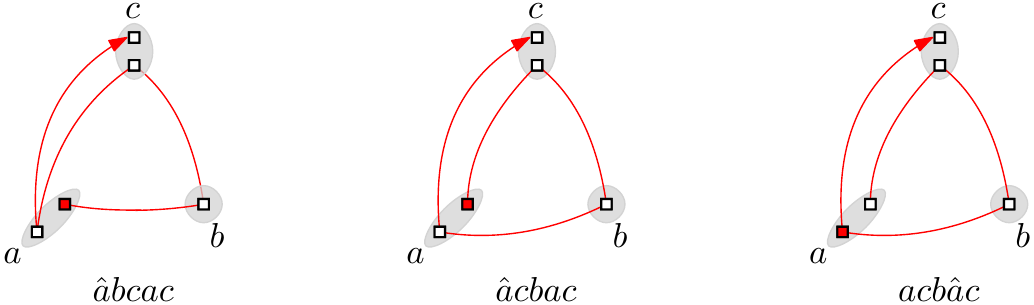}
\caption{An illustration of the walks $\hat{a}bcac$, $\hat{a}\protect\overleftarrow{bc}ac=\hat{a}{cb}ac$, and $a{cb}\hat{a}c$.
The grey bags correspond to vertices of the graph. The squares are copies of the corresponding bag-vertex and these together with the red path illustrate the order the vertices appear in the walk. Red squares correspond to labeled vertices.}
\label{fig:example2}
\end{figure}

At this moment, the situation becomes desperate: the more patches we introduce, the more issues appear and the whole construction falls apart.  
Moreover, on top of that,  one has to define the mapping $\phi$ recursively in order to deal with palindromic loops, making the situation even more complicated. 
We find it a bit surprising, that in the end a combination of label swaps and reverses allows the construction of the required mapping $\phi$. 
To make it happen, we use quite an involved strategy to identify what labels can be swapped and what subwalks could be reversed, and in what order. 
Whole \Cref{sec:proofinvo} is devoted to defining this strategy (for the more general setting of $(S,T)$-linkages) and to the proof of its correctness.

To evaluate the polynomial $f(\fcf)$, we apply quite standard dynamic programming techniques.
In particular, the polynomial can be evaluated in $2^k n^{\Oh(1)}$ time by dynamic programming over walks, where we store the length of the walk, the last two vertices of the walk, the subset of labels used so far (causing the $2^k$ factor), and whether the last vertex is labeled.
This is similar to the dynamic programming for \textsc{$T$-cycle}~\cite{BjorklundHT12}, with the difference only in that it is chosen in the dynamic programming which vertices of the walk are labeled, and that instead of a subset of $T$ we store the subset of the labels.

To extend the algorithm from a single $(s,t)$-path to an $(S,T)$-\lkg of order $p$, we define a family $\fcf$ of \emph{labeled \wlkgs} and a polynomial $f(\fcf)$ over them.
We note that by a simple reduction we can assume that $|S| = |T| = p$, and that $S$ and $T$ are disjoint.
A labeled \wlkg of order $p$ and total length $\ell$ is a $p$-tuple $\fw = (W^1, \ldots, W^p)$ of labeled walks $W^i$, whose sum of the lengths is $\ell$.
The family $\fcf$ contains labeled \wlkgs $\fw$ with the following properties:
They have order $p$, total length $\ell$, the starting vertices are ordered according to a total order on $V(G)$, ending vertices are distinct (each vertex in $T$ is an ending vertex of exactly one walk in $\fw$), the labeling is bijective (each label from $[k]$ is used exactly once), and no walk in $\fw$ contains a labeled digon.

The monomial $f(\fw)$ is then defined as
\[
f(\fw) = \prod_{i=1}^p f(W^i),
\]
and the polynomial $f(\fcf)$ as
\[
f(\fcf) = \sum_{\fw \in \fcf} f(\fw).
\]

The definitions are analogous to the single path case, in particular we recover the previously explained single path case by setting $p=1$.
The proof that if there exists a $k$-colored $(S,T)$-\lkg of order $p$ and total length $\ell$ then $f(\fcf)$ is non-zero is directly analogous to the one path case.
Also the proof that we can consider the smaller family $\fcfs \subseteq \fcf$ where all labeled vertices have different colors is analogous.

However, to prove that if there is no $k$-colored $(S,T)$-\lkgs of order $p$ and total length $\le \ell$ then $f(\fcfs)$ is identically zero we need new cancellation arguments beyond the previous cycle reversal and label swap arguments.
In particular, none of the previously considered arguments can be applied if we have a labeled \wlkg $\fw = (W^1, W^2)$ of order two, where both $W^1$ and $W^2$ are labeled paths that intersect.
In this case, the new argument is that we could swap the suffixes of $W^1$ and $W^2$ starting from the intersection point.
For example, for a \wlkg $\fw = (abc\hat{d}e, stcuv)$, we define $\phi(\fw) = (abcuv, stc\hat{d}e)$.
The property that the walks in $\fw$ have different ending vertices is crucial here to ensure that $\phi(\fw) \neq \fw$.

However, also with this suffix swap cancellation argument we run into problems.
In particular, the first challenge is that the suffix swap could create labeled digons, for example when $\fw = (ab\hat{c}de, stcbu)$ both of the walks are paths, but swapping the suffix after $c$ would create a labeled digon.
In this situation we can instead use the label swap operation on $c$, from the first walk to the second, but of course this will add again even more complications.
In the end, we manage to extend the strategy of $\phi$ from paths to \lkgs, but it makes the definition of $\phi$ even more complicated (see \Cref{def:invophi}, the path case uses the case groups A and C, while the \lkg case needs the addition of case groups B and D).

The dynamic programming for $(S,T)$-\lkg is similar to the $(s,t)$-path, extending the walks in the \wlkg one walk at the time.
It requires two new fields to store, the index of the walk that we are currently extending, and the subset of the ending vertices $T$ that have been already used.
Storing the used ending vertices causes the additional factor $2^p$ in the time complexity (as we can assume that $|T| = p$).

\subsection{Outline of \Cref{thm:detmaintheorem}\label{subsec:detoutline}}
Recall that the main difference to \Cref{thm:maintheorem} is that \Cref{thm:detmaintheorem} provides a deterministic algorithm that, moreover, works on directed graphs. The price is, however, that this algorithm is only suitable for the special case of finding an $(S, T)$-\lkg of length at least $k$, and the time complexity as a function of $k$ and $p$ is higher.
\Cref{thm:detmaintheorem} thus requires a completely different toolbox: the algorithm is based on ideas of random separation. Our result can be seen as a generalization of earlier works on finding paths and cycles of length at least $k$, the closest one being the result of Fomin et al.~\cite{FominGSS20} on finding an $(s, t)$-cycle of length at least $k$. Note that their result is stated for undirected graphs, and that the problem of finding an $(S,T)$-\lkg of order 2 and length at least $k$ is equivalent to the problem of finding an $(s,t)$-cycle of length at least $k$ (up to increasing $k$ by $2$) on undirected graphs.

Similarly to the earlier results, the case where the target $(S, T)$-\lkg is of length close to $k$ can be covered by a standard application of color-coding~\cite{AlonYZ95}. The difficulty is that the length of the $(S, T)$-\lkg can be arbitrarily larger than $k$. While because of that it would be intractable to highlight the target $(S, T)$-\lkg as a whole, it is still possible to apply random separation to give distinct colors to $k$-length segments at the end of each path in the $(S, T)$-\lkg. The main hurdle is then to argue that at least in one color we can pick a finishing segment as an arbitrary shortest path of length $k$, without intersecting any other path in the optimal solution. Afterwards, finding the desired $(S, T)$-\lkg is easy, as the length requirement is already satisfied; one only needs to find a suitable connection to complete the $(S, T)$-\lkg, which exists as witnessed by the optimal solution.

\begin{figure}[ht!]
    \centering
    \begin{subfigure}[b]{0.45\textwidth}
        \centering
        \includegraphics[width=\textwidth]{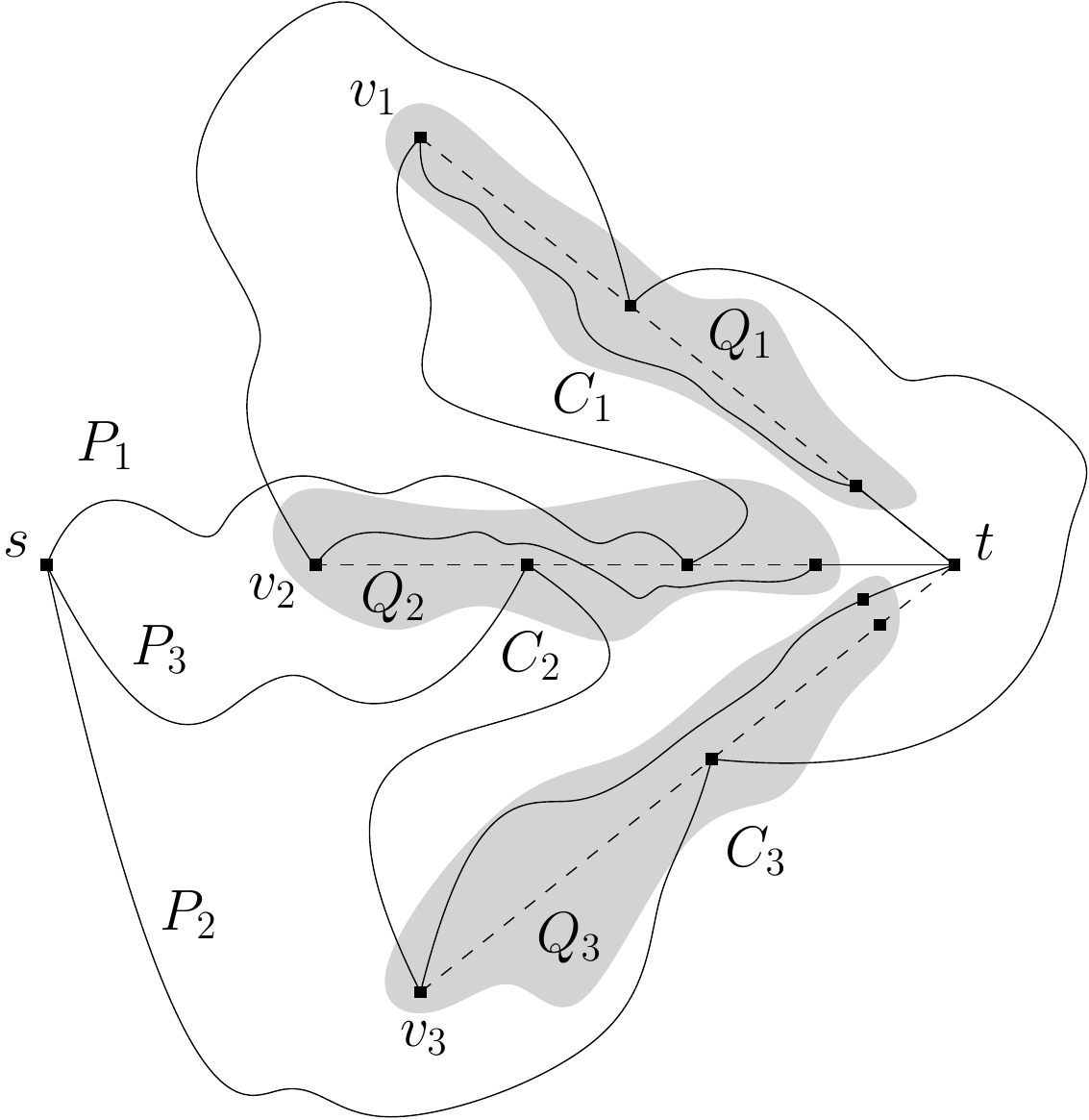}
        \caption{The given digraph $G$, edge directions are implicit along the paths $P_1$--$P_3$, and $Q_1$--$Q_3$.}
        \label{fig:tokens_a}
    \end{subfigure}
    \hfill
    \begin{subfigure}[b]{0.45\textwidth}
        \centering
        \includegraphics[width=\textwidth]{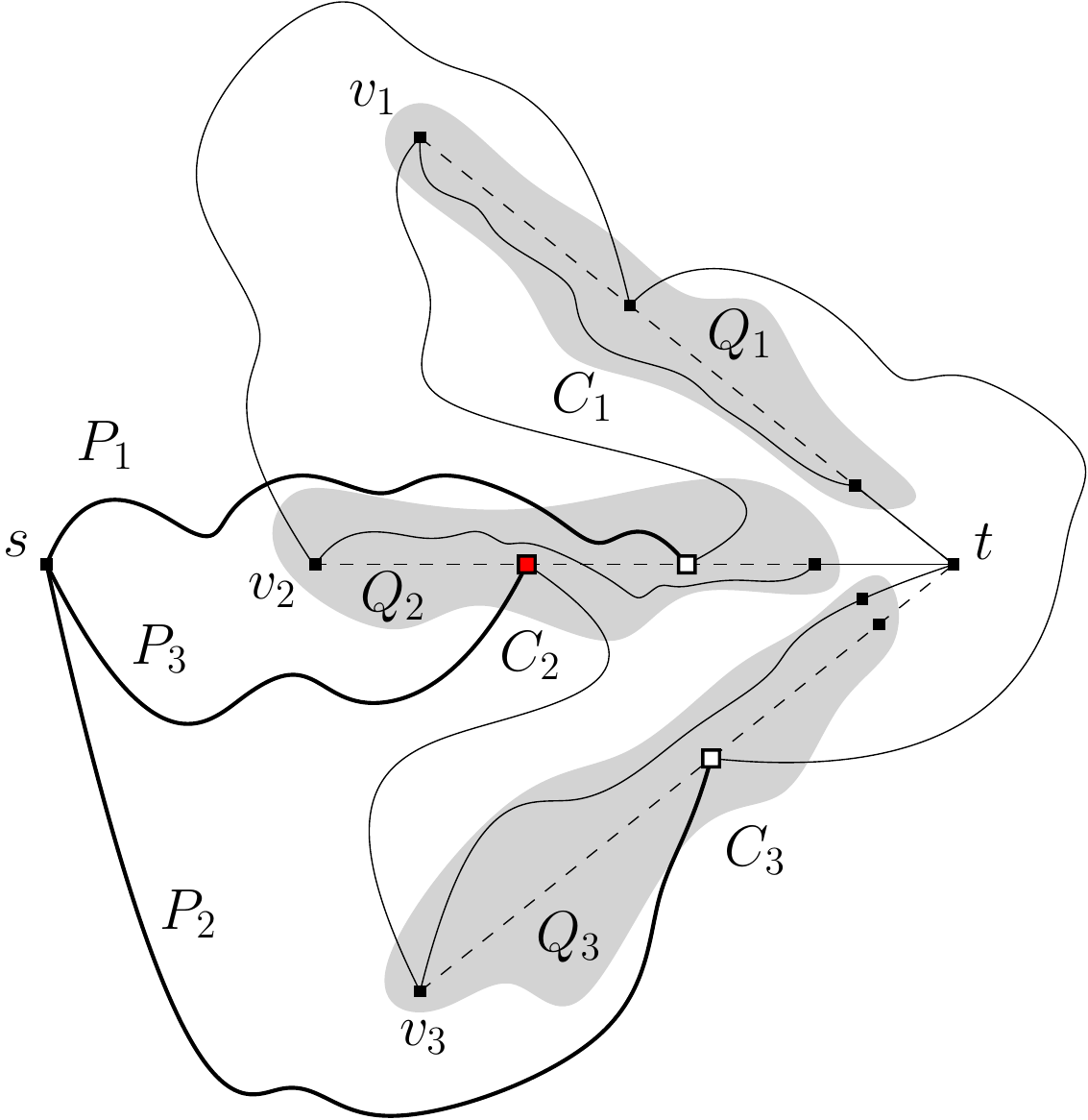}
        \caption{Starting state $S^1$: two tokens belong to $Q_2$ so \textbf{Push} is applied to $t^1_3$.}
        \label{fig:tokens_b}
    \end{subfigure}

    \begin{subfigure}[b]{0.45\textwidth}
        \centering
        \includegraphics[width=\textwidth]{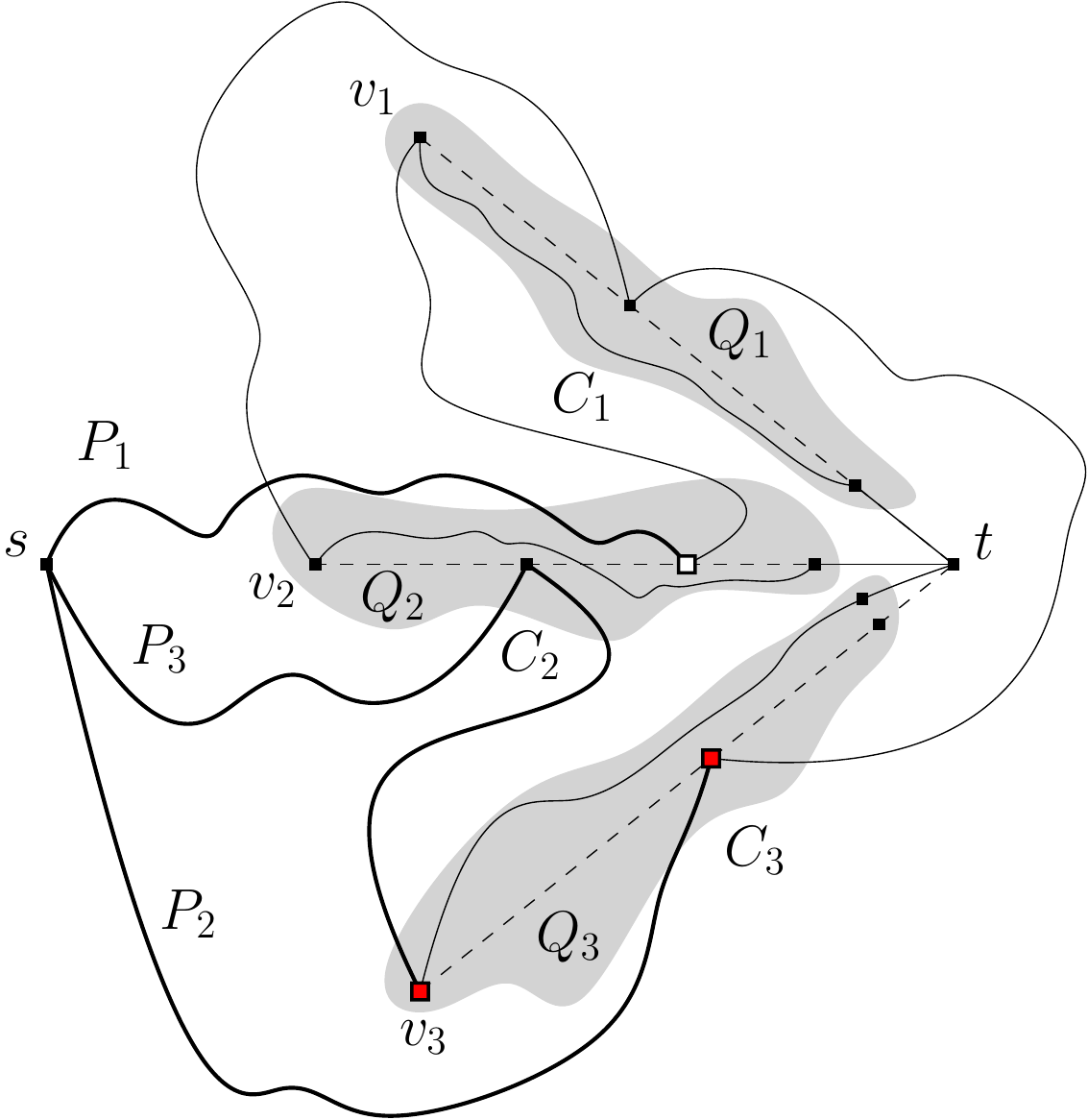}
        \caption{In $S^2$, $t^2_3 = v_3$ so \textbf{Clear} is applied to $Q_3$, moving $t^2_3$ to $t$ and $t^2_2$ next along $P_2$.}
        \label{fig:tokens_c}
    \end{subfigure}
    \hfill
    \begin{subfigure}[b]{0.45\textwidth}
        \centering
        \includegraphics[width=\textwidth]{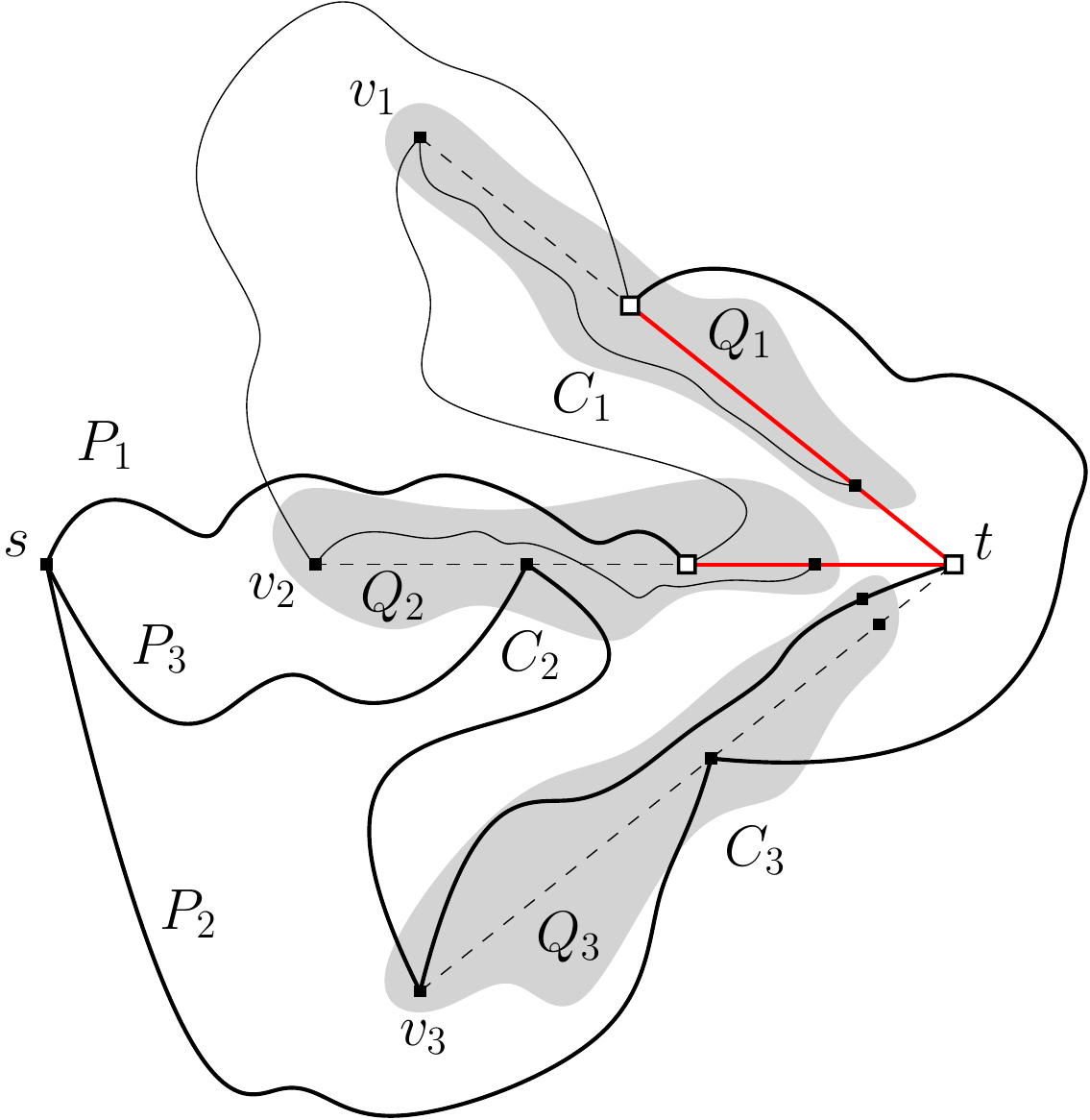}
        \caption{Finishing state $S^3$: $P_2$ is continued via $Q_1$, $P_1$ via $Q_2$, and $P_3$ is preserved to obtain the solution.}
        \label{fig:tokens_d}
    \end{subfigure}
    \caption{Illustration to the proof of \Cref{lemma:derandomize}. Large empty squares mark tokens in the current state; those with red filling are moved by the next rule.}
    \label{fig:tokens}
\end{figure}

\Cref{lemma:tokens} encapsulates the novel combinatorial result allowing the approach above, strongly generalizing a similar basic idea that appeared in~\cite{FominGSS20} for two undirected paths.
To give an intuition behind the lemma (see also \Cref{fig:tokens}), observe first that the problem of finding an $(S, T)$-\lkg of order $p$ and of total length at least $k$ is equivalent to the problem of finding an $(s, t)$\emph{-\lkg} of order $p$ and of total length at least $k + 2$, where an $(s, t)$-\lkg of order $p$ consists of $p$ internally-disjoint $(s, t)$-paths, for some $s, t \in V(G)$. Now let $(s, t)$-paths $P_1$, \ldots, $P_q$ come from the shortest solution, an $(s, t)$-\lkg of order $p$ and the smallest total length which is at least $k$. Let the sets $C_1$, \ldots, $C_q$ be the result of random separation applied to $k$-length suffixes of $P_1 - \{t\}$, \ldots, $P_q - \{t\}$, i.e., for each $i \in [q]$, the $k$-length suffix of $P_i - \{t\}$ is contained in $C_i$. The algorithm of \Cref{thm:detmaintheorem} seeks to find a solution where for some $i \in [q]$, $v_i$ is the $k$-th vertex of $P_i$ from $t$, and $Q_i$ is a $k$-length shortest path from $v_i$ to $t$ inside $C_i$, by guessing $v_i \in C_i$ and taking an arbitrary path $Q_i$ of the form above. The solution is then any collection of an $(s, v_i)$-path and $p - 1$ $(s, t)$-paths that do not intersect each other and $Q_i$, together with $Q_i$. If $Q_i$ does not intersect the $(s, v_j)$-prefix of $P_j$, for each $j \in [q]$, then the paths $P_1$, \ldots, $P_q$ certify that the desired collection of paths exists. Now comes \Cref{lemma:tokens}: it claims, roughly, that if this is not the case for all $i \in [q]$, then there is a shorter $(s, t)$-\lkg given by prefixes of $P_1$, \ldots, $P_q$ and suffixes of $Q_1$, \ldots, $Q_q$ (introducing another color to the random separation makes sure that the total length of the prefixes is still at least $k$), which is a contradiction. 

The proof of \Cref{lemma:tokens} can be imagined as the following token sliding game. First, we put a token on each $P_i$, at the first place of intersection with some $Q_j$. Then we move the tokens by applying two rules, \textbf{Push} and \textbf{Clear}. If two tokens end up on the same $Q_j$ for some $j \in [q]$, we move the farthest of them from $t$ further along its path $P_i$, until it hits another $Q_{j'}$; this is called \textbf{Push}. As for the \textbf{Clear}, if at any step $h$ the current token $t^h_i$ of the path $P_i$ reaches the vertex $v_i$, we forfeit this path: the token is moved to $t$, which corresponds to the $i$-th path of the shorter solution being exactly $P_i$, and all other tokens on $Q_i$ are moved next along their paths similarly to the rule \textbf{Push}. Moreover, every future application of any rule will not place a token on $Q_i$, skipping it to the next $Q_j$ that is still active. Clearly, this game is finite, as tokens are only being slid further along their paths. The main claim of \Cref{lemma:tokens} is that when the game is over, there is at least one remaining active token, these tokens are one per a path in $\{Q_j\}_{j \in [q]}$ (since \textbf{Push} is not applicable), and that all corresponding paths $P_i$ can be simultaneously extended each along its own $Q_j$ instead of taking their original routes, without intersections (since in \textbf{Push} we always keep the closest token to $t$). This is a shorter solution since a token of $P_i$, if active, is inside some $Q_j$ at distance less than $k$ from $t$, and the prefix of $P_i$ up to this token is shorter than the prefix of $P_i$ up to $v_i$.

Another challenge the proof of \Cref{thm:detmaintheorem} faces, is that while the random separation approach is well-known, it is normally applied to separating two, rarely three (e.g. \cite{FominGSS20}), sets. We, on the other hand, need to apply random separation to $p$ sets simultaneously, while making sure that it can be derandomized. To this end, in \Cref{lemma:derandomize} we devise in a deterministic way a family of functions that models random separation of $p$ sets of size at most $k$ each. The size of this family is bounded by $p^{\Oh(kp)} \log n$, which matches the inverse probability (up to the $\log n$ factor) of coloring the universe in $p$ colors uniformly at random so that each set receives its own color. The construction is based on perfect hash families~\cite{NaorSS95}.


\section{Preliminaries}\label{sec:prelim} 
In this section, we introduce basic notation and state some auxiliary results.

\subsection{Basic definitions and preliminary results}

We use $\mathbb{Z}_{\ge 1}$ to denote the set of positive integers and $\mathbb{Z}_{\ge 0}$ the set of non-negative integers. Also, given integers $p,q$ such that $p<q$, we use $[p,q]$ to denote the set $\{p,p+1,\ldots,q\}$ and, if $p\geq1$, we use $[p]$ to denote the set $\{1,\ldots,p\}$.

\medskip\noindent\emph{Parameterized Complexity.} We refer to the book of Cygan et al.~\cite{cygan2015parameterized} for introduction to the area. Here we only briefly mention the notions that are most important to state our results.  A \emph{parameterized problem} is a language $L\subseteq\Sigma^*\times\mathbb{N}$, where $\Sigma^*$ is a set of strings over a finite alphabet $\Sigma$. An input of a parameterized problem is a pair $(x,k)$, where $x$ is a string over $\Sigma$ and $k\in \mathbb{N}$ is a \emph{parameter}. 
A parameterized problem is \emph{fixed-parameter tractable} (or \classFPT) if it can be solved in time $f(k)\cdot |x|^{\Oh(1)}$ for some computable function~$f$.  
The complexity class \classFPT contains  all fixed-parameter tractable parameterized problems.

\medskip\noindent\emph{Graphs.}
We use standard graph-theoretic terminology and refer to the textbook of Diestel~\cite{Diestel12} for missing notions.
We consider only finite  graphs, and the considered graphs are assumed to be undirected if it is not explicitly said to be otherwise.  For  a graph $G$,  $V(G)$ and $E(G)$ are used to denote its vertex and edge sets, respectively. 
 Throughout the paper we use $n=|V(G)|=|G|$ and $m=|E(G)|$ if this does not create confusion. For a graph $G$ and a subset $X\subseteq V(G)$ of vertices, we write $G[X]$ to denote the subgraph of $G$ induced by $X$. 
For a vertex $v$, we denote by $N_G(v)$ the \emph{(open) neighborhood} of $v$, i.e., the set of vertices that are adjacent to $v$ in $G$. For $X\subseteq V(G)$, $N_G(X)=\big(\bigcup_{v\in X}N_G(v)\big)\setminus X$.
The \emph{degree} of a vertex $v$ is $d_G(v)=|N_G(v)|$. If $G$ is a digraph, $N_G^+(v)$ denotes the \emph{out-neighborhood} of $v$, i.e., the set of vertices that are adjacent to $v$ in $G$ via an arc from $v$, and $N_G^-(v)$ is the \emph{in-neighborhood}, defined symmetrically for arcs going to $v$. We may omit subscripts if the considered graph is clear from a context.

A walk $W$ of length $\ell$ in $G$ is a sequence of vertices $v_1, v_2, \ldots, v_\ell$, where $v_i v_{i+1} \in E(G)$ for all $1 \le i < \ell$.
The vertices $v_1$ and $v_\ell$ are the \emph{endpoints} of $W$ and the vertices $v_2,\ldots,v_{\ell-1}$ are the \emph{internal} vertices of $W$.
A path is a walk where no vertex is repeated.
For a path $P$ with endpoints $s$ and $t$, we say that $P$ is an $(s,t)$-path. 
A cycle is a path with the additional property that $v_\ell v_1 \in E(G)$ and $\ell \ge 3$.

\medskip\noindent\emph{DeMillo-Lipton-Schwartz-Zippel lemma.}
Our strategy involves the use the DeMillo-Lipton-Schwartz-Zippel lemma for randomized polynomial identity testing.
 
\begin{lemma}[\cite{DBLP:journals/jacm/Schwartz80,DBLP:conf/eurosam/Zippel79}]
\label{lem:schwartzzippel}
Let $p(x_1, \ldots, x_n)$ be a non-zero polynomial of total degree $d$ over a field $\mathbb{F}$, and let $S$ be a subset of $\mathbb{F}$.
If each $x_i$ is independently assigned an uniformly random value from $S$, then $p(x_1, \ldots, x_n) = 0$ with probability at most $d/|S|$.
\end{lemma}

\subsection{Hardness results}
We conclude this section by showing the  \classNP-hardness of finding a $k$-colored $(s,t)$-path on directed graphs, for any $k \ge 2$, and the optimality of the time complexity of \Cref{thm:maintheorem} assuming the Set Cover Conjecture (SeCoCo) of Cygan~et~al.~\cite{CyganDLMNOPSW16}.

We start with the hardness for directed graphs.
\begin{proposition}\label{prop:hard-directed}
For any integers $k,\ell\geq 2$, it is \classNP-complete to decide, given a directed graph $G$, a coloring $c\colon V(G)\rightarrow[\ell]$, and two vertices $s$ and $t$, whether $G$ has a $k$-colored $(s,t)$-path.  
\end{proposition} 
 \begin{proof}
 We show the claim for $k=\ell=2$ as it is straightforward to generalize the proof for other values of $k$ and $\ell$.   We reduce from the \textsc{Disjoint Paths} problem on directed graphs.  The task of this problem is, given a (directed) graph $G$ and $k$ pairs of terminal vertices $(s_i,t_i)$ for $i\in\{1,\ldots,k\}$, decide whether $G$ has vertex-disjoint $(s_i,t_i)$-paths for $i\in\{1,\ldots,k\}$. This problem is well-known to be \classNP-complete on directed graphs even if $k=2$~\cite{FortuneHW80}. 
 Consider an instance $(G,(s_1,t_1),(s_2,t_2))$ of \textsc{Disjoint Paths}, where $G$ is a directed graph. We assume that the terminal vertices are pairwise distinct. We construct the directed graph $G'$ from $G$ by adding a vertex $w$ and arcs $(t_1,w)$ and $(w,s_2)$. Note that $G$ has vertex-disjoint $(s_1,t_1)$ and $(s_2,t_2)$-paths if and only if $G'$ has an $(s_1,t_2)$-path containing $w$. We define the coloring $c$ by setting $c(w)=1$ and defining $c(v)=2$ for all $v\in V(G')\setminus\{w\}$. Clearly, $G'$ has a $2$-colored $(s_1,t_2)$-path if and only if $G'$ has an $(s_1,t_2)$-path containing $w$. This immediately implies \classNP-hardness. 
 \end{proof}

Then, we show that \Cref{thm:maintheorem} is optimal assuming the Set Cover Conjecture.

\begin{proposition}\label{prop:hard-secoco}
If there is a $(2-\varepsilon)^k n^{\Oh(1)}$ time algorithm for finding a $k$-colored path in a $k$-colored graph for some $\varepsilon > 0$, then there is a $(2-\varepsilon)^n (mn)^{\Oh(1)}$ time algorithm for \textsc{Set Cover}.
\end{proposition}
\begin{proof}
In the {\sc Set Cover} problem, we are given a universe $U$ of $n$ elements, a collection $\mathcal{S}$ of $m$ subsets of $U$, and an integer $t$ and we ask whether there is a collection $\mathcal{S}'\subseteq \mathcal{S}$ of size $t$ such that for every $u\in U$, there is a set $S\in\mathcal{S}'$ such that $u\in S$.

Given an instance $(U,\mathcal{S},t)$ of \textsc{Set Cover},
where $|U|=n$ and $\mathcal{S}=\{S_1,\ldots, S_m\}$,
we construct a graph $G$ as follows.
We first construct the graph $H$ by considering two vertices $a$ and $b$ and adding $m$ internally vertex-disjoint $(a,b)$-paths $P_{S_1},\ldots, P_{S_m}$, where for every $i\in[m]$,
the vertices in $P_{S_i}$ are bijectively mapped to the elements of $S_i$.
We call $a$ the \emph{source} of $H$ and $b$ the \emph{sink} of $H$.
We finally construct a graph $G$ that is obtained by considering $t$ copies $H_1,\ldots,H_t$ of $H$, for each $i\in[t-1]$, identifying the sink $b_i$ of $H_i$ with the source $a_{i+1}$ of $H_{i+1}$, and adding two new vertices $v$ and $v'$ of degree one, adjacent to $a_1$ and $b_t$ respectively.
See~\autoref{fig:prop2G} for an illustration of the construction of graph $G$.
Note that $t\leq m$ and $|V(G)|=(mn)^{\Oh(1)}$.

\begin{figure} 
\begin{center}
\begin{tikzpicture}
\node[circle,draw=black,fill=white,inner sep= 1.5pt,minimum width = 1pt,label={below:$v$}] (s) at (-0.8,0) {};

\node[circle,draw=black,fill=white,inner sep= 1.5pt,minimum width = 1pt,label={below:$a_1$}] (v1) at (0,0) {};

\node[circle,draw=black,fill=white,inner sep= 1.5pt,minimum width = 3pt] (v2) at (3,0) {};

\node[circle,draw=black,fill=white,inner sep= 1.5pt, minimum width = 1pt] (v3) at (6,0) {};

\node[circle,draw=black,fill=white,inner sep= 1.5pt,minimum width = 1pt] (vk) at (8,0) {};

\node[circle,draw=black,fill=white,inner sep=1.5pt,minimum width = 1pt,label={below:$~b_{t}$}] (vk1) at (11,0) {};

\node[circle,draw=black,fill=white,inner sep= 1.5pt,minimum width = 1pt,label={below:$v'$}] (t) at (11.8,0) {};

\draw (s) to (v1) (t) to (vk1);

\draw (v1) to [bend right = 20] 
node[pos=0.33,circle,fill=black,inner sep=1pt] {} node[pos=0.66,circle,fill=black,inner sep=1pt] {} 
(v2);
\draw (v1) to [bend right = 35]
 node[pos=0.25,circle,fill=black,inner sep=1pt] {} 
 node[pos=0.50,circle,fill=black,inner sep=1pt] {} node[pos=0.75,circle,fill=black,inner sep=1pt] {} 
 (v2);
\draw (v1) to [bend right = 50]
 node[pos=0.2,circle,fill=black,inner sep=1pt] {} node[pos=0.4,circle,fill=black,inner sep=1pt] {} 
node[pos=0.6,circle,fill=black,inner sep=1pt] {} 
node[pos=0.8,circle,fill=black,inner sep=1pt] {}
 (v2);

\draw (v1) to [bend left = 20]
node[pos=0.2,circle,fill=black,inner sep=1pt] {} 
node[pos=0.4,circle,fill=black,inner sep=1pt] {} 
node[pos=0.6,circle,fill=black,inner sep=1pt] {} 
node[pos=0.8,circle,fill=black,inner sep=1pt] {}
 (v2);
\draw (v1) to [bend left = 35]
node[pos=0.33,circle,fill=black,inner sep=1pt] {} 
node[pos=0.66,circle,fill=black,inner sep=1pt] {} 
 (v2);
\draw (v1) to [bend left = 50] 
node[pos=0.25,circle,fill=black,inner sep=1pt] {} 
node[pos=0.50,circle,fill=black,inner sep=1pt] {} 
node[pos=0.75,circle,fill=black,inner sep=1pt] {} 
 (v2);

\draw (v2) to [bend right = 20]
node[pos=0.33,circle,fill=black,inner sep=1pt] {} 
node[pos=0.66,circle,fill=black,inner sep=1pt] {} 
(v3);

\draw (v2) to [bend right = 35]
node[pos=0.25,circle,fill=black,inner sep=1pt] {} 
node[pos=0.50,circle,fill=black,inner sep=1pt] {} 
node[pos=0.75,circle,fill=black,inner sep=1pt] {} 
 (v3);

\draw (v2) to [bend right = 50]
node[pos=0.2,circle,fill=black,inner sep=1pt] {} 
node[pos=0.4,circle,fill=black,inner sep=1pt] {} 
node[pos=0.6,circle,fill=black,inner sep=1pt] {} 
node[pos=0.8,circle,fill=black,inner sep=1pt] {}
 (v3);

\draw (v2) to [bend left = 20] 
node[pos=0.2,circle,fill=black,inner sep=1pt] {} 
node[pos=0.4,circle,fill=black,inner sep=1pt] {} 
node[pos=0.6,circle,fill=black,inner sep=1pt] {} 
node[pos=0.8,circle,fill=black,inner sep=1pt] {}
(v3);

\draw (v2) to [bend left = 35] 
node[pos=0.33,circle,fill=black,inner sep=1pt] {} 
node[pos=0.66,circle,fill=black,inner sep=1pt] {} 
(v3);

\draw (v2) to [bend left = 50] 
node[pos=0.25,circle,fill=black,inner sep=1pt] {} 
node[pos=0.50,circle,fill=black,inner sep=1pt] {} 
node[pos=0.75,circle,fill=black,inner sep=1pt] {} 
(v3);

\draw (vk) to [bend right = 20] 
node[pos=0.33,circle,fill=black,inner sep=1pt] {} 
node[pos=0.66,circle,fill=black,inner sep=1pt] {} 
(vk1);

\draw (vk) to [bend right = 35] 
node[pos=0.25,circle,fill=black,inner sep=1pt] {} 
node[pos=0.50,circle,fill=black,inner sep=1pt] {} 
node[pos=0.75,circle,fill=black,inner sep=1pt] {} 
(vk1);

\draw (vk) to [bend right = 50] 
node[pos=0.2,circle,fill=black,inner sep=1pt] {} 
node[pos=0.4,circle,fill=black,inner sep=1pt] {} 
node[pos=0.6,circle,fill=black,inner sep=1pt] {} 
node[pos=0.8,circle,fill=black,inner sep=1pt] {}
(vk1);

\draw (vk) to [bend left = 20] 
node[pos=0.2,circle,fill=black,inner sep=1pt] {} 
node[pos=0.4,circle,fill=black,inner sep=1pt] {} 
node[pos=0.6,circle,fill=black,inner sep=1pt] {} 
node[pos=0.8,circle,fill=black,inner sep=1pt] {}
(vk1);

\draw (vk) to [bend left = 35]
node[pos=0.33,circle,fill=black,inner sep=1pt] {} 
node[pos=0.66,circle,fill=black,inner sep=1pt] {} 
 (vk1);

\draw (vk) to [bend left = 50]
node[pos=0.25,circle,fill=black,inner sep=1pt] {} 
node[pos=0.50,circle,fill=black,inner sep=1pt] {} 
node[pos=0.75,circle,fill=black,inner sep=1pt] {} 
 (vk1);

\draw [thick,decorate,
    decoration = {brace}] (-1.2,-0.65) to node[pos=0.5,label={left:$m$}] {}  (-1.2,0.65);
\node[label={$\vdots$}] () at (1.5,-0.4) {};
\node[label={$\vdots$}] () at (4.5,-0.4) {};
\node[label={$\cdots$}] () at (7,-0.4) {};
\node[label={$\vdots$}] () at (9.5,-0.4) {};

\node[label={below:$H_1$}] () at (1.5,-0.8) {};
\node[label={below:$H_2$}] () at (4.5,-0.8) {};
\node[label={below:$H_t$}] () at (9.5,-0.8) {};

\end{tikzpicture}
\caption{Construction of the graph $G$.
}\label{fig:prop2G}
\end{center}
\end{figure}
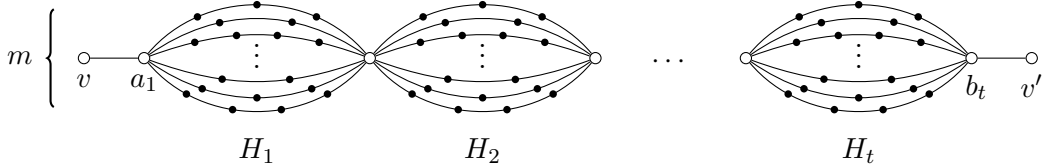

Assuming an ordering $u_1,\ldots,u_n$ of $U$,
for each $i\in[n]$, we assign color $i$ to all vertices of $G$ that correspond to $u_i$, color $n+1$ and $n+2$ to $v$ and $v'$, and color $n+3$ to all vertices in $V(G)\setminus\{v,v'\}$ that do not correspond to members of $U$.
Observe that $(U,\mathcal{S},t)$ is a yes-instance of \textsc{Set Cover} if and only if there is an $n+3$-colored path in $G$.
Therefore, a $(2-\varepsilon)^k n^{\Oh(1)}$ time algorithm for finding a $k$-colored path in a $k$-colored $n$-vertex graph implies the existence of a $(2-\varepsilon)^n (mn)^{\Oh(1)}$ time algorithm for finding a set cover of size $t$ in a universe $U$ of size $n$ with a collection $\mathcal{S}$ of $m$ subsets of $U$.
\end{proof}

 \section{Randomized algorithm for colored $(S,T)$-\lkgs \label{sec:theta}}
In this section we prove the main result, i.e., \Cref{thm:weightedmain}.
Recall that \Cref{thm:maintheorem} is a special case of \Cref{thm:weightedmain}.

Let $G$ be an $n$-vertex graph, $p$ an integer, and $S, T\subseteq V(G)$.
An $(S,T)$-\lkg of order $p$ is a set $\ps$ of $p = |\ps|$ vertex-disjoint paths between $S$ and $T$.
We denote by $V(\ps)$ the vertices in the paths of $\ps$.
The length of an $(S,T)$-\lkg is the total number $|V(\ps)|$ of vertices in the paths.
Let $c : V(G) \rightarrow [n]$ an arbitrary coloring of $G$, and $\we : V(G) \rightarrow \mathbb{Z}_{\ge 1}$ a weight function.
For positive integers $k$ and $w$, we say that an $(S,T)$-\lkg $\ps$ is $(k,w)$-colored if there exists a set $X \subseteq V(\ps)$ with $|X| = k$, all vertices of $X$ have different colors, and $\we(X) = \sum_{v \in X} \we(v)  = w$.
We give a $2^{p+k} n^{\Oh(1)} w$ time algorithm for the problem of finding a minimum length $(k,w)$-colored $(S,T)$-\lkg of order $p$ (\Cref{thm:weightedmain}).

We will assume that $|S| = |T| = p$, and $S$ and $T$ are disjoint, as the general case can be reduced to this case by the following reduction:
We add $p$ vertices $s_1, \ldots, s_p$ with $N(s_i) = S$ and $p$ vertices $t_1, \ldots, t_p$ with $N(t_i) = T$, all with the same new color and weight equal to $k \cdot \max_{v \in V(G)} \we(v) + 1$.
Then, we can set $S = \{s_1, \ldots, s_p\}$ and $T = \{t_1, \ldots, t_p\}$, and solve the problem with $k$ increased by one and $w$ increased by $k \cdot \max_{v \in V(G)} \we(v) + 1$.
Because we can assume that the original weights are at most $w+1$, this increases the target weight $w$ by a factor $\Oh(k)$, and therefore does not increase the time complexity of the algorithm.
 
\medskip

\subsection{Labeled walks and \wlkgs}
In this subsection we define labeled walks and labeled \wlkgs.

\medskip\noindent\emph{Labeled walks.}
Let $\ell$ be an integer.
A {\em walk of length $\ell$} in $G$ is a sequence of vertices $v_1, \ldots, v_\ell$ of $G$, where $v_{i} v_{i+1} \in E(G)$ for all $1 \le i < \ell$.
A {\em labeled walk of length $\ell$} is a pair of sequences $W = ((v_1, v_2, \ldots, v_\ell), (r_1, r_2, \ldots, r_\ell))$, where $v_1, \ldots, v_\ell$ is a walk of length $\ell$, and $r_1, \ldots, r_\ell$ is a sequence of integers from $[0,k]$, indicating a labeling.
The interpretation of the labeling is that $r_i = 0$ indicates that the index $i$ is unlabeled and $r_i \neq 0$ indicates that the index $i$ is labeled with the label $r_i \in [k]$.
A labeled walk is \emph{injective} if each label from $[k]$ appears in it at most once.
Most of the labeled walks that we treat in the algorithm have length at least one, but the definition allows also an empty labeled walk of length zero.
The set of vertices \emph{collected} by $W$ is $\col(W) = \{v_i \mid r_i \neq 0\}$, i.e., the set of vertices that occur at labeled indices.
The set of edges of $W$ is $E(W) = \{v_{i} v_{i+1} \mid 1 \le i < \ell\}$.
An index $i$ in a labeled walk of length $\ell$ is a \emph{digon} if $1 < i < \ell$ and $v_{i-1} = v_{i+1}$ (see~\autoref{fig:digon} for an illustration).
An index $i$ in a labeled walk is a \emph{labeled digon} if it is a digon and $r_i \neq 0$.
\begin{figure}[ht]
\centering
\begin{tikzpicture}
\node[fill=black,circle,inner sep=0pt,minimum size=5pt, label={above left:$v_{i-1}$}] (A) at (0,0) {};
\node[fill=black,circle,inner sep=0pt,minimum size=5pt, label={above right:$v_{i+1}$}] () at (0,0) {};
\node[fill=black,circle,inner sep=0pt,minimum size=5pt, label={above:$v_{i}$}] (B) at (0,1.5) {};
\node[fill=black,circle,inner sep=0pt,minimum size=5pt, label={above:$v_{i-2}$}] (V) at (-2,-0) {};
\node[fill=black,circle,inner sep=0pt,minimum size=5pt, label={above:$v_{i+2}$}] (U) at (2,-0) {};

\draw[dashed] ($(V)+(-1,0)$) -- (V.center);

\draw (V.center)  to   (A.center) to  (B.center) to (A.center) to (U.center);

\node[label={right:$v_{i}v_{i+1}$}] () at (0.1,0.8)  {};
\node[label={left:$v_{i-1}v_{i}$}] () at (-0.1,0.8)  {};

\draw[dashed] ($(U)+(1,0)$) -- (U.center);

\end{tikzpicture}
\caption{An example of a labeled walk $W$ with a digon $i$.}
\label{fig:digon}
\end{figure}
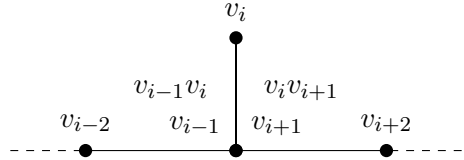

\medskip\noindent\emph{Labeled \wlkgs.}
A \emph{labeled \wlkg of order $p$} is a tuple $\fw = (W^1, \ldots, W^p)$, where each $W^i = ((v^i_1, \ldots, v^i_{\ell_i}), (r^i_1, \ldots, r^i_{\ell_i}))$ is a labeled walk of length $\ell_i \ge 1$.
The {\em length} of $\fw$ is $\sum_{i=1}^p \ell_i$.
The {\em set of edges} of $\fw$ is $E(\fw) = \bigcup_{i=1}^p E(W^i)$.
The set of vertices {\em collected} by $\fw$ is $\col(\fw) = \bigcup_{i=1}^p \col(W^i)$.
The weight $\we(\fw)$ of $\fw$ is the sum of the weights of the labeled vertices, i.e., $\we(\fw) = \sum_{i=1}^p \sum_{j \in [\ell_i] \mid r^i_j \neq 0} \we(v^i_j)$.
Note that the weight of a vertex can be counted more than once if the vertex occurs as labeled more than once.
A labeled \wlkg is \emph{injective} if each label from $[k]$ appears in it at most once, and \emph{bijective} if each label from $[k]$ appears in it exactly once.
Note that every labeled walk in an injective labeled \wlkg is injective.

The set of \emph{ending} vertices of a labeled \wlkg $\fw$ of order $p$ is $\ev(\fw) = \{v^i_{\ell_i} \mid i \in [p]\}$.
The tuple of \emph{starting} vertices of $\fw$ is $\svs(\fw) = (v^1_1, \ldots, v^p_1)$.
Let $<$ be a total order on $V(G)$.
A labeled \wlkg is \emph{ordered} if $\svs(\fw)$ is ordered according to $<$, i.e., $v^i_1 < v^{i+1}_1$ holds for all $1 \le i < p$.
The asymmetry that the starting vertices are an ordered tuple while the ending vertices are an unordered set is essential for our algorithm.
A \emph{labeled \lkg} is a labeled \wlkg where every vertex occurs at most once, i.e., the walks are vertex-disjoint paths.

We also define \emph{semiproper} and \emph{proper} labeled \wlkgs.
The intuition here is that,
in~\Cref{subsec:thm1:algo}, we define a polynomial  over semiproper \wlkgs (see also \Cref{def:familywalkages}). Then, \wlkgs that are semiproper but not proper will be handled by using standard techniques and therefore we can focus on proper \wlkgs.
Dealing with proper \wlkgs will be the most technical part of the proof.
A labeled \wlkg is \emph{semiproper} if it is injective, no walk in it contains labeled digons, and the ending vertices of the \wlkg are distinct, i.e., $v^i_{\ell_i} \neq v^j_{\ell_j}$ for $i \neq j$.
A labeled \wlkg $\fw$ is \emph{proper} if it is semiproper and all of its labeled indices correspond to vertices of different colors, i.e., $|c(\col(\fw))| = |\{(i,j) \mid r^i_{j} \neq 0\}|$.
Note that being proper implies that no vertex is labeled twice, and note that if $\fw$ is bijective and proper then $|c(\col(\fw))| = k$.

\subsection{Algorithm}\label{subsec:thm1:algo}
We assume that there is a total order $<$ on $V(G)$, and for a set $S \subseteq V(G)$ we denote by $\ordv(S)$ the tuple containing the elements of $S$ ordered according to $<$.
Note that $G$ contains a $(k,w)$-colored $(S,T)$-\lkg of order $p$ and length $\ell$ if and only if there is a bijective proper ordered labeled \lkg $\fw$ with order $p$, length $\ell$, weight $\we(\fw) = w$, tuple of starting vertices $\svs(\fw) = \ordv(S)$, and set of ending vertices $\ev(\fw) = T$.
We define a family of labeled \wlkgs that includes all such labeled \lkgs, but relaxes the condition of being a \lkg to \wlkg, and the condition of being proper to semiproper.

For each integer $\ell$, we define a family of labeled \wlkgs $\fcf$ of length $\ell$.

\begin{definition}[Family $\fcf$]\label{def:familywalkages}
Let $\ell$ a positive integer.
The family $\fcf$ consists of the bijective semiproper ordered labeled \wlkgs $\fw$ with order $p$, length $\ell$, weight $\we(\fw) = w$, tuple of starting vertices $\svs(\fw) = \ordv(S)$, and set of ending vertices $\ev(\fw) = T$.
\end{definition}

\medskip\noindent\emph{Definition of the polynomial.}
Let $q = 2^{3 + \lceil \log_2 n \rceil}$ and keep in mind that GF($q$) is a finite field of characteristic 2 and order $q \ge 8n$.
Next, we define a polynomial over GF($q$) that will be evaluated at a random point by our algorithm.
For each edge $uv \in E(G)$ we associate a variable $f_e(uv)$, for each vertex $v \in V(G)$ we associate a variable $f_v(v)$, and for each color-label-pair $(x,y) \in [n] \times [k]$ we associate a variable $f_c(x,y)$.
For a labeled walk $W = ((v_1, \ldots, v_\ell), (r_1, \ldots, r_\ell))$ we associate the monomial 
\[f(W) = \prod_{i=1}^{\ell-1} f_e(v_{i}v_{i+1}) \cdot \prod_{i \in [\ell] \mid r_i \neq 0} f_v(v_i) \cdot f_c(c(v_i), r_i).\]
For a labeled \wlkg $\fw = (W^1, \ldots, W^p)$ we associate the monomial
\[f(\fw) = \prod_{i=1}^p f(W^i).\]
For a family $\mathcal{F}$ of labeled \wlkgs we associate the polynomial
\[f(\mathcal{F}) = \sum_{\fw \in \mathcal{F}} f(\fw).\]

In particular, because the \wlkgs in $\fcf$ are bijective, every monomial in the polynomial $f(\fcf)$ has degree $\ell-p+2k$, being a product of $\ell-p$ variables corresponding to the edges of the \wlkg, $k$ variables corresponding to the labeled vertices, and $k$ variables corresponding to the color-label-pairs.

\medskip\noindent\emph{Algorithm for finding a $(k,w)$-colored $(S,T)$-\lkg.}
Our algorithm for finding a $(k,w)$-colored $(S,T)$-\lkg of order $p$ works as follows.
Starting with $\ell = p$, we evaluate the polynomial $f(\fcf)$ at a random point $x$ over GF($q$), for increasing values of $\ell$.
If $f(\fcf)(x) \neq 0$, we return that $G$ contains a $(k,w)$-colored $(S,T)$-\lkg of order $p$, and moreover that the shortest $(k,w)$-colored $(S,T)$-\lkg of order $p$ has length $\ell$.
Otherwise, we continue increasing $\ell$ until $\ell = n+1$ in which case we return that $G$ does not contain a $(k,w)$-colored $(S,T)$-\lkg of order $p$.
\medskip

For the proof of correctness of the algorithm, in \Cref{subsec:randcorr} we show that with probability at least $1/2$ this algorithm returns the length of the shortest $(k,w)$-colored $(S,T)$-\lkg of order $p$, and never returns a length shorter than the shortest $(k,w)$-colored $(S,T)$-\lkg of order $p$.

\medskip\noindent\emph{Proof of time complexity of the algorithm.}
Next we prove the time complexity of the algorithm.
The evaluation of the polynomial is done using dynamic programming.
This is a standard application of dynamic programming over walks while keeping track of the set of labels used so far, the weight of the labeled vertices, and the set of ending vertices used.
We prove that it can be performed in time $2^{p+k} n^{\Oh(1)}w$.

\begin{lemma}
\label{lem:randtime}
Let $S,T$ be disjoint subsets of $V(G)$ of size $|S|=|T|=p$, $c : V(G) \rightarrow [n]$ a coloring of $G$, $\we : V(G) \rightarrow \mathbb{Z}_{\ge 1}$ a weight function, $\ell \le n$ an integer, $k, w$ integers, and $q = 2^{3 + \lceil \log_2 n \rceil}$.
The polynomial $f(\fcf)$ can be evaluated at a random point over GF($q$) in time $2^{p+k} n^{\Oh(1)} w$.
\end{lemma}
\begin{proof}
We associate random values over GF($q$) to all variables $f_v(v)$, $f_e(uv)$, and $f_c(x,y)$, and from now denote by $f_v(v)$, $f_e(uv)$, and $f_c(x,y)$ these associated values, and by extension for a \wlkg $\fw$ denote by $f(\fw)$ the value associated to the monomial $f(\fw)$ and for a family of \wlkgs $\mathcal{F}$ denote by $f(\mathcal{F})$ the value associated to the polynomial $f(\mathcal{F})$.
Now, the task is to compute $f(\fcf)$.

Informally, we will compute $f(\fcf)$ by dynamic programming over partial \wlkgs, growing the \wlkgs one labeled walk at a time in the order specified by $\ordv(S)$.

Denote $\ordv(S) = (s_1, s_2, \ldots, s_p)$ and for any $t \in [p]$ denote by $\texttt{pre}_t(S)$ the length-$t$ prefix of $\ordv(S)$.
For every integer $t \in [p]$, integer $l \in [\ell]$, set $L \subseteq [k]$ of labels, set $T' \subseteq T$ of ending vertices, weight $w' \in [0,w]$, vertices $x,y \in V(G)$, and integer $o \in \{0,1\}$, we define
\[D(t, l, L, T', w', x, y, o) = f(\mathcal{F}(t, l, L, T', w', x, y, o)),\]
where we define $\mathcal{F}(t, l, L, T', w', x, y, o)$ to be the family of labeled \wlkgs $\fw = (W^1, \ldots, W_t)$, where for each $i \in [t]$, $W^i = ((v^i_1, \ldots, v^i_{\ell_i}), (r^i_1, \ldots, r^i_{\ell_i}))$,
that satisfy the following properties:
\begin{enumerate}
\item Each labeled walk $W^i$ in $\fw$ has length at least $2$ and does not contain labeled digons,
\item $\fw$ has order $t$ and ordered tuple of starting vertices $\svs(\fw) = \texttt{pre}_t(S)$,
\item $\fw$ has length $l$,
\item $\fw$ is injective and the set of used labels is $L$,
\item the set of ending vertices of all but the last walk in $\fw$ is $\ev((W^1, \ldots, W^{t-1})) = T'$,
\item $\fw$ has weight $\we(\fw) = w'$,
\item the last vertex of the last walk in $\fw$ is $v^t_{\ell_t} = x$,
\item the second last vertex of the last walk in $\fw$ is $v^t_{\ell_t-1} = y$, and
\item if $o=0$, then $r_{t,\ell_t} = 0$, otherwise $r_{t,\ell_t} \neq 0$.
\end{enumerate}

In other words, $t$ specifies the number of walks, $l$ specifies the length, $L$ specifies the used labels, $T'$ specifies the used ending vertices, $w'$ specifies the weight, $x$ specifies the last vertex of the last walk, $y$ specifies the second last vertex of the last walk, and $o$ specifies whether the last vertex of the last walk is labeled.
Note that it can be without loss of generality assumed that each walk has length at least $2$ because $S$ and $T$ are disjoint.

Then, we define also a shorthand that for $t \in [p]$, $l \in [\ell]$, $L \subseteq [k]$, $T' \subseteq T$, and $w' \in [0,w]$,
\[
D(t, l, L, T', w') = \sum_{x \in T'} \sum_{y \in N(x)} \sum_{o \in \{0,1\}} D(t, l, L, T' \setminus \{x\}, w', x, y, o),
\]
which intuitively denotes the polynomial corresponding to a ``completed'' \wlkg of $t$ walks with length $l$, used labels $L$, used ending vertices $T'$, and weight $w'$.

Now it holds that
\[
f(\fcf) = D(p, \ell, [k], T, w),
\]
and therefore computing $f(\fcf)$ can be done by computing all of the values $D(t, l, L, T', w', x, y, o)$ by dynamic programming.

Next we specify this computation by dynamic programming.
All values that we do not specify here will be set to zero.
First, to initialize, we define a special value $D(0, 0, \emptyset, \emptyset, 0) = 1$ corresponding to a family of \wlkgs containing one empty \wlkg.

Next, we describe computing the states where $o=0$, i.e., the last vertex is not labeled, for all $t \in [p]$, $l \in [\ell]$, $L \subseteq [k]$, $T' \subseteq T$, $w' \in [0,w]$, $x \in V(G)$, and $y \in N(x)$, assuming that all the states with smaller $l$ have already been computed.
There are four cases, corresponding to the four lines of \Cref{eq:dp1}.
In the first case the walk $W_t$ has length at least three, its second last vertex $y$ is not labeled, and we are extending the walkage by adding one not labeled vertex $x$ to $W_t$.
Second case is the same, but the second last vertex $y$ is labeled and thus we have to ensure to not create a labeled digon.
Third case is the case that we are extending the walkage by adding one more labeled walk, consisting of two vertices $y,x$, where $y = s_t$, neither of them labeled.
Fourth case is like the third, but the first vertex $y = s_t$ of the new walk is labeled.
Recall the notation that $[y = s_t] = 1$ if $y = s_t$ holds, and $0$ otherwise.

\begin{equation}
\label{eq:dp1}
\begin{split}
&D(t, l, L, T', w', x, y, 0) = f_e(xy)\\
&\cdot \left( \sum_{z \in V(G)} D(t, l-1, L, T', w', y, z, 0)\right.\\
&+ \sum_{z \in V(G) \setminus \{x\}} D(t, l-1, L, T', w', y, z, 1)\\
&+ [y = s_t] \cdot D(t-1, l-2, L, T', w')\\
&\left. + [y = s_t] \cdot \sum_{r \in L} f_v(y) \cdot f_c(c(y), r) \cdot D(t-1, l-2, L \setminus \{r\}, T', w'-\we(y)) \right).
\end{split}
\end{equation}

Then, we describe computing the states where $o=1$, i.e., the last vertex is labeled, for all $t \in [p]$, $l \in [\ell]$, $L \subseteq [k]$, $T' \subseteq T$, $w' \in [0,w]$, $x \in V(G)$, and $y \in N(x)$, assuming that all of the states with smaller $l$ have already been computed.
There are again four cases, analogously to \Cref{eq:dp1}.

\begin{equation}
\label{eq:dp2}
\begin{split}
&D(t, l, L, T', w', x, y, 1) = \sum_{r \in L} f_v(x) \cdot f_c(c(x), r) \cdot f_e(xy)\\
&\cdot \left( \sum_{z \in V(G)} D(t, l-1, L \setminus \{r\}, T', w'-\we(x), y, z, 0) \right.\\
&+ \sum_{z \in V(G) \setminus \{x\}} D(t, l-1, L \setminus \{r\}, T', w'-\we(x), y, z, 1)\\
&+ [y = s_t] \cdot D(t-1, l-2, L \setminus \{r\}, T', w'-\we(x))\\
&\left. + [y = s_t] \cdot \sum_{r' \in L \setminus \{r\}} f_v(y) \cdot f_c(c(y), r') \cdot D(t-1, l-2, L \setminus \{r, r'\}, T', w'-\we(x)-\we(y)) \right).
\end{split}
\end{equation}

This completes the description of the dynamic programming, which shows that each of the states $D(t, l, L, T', w', x, y, o)$ can be computed in $n^{\Oh(1)}$ time given the values of the states with smaller~$l$.
As there are $p \cdot \ell \cdot 2^k \cdot 2^p \cdot (w+1) \cdot n \cdot n \cdot 2 = \Oh(p 2^{p+k} n^3 w)$ states, the algorithm works in time $2^{p+k} n^{\Oh(1)} w$.
\end{proof}

As the algorithm can be implemented by $\Oh(n)$ applications of \Cref{lem:randtime}, the algorithm has time complexity $2^{p+k} n^{\Oh(1)} w$.
Recovering the solution can be done by a factor of $\Oh(n^2)$ more applications.

\subsection{Correctness\label{subsec:randcorr}}
To prove the correctness of the algorithm, we show that

\begin{enumerate}
\item[(a)] the polynomial $f(\fcf)$ 
is non-zero if $G$ contains a 
$(k,w)$-colored $(S,T)$-\lkg of order $p$ and length $\ell$ and
\item[(b)] the polynomial $f(\fcf)$ is the identically zero polynomial if the graph does not contain a $(k,w)$-colored $(S,T)$-\lkg of order $p$ and length $\le \ell$. 
\end{enumerate}

Because $f(\fcf)$ has degree $\ell-p+2k \le 3n \le q/2$, it follows from \Cref{lem:schwartzzippel} and (a) that if $G$ contains a $(k,w)$-colored $(S,T)$-\lkg of order $p$ and length $\ell$, then evaluating $f(\fcf)$ at a random point of GF($q$) has probability at least $1/2$ to be non-zero.
From (b) it follows that if $G$ does not contain a $(k,w)$-colored $(S,T)$-\lkg of order $p$ and length $\le \ell$, then evaluating $f(\fcf)$ at a random point is guaranteed to be zero.
This establishes that the algorithm is correct with probability at least $1/2$, with one-sided error.

The part (a) is relatively easy to prove (\Cref{lem:cornonzero}).
To prove (b), we first show  that the monomials in $f(\fcf)$ corresponding to non-proper labeled \wlkgs cancel out (\Cref{lem:properwalkage}).
This argument is based on the now-standard technique of bijective labeling based cancellation introduced in~\cite{DBLP:journals/siamcomp/Bjorklund14}.
The remaining part of the proof of (b) is much more complicated and is the main technical challenge.
It is based on the technical \Cref{lem:main_phi_exists}, whose proof is postponed to \Cref{sec:proofinvo}.

We start with (a).

\begin{lemma}
\label{lem:cornonzero}
If $G$ has a $(k,w)$-colored $(S,T)$-\lkg of order $p$ and length $\ell$, then $f(\fcf)$ is non-zero.
\end{lemma}
\begin{proof}
Consider a $(k,w)$-colored $(S,T)$-\lkg $\ps$ of order $p$ and length $\ell$.
Let $X \subseteq V(\ps)$ be the set of vertices with $|X| = k$, different colors, and weight $\we(X) = w$.
We can turn $\ps$ into a proper labeled \lkg $\fw$ of order $p$, length $\ell$, weight $w$, where $\svs(\fw) = \ordv(S)$ and $\ev(\fw) = T$, by ordering the paths based on their starting vertices and assigning the labels $[k]$ arbitrarily to the vertices $X$ when $\fw$ intersects $X$.

Therefore $\fw \in \fcf$, so it remains to prove that $\fw$ is the only labeled \wlkg in $\fcf$ that corresponds to the monomial $f(\fw)$, which then implies that the monomial $f(\fw)$ occurs in the polynomial $f(\fcf)$ with coefficient $1$, implying that $f(\fcf)$ is non-zero.

Notice that from $f(\fw)$, from the edge variables $f_e$ we can recover the edges $E(\fw)$ of $\fw$, from the vertex variables $f_v$ we can recover the labeled vertices $X$, and because vertices in $X$ have different colors, from the color-label pair variables $f_c$ we can recover how the labels correspond to the labeled vertices.
Therefore as the ordering of the paths is fixed by $\ordv(S)$ and every vertex appears in $\fw$ at most once, we have that $\fw$ is the unique element of $\fcf$ that corresponds to the monomial $f(\fw)$.
\end{proof}

Then, we deal with non-proper \wlkgs in $\fcf$.
Let $\fcfs \subseteq \fcf$ denote the family of proper labeled \wlkgs in $\fcf$, i.e., the labeled \wlkgs in $\fcf$ where all labeled indices have vertices of different colors.

\begin{lemma}\label{lem:properwalkage}
It holds that $f(\fcfs) = f(\fcf)$.
\end{lemma}
\begin{proof}
We will show that there is a function $\phi : \fcf \setminus \fcfs \rightarrow \fcf \setminus \fcfs$ that is an $f$-invariant fixed-point-free involution, i.e., for all $\fw \in \fcf \setminus \fcfs$ it holds that (1) $f(\phi(\fw)) = f(\fw)$, (2) $\phi(\fw) \neq \fw$, and (3) $\phi(\phi(\fw)) = \fw$.
This implies that the set $\fcf \setminus \fcfs$ can be partitioned into pairs $\{\fw, \phi(\fw)\}$ with $f(\fw) = f(\phi(\fw))$, and therefore every monomial corresponding to a labeled \wlkg in $\fcf \setminus \fcfs$ occurs in $f(\fcf)$ an even number of times, and therefore they cancel out because $f$ is over a field of characteristic 2.

The function $\phi$ is defined as follows.
Let $\fw = (W^1, \ldots, W^p)$ be a labeled \wlkg in $\fcf \setminus \fcfs$, where $W^i = ((v^i_1, \ldots, v^i_{\ell_i}), (r^i_1, \ldots, r^i_{\ell_i}))$.
Because $\fw$ is semiproper but not proper, there exists two different labeled indices that have a vertex of the same color, i.e., pairs $(i,a)$ and $(j,b)$ with $i,j \in [p]$, $a \in [\ell_i]$, $b \in [\ell_j]$, $(i,a) \neq (j,b)$, $c(v^i_a) = c(v^j_b)$, $r^i_a \neq 0$, and $r^j_b \neq 0$.
Let $(i,a), (j,b)$ be the lexicographically smallest such pair.
We set  $\phi(\fw)$ to be the labeled walkage obtained from $\fw$ after swapping $r^i_a$ with $r^j_b$.

First, we observe that $\phi(\fw) \in \fcf$.
Indeed, it cannot make a bijective \wlkg into non-bijective, and as it does not change the sequence of vertices of $\fw$ or which indices are labeled, it cannot make a semiproper walk into non-semiproper, or change the order, the length, the weight, the tuple of starting vertices, or the set of ending vertices.
Also $\phi(\fw)$ is not proper, i.e., $\phi(\fw) \in \fcf \setminus \fcfs$, because the vertices $v^i_a$ and $v^j_b$ are still labeled and have the same color.

To see why $f(\phi(\fw)) = f(\fw)$, note that, since $\phi$ does not change the vertices, it also does not change the edge variables $f_e$ of the monomial, it does not change which vertices are labeled so it does not change the vertex variables $f_v$ of the monomial, and because the vertices $v^i_a$ and $v^j_b$ have the same color the color-label-pair variables $f_c$ of the monomial are also not changed.

Also, we have that $\phi(\fw) \neq \fw$, since the fact that $\fw$ is bijective implies that $r^i_a \neq r^j_b$.
Also, $\phi(\phi(\fw)) = \fw$ because the swapping does not change which indices are labeled, and therefore does not change the lexicographically smallest pair of labeled indices with the same colors.
\end{proof}

As a result of \Cref{lem:properwalkage},
we can work with $f(\fcfs)$ instead of $f(\fcf)$.

\medskip

The most complicated part of the correctness proof will be to show part (b), that is, if there is no $(k,w)$-colored $(S,T)$-\lkg of order $p$ and length at most $\ell$, then $f(\fcfs)$ (and, thus by \Cref{lem:properwalkage}, $f(\fcf)$) is an identically zero polynomial.
Most of this proof will be presented in \Cref{sec:proofinvo}, but we introduce here the statement the lemma that we will prove in \Cref{sec:proofinvo}.
For this, we define \emph{barren labeled \wlkgs}.

\begin{definition}[Barren labeled \wlkg]
\label{def:barren}
A labeled \wlkg $\fw$ of length $\ell$ is \emph{barren} if there exists no labeled \lkg $\fw'$ with starting vertices $\svs(\fw') = \svs(\fw)$, set of ending vertices $\ev(\fw') = \ev(\fw)$, set of collected vertices $\col(\fw') = \col(\fw)$, length $\le \ell$ and edges $E(\fw') \subseteq E(\fw)$.
\end{definition}

In other words, a labeled \wlkg $\fw$ of length $\ell$ is barren if its edges form a subgraph of $G$ where no labeled {\sl \lkg} $\fw'$ of length at most $\ell$ can have the same sets of starting vertices, ending vertices, and collected vertices as $\fw$.
Intuitively, this means that the labeled walkage $\fw$ can not be ``untangled'' to give a corresponding labeled \lkg.
In particular, observe that because the ``untangling'' preserves the set of collected vertices, i.e., $\col(\fw') = \col(\fw)$, if no $(k,w)$-colored $(S,T)$-\lkgs of order $p$ and length at most $\ell$ exists, then all labeled \wlkgs in $\fcfs$ are barren.

Next, we state the main technical lemma for establishing the correctness of our algorithm.
\Cref{sec:proofinvo} is devoted to its proof.

\begin{lemma}
\label{lem:main_phi_exists}
Let $G$ be a graph and let $\bls$ the set of all proper barren labeled \wlkgs  in $G$.
There exists a function $\phi : \bls \rightarrow \bls$ so that for all $\fw \in \bls$, the function $\phi$ satisfies that
\begin{enumerate}
\item $\phi(\phi(\fw)) = \fw$ ($\phi$ is involution),
\item $\phi(\fw) \neq \fw$ ($\phi$ is fixed-point-free),
\item $f(\phi(\fw)) = f(\fw)$ ($\phi$ preserves the monomial),
\item $\ev(\phi(\fw)) = \ev(\fw)$ ($\phi$ preserves the set of ending vertices), and
\item $\svs(\phi(\fw)) = \svs(\fw)$ ($\phi$ preserves the ordered tuple of starting vertices).
\end{enumerate}
\end{lemma}

The main reason for defining the function $\phi$ for all proper barren labeled \wlkgs instead of just barren \wlkgs in $\fcfs$ is that $\phi$ will be defined recursively, and in the recursion we will anyway need to handle all proper barren labeled \wlkgs.

Now, the proof of (b) is an easy consequence of \Cref{lem:main_phi_exists}.

\begin{lemma}
\label{lem:corzero}
If $G$ has no $(k,w)$-colored $(S,T)$-\lkg of order $p$ and length $\le \ell$, then $f(\fcfs)$ is an identically zero polynomial.
\end{lemma}
\begin{proof}
First, because $G$ has no $(k,w)$-colored $(S,T)$-\lkg of order $p$ and length $\le \ell$, all labeled \wlkgs in $\fcfs$ are barren, i.e., $\fcfs \subseteq \bls$.

We show that if $\fw \in \fcfs$, then $\phi(\fw) \in \fcfs$.
By definition, $\phi(\fw)$ is proper.
By (3), $\phi$ preserves the set of labeled vertices and moreover because the labeled vertices have different colors it preserves also the label-vertex mapping, and therefore $\phi(\fw)$ is bijective and has weight $w$.
By (4) and (5), $\phi$ preserves the set of ending vertices and the ordered tuple of starting vertices.
By (3), $\phi$ also preserves the length $\ell$, as the order of $\fw$ is preserved by (5).
Therefore the restriction $\phi\restriction_{\fcfs}$ is a function $\phi\restriction_{\fcfs} : \fcfs \rightarrow \fcfs$.

Then, by (1-3), $\phi\restriction_{\fcfs}$ is an $f$-invariant fixed-point-free involution on $\fcfs$, implying that the set $\fcfs$ can be partitioned into pairs $\{\fw, \phi(\fw)\}$ with $f(\fw) = f(\phi(\fw))$, and therefore for every monomial $f(\fw)$, there is an even number of labeled \wlkgs $\fw \in \fcfs$ corresponding to it, and therefore because $f(\fcfs)$ is over a field of characteristic 2, it is identically zero.
\end{proof}

\subsection{Proof of \autoref{lem:main_phi_exists}\label{sec:proofinvo}}
In this subsection we prove \autoref{lem:main_phi_exists} by explicitly defining the function $\phi$ and then showing that it has all of the required properties.

In order to define $\phi$ we first introduce some notation for manipulating labeled walks and labeled \wlkgs.
Let $W = ((v_1, \ldots, v_\ell), (r_1, \ldots, r_\ell))$ be a labeled walk.
For indices $a,b$ with $1 \le a \le b \le \ell$, we denote by $W[a,b]$ the labeled subwalk between $a$ and $b$, inclusive, i.e., the labeled walk $W[a,b] = ((v_a, \ldots, v_b), (r_a, \ldots, r_b))$.
If $a>b$, then $W[a,b]$ denotes an empty labeled walk.

The involution $\phi$ will use three types of operations: reversing a subwalk, swapping a label from one occurrence of a vertex to another occurrence of it (possibly in a different walk), and swapping suffixes of two walks.

The subwalk reversal operation is defined as follows.
Let $W$ be a labeled walk of length $\ell$ and $a,b$ indices with $1 \le a \le b \le \ell$.
The walk obtained from $W$ by reversing the subwalk between $a$ and $b$, inclusive, including the labels, is denoted by $W\overleftarrow{[a,b]}$.
For example, if $W = ((v_1,v_2,v_3,v_4), (0,1,0,2))$, then $W\overleftarrow{[2,3]} = ((v_1,v_3,v_2,v_4), (0,0,1,2))$.
A labeled walk $W$ is a \emph{palindrome} if $W = W\overleftarrow{[1,\ell]}$ holds, i.e., the labeled walk is the same in reverse.
Note that $W\overleftarrow{[a,b]}=W$ holds if and only if $W[a,b]$ is a palindrome and that a subwalk $W[a,b]$ can be a palindrome only if its length is odd or if it is the empty walk.
We will use the following lemma about palindromic subwalks of labeled walks, and in particular the reason to forbid labeled digons is to make this lemma true.
Recall that any labeled walk in a proper labeled \wlkg is injective and does not contain labeled digons.
Recall also that $\col(W[a+1,b-1]) = \emptyset$ if and only if $W$ has no labels in the subwalk $W[a+1,b-1]$.

\begin{lemma}
\label{lem:phi_palindrome}
Let $W = ((v_1, \ldots, v_\ell), (r_1, \ldots, r_\ell))$ be an injective labeled walk of length $\ell$ that does not contain labeled digons, and let $a,b \in [\ell]$.
If $v_a = v_b$ and $W[a+1,b-1]$ is a palindrome, then $\col(W[a+1,b-1]) = \emptyset$.
\end{lemma}
\begin{proof}
First, because $W[a+1,b-1]$ is injective and palindrome, the only vertex of $W[a+1,b-1]$ that can be labeled is the middle vertex.
However, a label cannot occur at the middle vertex of a palindrome with more than one vertex because it would be a labeled digon.
If $W[a+1,b-1]$ has exactly one vertex, then again this vertex cannot be labeled because $v_a = v_b$ and $W$  does not contain labeled digons.
\end{proof}

The label swap operation is defined as follows.
Let $\fw = (W^1, \ldots, W^p)$ be a labeled \wlkg of order $p$, where for each $i \in [p]$ the \wlkg $W^i$ is denoted by $((v^i_1, \ldots, v^i_{\ell_i}), (r^i_1, \ldots, r^i_{\ell_i}))$.
Let $(i,a)$, $(j,b)$ be pairs with $i,j \in [p]$, $a \in [\ell_i]$, $b \in [\ell_j]$, $v^i_a = v^j_b$, and exactly one of $r^i_a$ and $r^j_b$ equal to zero (i.e. one of them unlabeled and one labeled).
The labeled \wlkg obtained from $\fw$ by swapping $r^i_a$ with $r^j_b$ is denoted by $\fw \frown^{i,j}_{a,b}$.
Note that because $r^i_a \neq r^j_b$, it holds that $\fw \frown^{i,j}_{a,b} \neq \fw$. See~\autoref{fig:swaplabels} for an illustration of the label swap operation.

\begin{figure}[ht]
\centering
\includegraphics[width=15cm]{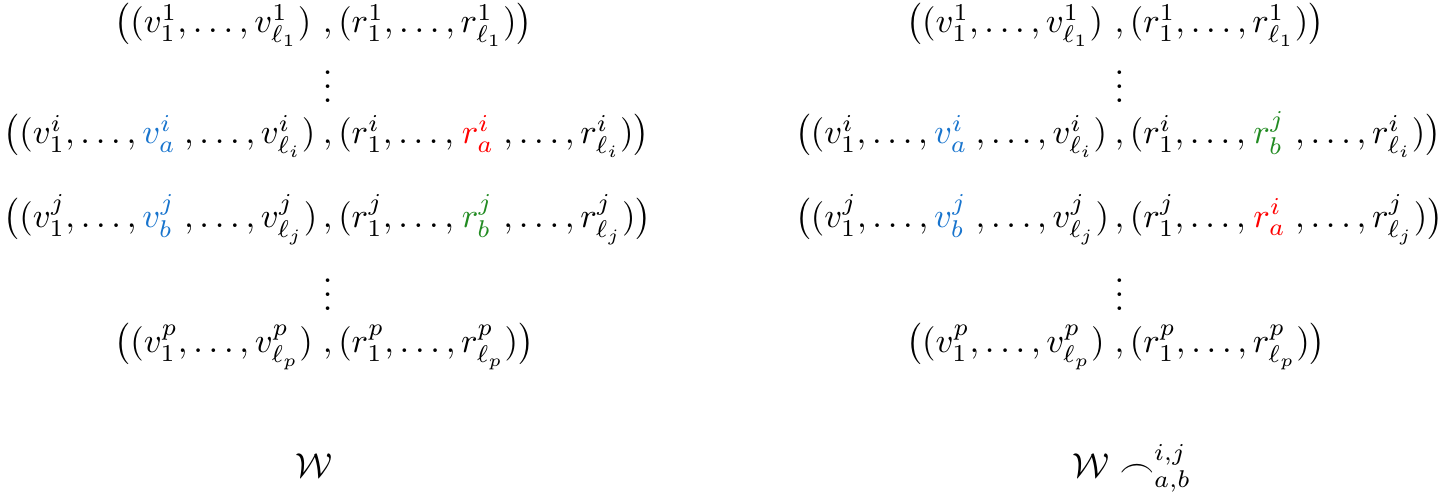}
\caption{An illustration of the label swap operation. On the left: a labeled \wlkg $\fw= (W^1, \ldots, W^p)$, and pairs
$(i,a)$, $(j,b)$ with $i,j \in [p]$, $a \in [\ell_i]$, $b \in [\ell_j]$, $v^i_a = v^j_b$, and exactly one of $r^i_a$ and $r^j_b$ equal to zero.
Note that we allow $i=j$.
On the right: the labeled \wlkg $\fw \frown^{i,j}_{a,b}$.}
\label{fig:swaplabels}
\end{figure}

\begin{figure}[ht]
\centering
\includegraphics[width=15cm]{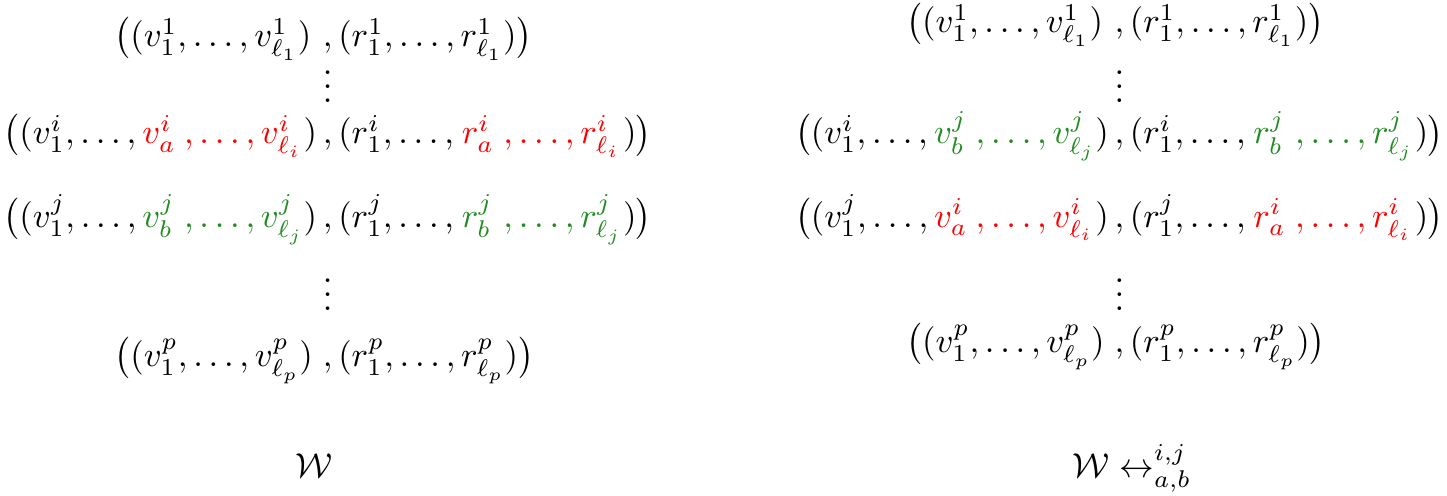}
\caption{An illustration of the suffix swap operation. On the left: a labeled \wlkg $\fw= (W^1, \ldots, W^p)$, and pairs
$(i,a)$, $(j,b)$ with $i,j \in [p]$, $a \in [\ell_i+1]$, $b \in [\ell_j+1]$, and $i \neq j$.
On the right: the labeled \wlkg $\fw \leftrightarrow^{i,j}_{a,b}$.}
\label{fig:suffixswap}
\end{figure}

The suffix swap operation is defined as follows.
Let $(i,a)$ and $(j,b)$ be pairs with $i,j \in [p]$, $a \in [\ell_i+1]$, $b \in [\ell_j+1]$, and $i \neq j$.
The labeled \wlkg obtained from $\fw$ by swapping the suffix of $W^i$ starting at index $a$ with the suffix of $W^j$ starting at index $b$ is denoted by $\fw \leftrightarrow^{i,j}_{a,b}$.
Note that here we allow that $a = \ell_i+1$ or $b = \ell_j+1$, with the interpretation that this corresponds to the empty suffix.
Clearly, if both $a = \ell_i+1$ and $b = \ell_j+1$, then this operation does not do anything, but otherwise if $\fw$ is a proper labeled \wlkg, applying this operation will in fact always result in a different \wlkg because of the different ending vertices condition.
See~\autoref{fig:suffixswap} for an illustration of the suffix swap operation.

If $W^1$ and $W^2$ are labeled walks so that the last vertex of $W^1$ is adjacent to the first vertex of $W^2$, then $W^1 \concs W^2$ denotes the concatenation of $W^1$ and $W^2$.
If $\fw = (W^1, \ldots, W^p)$ is a labeled \wlkg and $W$ is a labeled walk, then $W \concm \fw$ denotes the labeled \wlkg $(W \concs W^1, \ldots, W^p)$ and $W \sqcup \fw$ denotes the labeled \wlkg $(W, W^1, \ldots, W^p)$.

Next we define the function $\phi$ of \autoref{lem:main_phi_exists}.
We will provide some intuition about $\phi$ right after the definition, and 
\Cref{fig:definition_of_phi_case_A,fig:definition_of_phi_case_B,fig:definition_of_phi_case_C1a,fig:definition_of_phi_case_C1b,fig:definition_of_phi_case_C1c,fig:definition_of_phi_case_C2_C3,fig:definition_of_phi_case_C4,fig:definition_of_phi_case_CXDX,fig:definition_of_phi_case_D} demonstrate different cases of it.
The definition of $\phi$ will be recursive, using induction by the length of the \wlkg.

\begin{definition}[The function $\phi$]
\label{def:invophi}
Let $\fw = (W^1, \ldots, W^p)$ be a proper barren labeled \wlkg of order $p$.
For each $i \in [p]$, denote $W^i = ((v^i_1, \ldots, v^i_{\ell_i}), (r^i_1, \ldots, r^i_{\ell_i}))$.
The value $\phi(\fw)$ is defined, in some cases recursively, by selecting the first matching case from the following list:
\begin{enumerate}
\item[A.] if the vertex $v^1_1$ occurs only once in $\fw$:
\begin{enumerate}
\item[1.] if $\ell_1 \ge 2$, then $\phi(\fw) = W^1[1,1] \concm \phi(W^1[2,\ell_1], W^2, \ldots, W^p)$.
\item[2.] otherwise (i.e., $\ell_1 = 1$), $\phi(\fw) = W^1 \sqcup \phi(W^2, \ldots, W^p)$.
\end{enumerate}
\item[B.] if the vertex $v^1_1$ occurs in at least three different walks $W^i$:\\
There must be at least two different walks $W^i$ that contain $v^1_1$ but do not contain it as labeled.
Let $i,j$ be the two smallest indices so that both $W^i$ and $W^j$ contain $v^1_1$ but do not contain it as labeled.
Let $a$ be the index of the first occurrence of $v^1_1$ in $W^i$ and $b$ be the index of the first occurrence of $v^1_1$ in $W^j$.
Now, $\phi(\fw) = \fw \leftrightarrow^{i,j}_{a,b}$.
\item[C.] if the vertex $v^1_1$ occurs only in the walk $W^1$:\\
By the case (A), the vertex $v^1_1$ occurs multiple times in $W^1$.
Let $b$ be the index of the last occurrence of $v^1_1$ in $W^1$ and $a$ be the index of the second last occurrence of $v^1_1$ in $W^1$.
Note that $a=1$ if $v^1_1$ occurs only twice in $W^1$, and note also that $1 \le a \le b-2$.
\begin{enumerate}
\item[1.] if $r^1_{1} = r^1_{b} = 0$:
\begin{enumerate}
\item[(a)] if $W^1[2,b-1]$ is not a palindrome, then $\phi(\fw) = (W^1\overleftarrow{[2,b-1]}, W^2, \ldots, W^p)$.
\item[(b)] otherwise, if $b < \ell_1$, then $\phi(\fw) = W^1[1,b] \concm \phi(W^1[b+1,\ell_1], W^2, \ldots, W^p)$.
\item[(c)] otherwise (i.e., $b = \ell_1$), $\phi(\fw) = W^1 \sqcup \phi(W^2, \ldots, W^p)$.
\end{enumerate}
\item[2.] if the index $b$ is not a digon in $W^1$, then $\phi(\fw) = \fw \frown^{1,1}_{1,b}$.\\
\emph{Note: If neither case (1) nor (2) applies, then $r^1_{1}\neq 0$.}
\item[3.] if $W^1[2,a-1]$ is not a palindrome, then $\phi(\fw) = (W^1\overleftarrow{[2,a-1]}, W^2, \ldots, W^p)$.\\
\emph{Note: If $a=1$, then $W^1[2,a-1]$ is the empty walk which is a palindrome.}
\item[4.] if $v^1_{a+1} = v^1_{b-1}$:
\begin{enumerate}
\item[(a)] if $W^1[a+1,b-1]$ is not a palindrome, then $\phi(\fw) = (W^1\overleftarrow{[a+1,b-1]}, W^2, \ldots, W^p)$.
\item[(b)] otherwise, $\phi(\fw) = W^1[1,b] \concm \phi(W^1[b+1,\ell_1], W^2, \ldots, W^p)$.\\
\emph{Note: Here $W^1[b+1,\ell_1]$ cannot be an empty walk because by case (C.2) $b$ is a digon in $W^1$.}
\end{enumerate}
\item[X.] otherwise, $\phi(\fw) = W^1[1,a] \concm \phi(W^1[a+1,\ell_1], W^2, \ldots, W^p)$.\\
\emph{Note: The case C.X will form a ``common case'' with the case D.X.}
\end{enumerate}
\item[D.] if the vertex $v^1_1$ occurs in exactly two different walks:\\
Let $i$ be the index of the another walk $W^i$ in which $v^1_1$ occurs and let $b$ be the index of the first occurrence of $v^1_1$ in $W^i$.
\begin{enumerate}
\item[1.] if $r^1_{1} = r^i_b = 0$, then $\phi(\fw) = \fw \leftrightarrow^{1,i}_{1,b}$.
\item[2.] if the index $b$ is not a digon in $W^i$, then $\phi(\fw) = \fw \frown^{1,i}_{1,b}$.\\
\emph{Note: If neither case (1) nor (2) applies, then $r^1_{1}\neq 0$.}
\item[3.] if $v^1_1$ occurs at least twice in $W^i$, then let $c$ be the index of its second occurrence and $\phi(\fw) = \fw \leftrightarrow^{1,i}_{2,c+1}$.\\
\emph{Note: It can happen that one of the suffixes in this case is empty. However, both of them cannot be empty at the same time because $W^1$ and $W^i$ have different ending vertices because $\fw$ is proper.}\\
\emph{Note: In the remaining cases, $v^1_1$ occurs exactly once in $W^i$, and this occurrence is a digon at index $b$.}\\
Now, let $a$ be the index of the last occurrence of $v^1_1$ in $W^1$ (if $v^1_1$ occurs only once in $W^1$, then $a=1$).
\item[4.] if $W^1[2,a-1]$ is not a palindrome, then $\phi(\fw) = (W^1\overleftarrow{[2,a-1]}, W^2, \ldots, W^p)$.\\
\emph{Note: If $a=1$, then $W^1[2,a-1]$ is the empty walk which is a palindrome.}
\item[5.] if $a = \ell_1$, then $\phi(\fw) = W^1 \sqcup \phi(W^2, \ldots, W^p)$.
\item[6.] if $v^1_{a+1} = v^i_{b+1}$, then $\phi(\fw) = \fw \leftrightarrow^{1,i}_{a+1,b+1}$.\\
\emph{Note: By case (5) it holds that $a < \ell_1$ and by case (2) it holds that $b < \ell_i$.}
\item[X.] otherwise, $\phi(\fw) = W^1[1,a] \concm \phi(W^1[a+1,\ell_1], W^2, \ldots, W^p)$.\\
\emph{Note: The case D.X will form a ``common case'' with the case C.X.}
\end{enumerate}
\end{enumerate}
\end{definition}

\begin{figure}[ht]
\centering
\includegraphics[width=16.5cm]{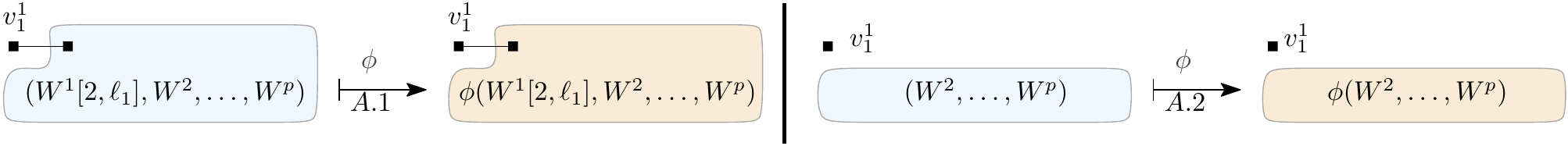}
\caption{Examples of cases A.1 and A.2 of the definition of $\phi$.
The vertex $v_1^1$ can be either labeled or unlabeled.}
\label{fig:definition_of_phi_case_A}
\end{figure}

\begin{figure}[ht]
\centering
\includegraphics[width=6.5cm]{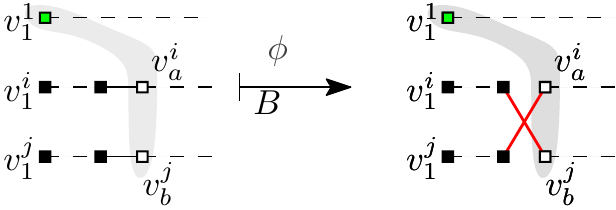}
\caption{Example of case B of the definition of $\phi$.
All vertices inside the grey bag are different occurrences of the same vertex $v^1_1$ of the graph. The white vertices $v_a^i$ and $v_b^j$ are unlabeled, the black vertices could be labeled or unlabeled, and the green vertex $v^1_1$ is labeled in this specific example.}
\label{fig:definition_of_phi_case_B}
\end{figure}

\begin{figure}[ht]
\centering
\includegraphics[width=7.5cm]{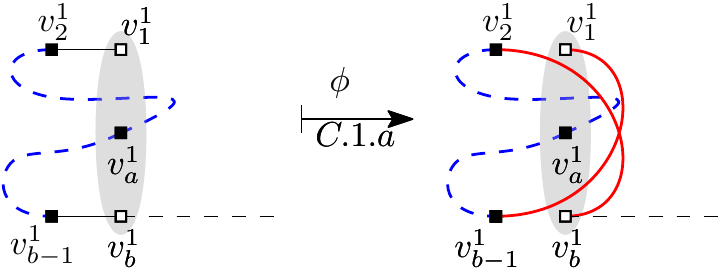}
\caption{Example of case C.1.a of the definition of $\phi$.
All vertices inside the grey bag are different occurrences of the same vertex $v^1_1$ of the graph.
The vertices $v^1_1$ and $v^1_b$ are unlabeled and the black vertices $v^1_2$, $v^1_a$, and $v^1_{b-1}$ can be either labeled or unlabeled.}
\label{fig:definition_of_phi_case_C1a}
\end{figure}

\begin{figure}[ht]
\centering
\includegraphics[width=12cm]{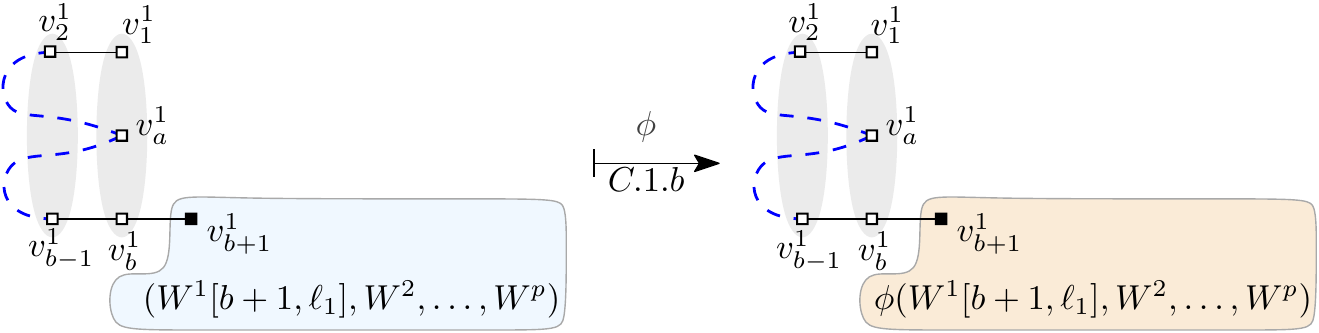}
\caption{Example of case C.1.b of the definition of $\phi$.
All vertices inside the grey bags are the same vertex of the graph.
By case C.1.a, the blue subwalk $W^1[2,b-1]$ is palindrome, and therefore by \Cref{lem:phi_palindrome}, the vertices in it are unlabeled.}
\label{fig:definition_of_phi_case_C1b}
\end{figure}

\begin{figure}[ht]
\centering
\includegraphics[width=9.5cm]{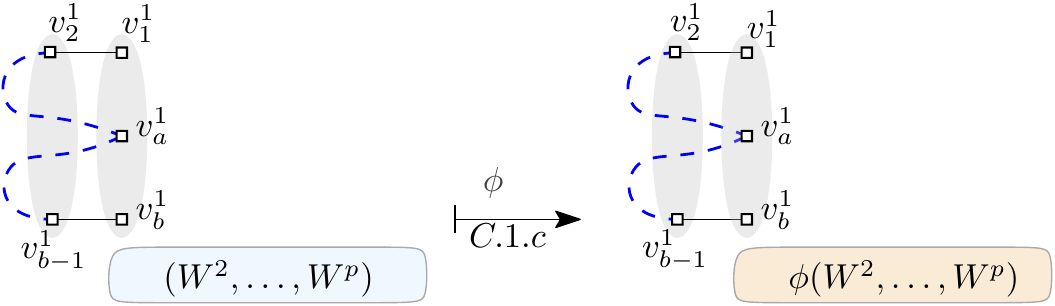}
\caption{Example of case C.1.c of the definition of $\phi$.
All vertices inside the grey bags correspond to the same vertex of the graph.
By case C.1.a, the blue subwalk $W^1[2,b-1]$ is palindrome, and therefore by \Cref{lem:phi_palindrome}, the vertices in it are unlabeled.}
\label{fig:definition_of_phi_case_C1c}
\end{figure}

\begin{figure}[ht]
\centering
\includegraphics[width=12.5cm]{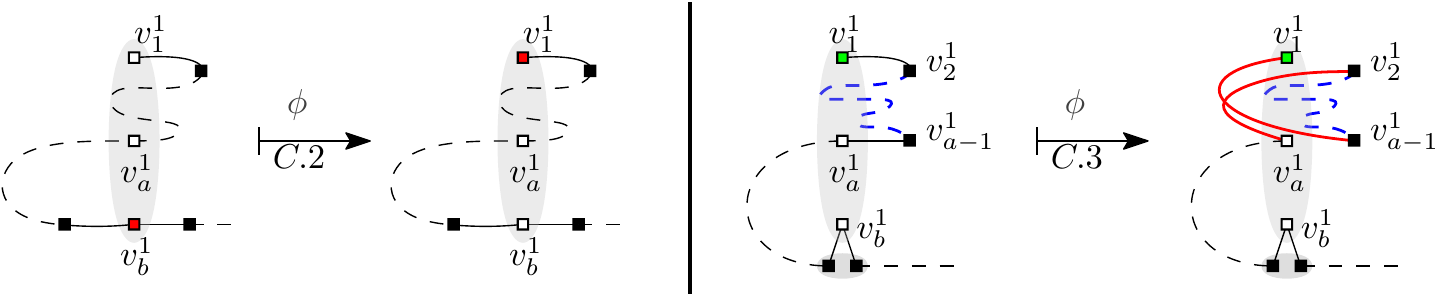}
\caption{Example of cases C.2 and C.3 of the definition of $\phi$.
Vertices inside the same grey bags correspond to the same vertex of the graph. On the left part of the figure (case C.2), in the initial configuration the vertex $v_1^1$ is unlabeled  and the vertex $v_b^1$ is labeled (with color red) and the application of $\phi$ in this case exchanges this label from $v_b^1$ to $v_1^1$.
On the right part of the figure, $v_b^1$ is a digon and therefore it is unlabeled and by cases C.1 and C.2, $v_1^1$ has to be labeled (depicted in green).}
\label{fig:definition_of_phi_case_C2_C3}
\end{figure}

\begin{figure}[ht]
\centering
\includegraphics[width=16cm]{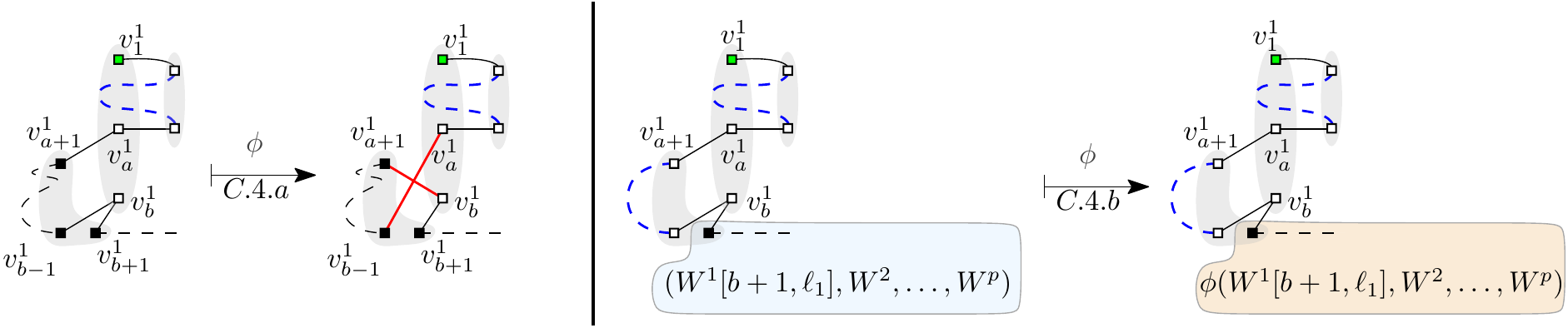}
\caption{Examples of cases C.4.a and C.4.b of the definition of $\phi$.
By case C.2, $b$ is a digon on $W^1$ and by case C.3, $W^1[2,a-1]$ is a palindrome.
For both case C.4.a and case C.4.b, we have that $v_{a+1}^1 = v_{b-1}^1$ ($v_{a+1}^1$, $v_{b-1}^1$, and $v_{b+1}^1$ are in the same grey bag). 
If $W^1[a+1,b-1]$ is not a palindrome, then we are in case C.4.a (on the left), while if $W^1[a+1,b-1]$ is a palindrome, we are in case C.4.b. (on the right).}
\label{fig:definition_of_phi_case_C4}
\end{figure}

\begin{figure}[ht]
\centering
\includegraphics[width=11.5cm]{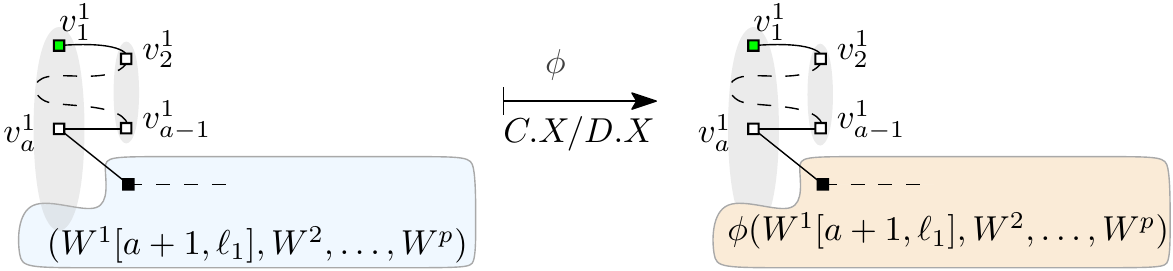}
\caption{Examples of cases C.X and D.X of the definition of $\phi$.}
\label{fig:definition_of_phi_case_CXDX}
\end{figure}

\begin{figure}[ht]
\centering
\includegraphics[width=16.5cm]{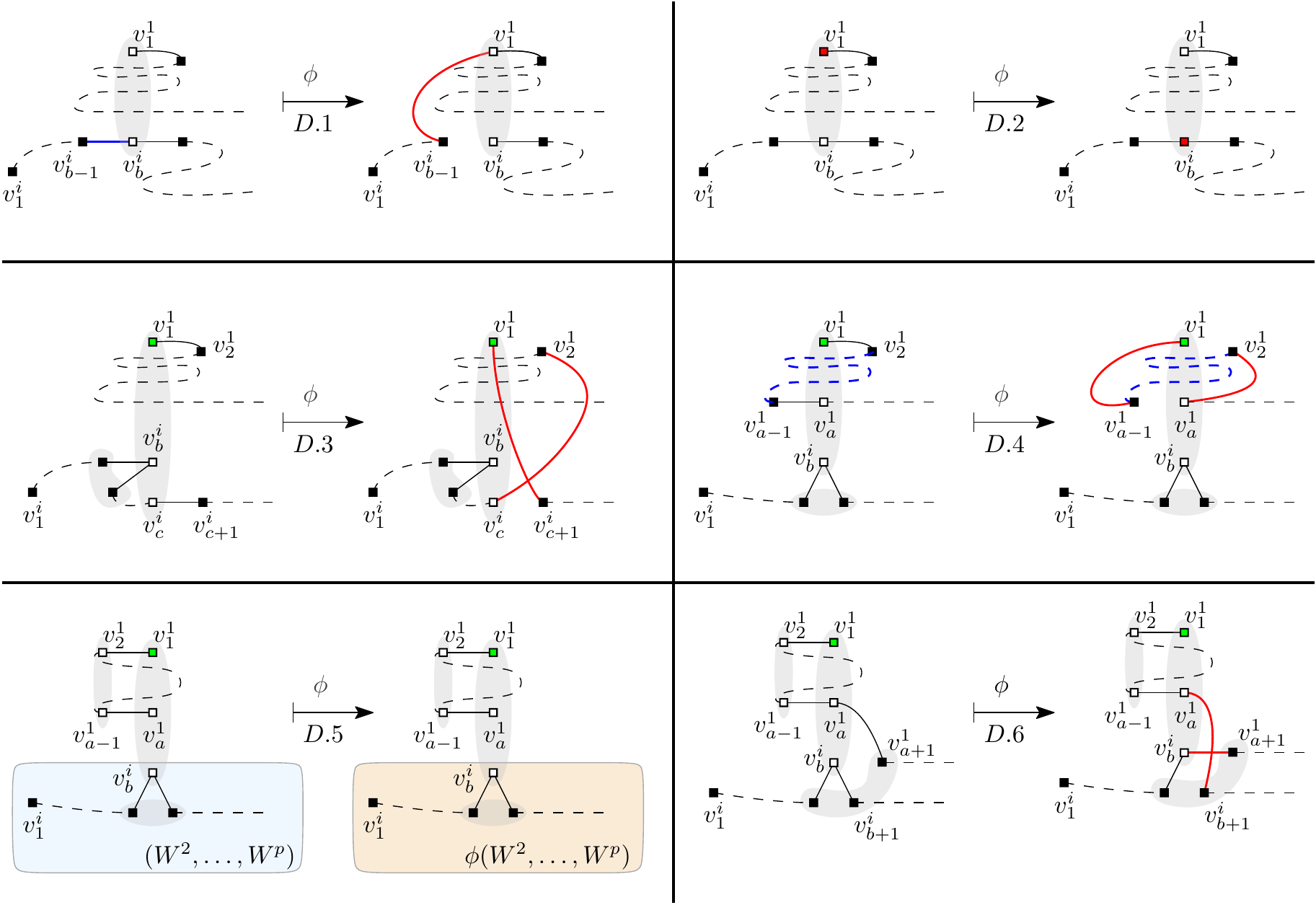}
\caption{Example of cases D.1-D.6 of the definition of $\phi$.
In each grey bag, all vertices inside the corresponding bag are the same vertex of the graph. Labeled vertices are depicted in red and green. White vertices correspond to unlabeled vertices and black vertices can be either labeled or unlabeled.}
\label{fig:definition_of_phi_case_D}
\end{figure}

\medskip
\noindent\emph{Intuition for $\phi$.}
Before laboriously proving that $\phi$ indeed is a function from proper barren labeled \wlkgs to proper barren labeled \wlkgs satisfying the required properties, let us give some rough outline of ideas behind it.
First, the general idea is that if the \wlkg $\fw$ goes to a certain case, then the \wlkg $\phi(\fw)$ goes again to the same case, which then maps it back to $\fw$.
The only exception is that the cases C.X and D.X could map to each other.

Then, let us consider the cases relevant for a single walk, i.e., the case A.1 and the cases under C.
Here, the intuition of case A.1 is to just move forward in the walk: we don't care much about what $\phi$ does to the rest of the walk because it must preserve the vertex right after $v^1_1$, and attaching $v^1_1$ to the front will not create a digon because $v^1_1$ occurs only at one index.
Then, case C.1.a is the standard loop reversal case, which is safe because neither index $1$ nor $b$ is labeled.
The case C.1.b (and C.1.c) corresponds to ignoring a palindromic subwalk, which can be safely done by \Cref{lem:phi_palindrome}.
Then, case C.2 is the standard label swap case, which is safe because the index $b$ is a not digon (note that the index $1$ is never a digon).
The cases C.1--C.2 are in some sense the ``easy cases'', while the cases C.3--C.X require more analysis of the remaining situation and quite unintuitive design.
First, if neither C.1 nor C.2 applies, we know that the index $1$ is labeled and the index $b$ is digon.
The purpose of case C.3 is to, in some sense reduce to a situation where we pretend that the vertex $v^1_1$ occurs only at indices $1$, $a$, and $b$, as the walk between $1$ and $a$ is an irrelevant palindromic loop.
Then, case C.4 handles a corner condition which would prevent case C.X from working.
The case C.X ignores the palindromic loop between $1$ and $a$, leaving the only occurrence of the vertex $v^1_1$ in the rest of the walk to be at the digon $b$, which in some sense makes it ``harmless'' in that the recursive calls will never need to analyse the vertex $v^1_b$ again as the first vertex.

The intuition for the case of multiple walks is as follows.
First, the case A.2 is just an analogue of A.1 when the first walk has length $1$.
Then, if the vertex $v^1_1$ occurs multiple times, we consider three different cases: $v^1_1$ occurs in at least three walks, $v^1_1$ occurs in one walk, and $v^1_1$ occurs in two walks.
Here, the three walks case B is quite easy, as we can just consider two of the walks where $v^1_1$ is not labeled, circumventing all issues with labeled digons.
When $v^1_1$ occurs in only one walk we go to the one walk case C.
Then, when $v^1_1$ occurs in two walks $W^1$ and $W^i$, the intuition of cases under D is that we concatenate $W^1$ with reversed $W^i$, with some special marker in between, and then apply the single walk cases under C for this concatenation.
Here, in the case D.X this can change whether $v^1_1$ occurs in two walks or a single walk, and therefore it is necessary to have the common case of C.X and D.X, moreover taking care in the proof that moving back from C.X to D.X will be handled correctly.

\medskip
\noindent\emph{Correctness proof for $\phi$.}
We will then proceed to first show that $\phi$ is well-defined, then that $\phi$ maps proper barren labeled \wlkgs to proper barren labeled \wlkgs, and then that $\phi$ satisfies all of the properties stated in \autoref{lem:main_phi_exists}, with $\phi(\phi(\fw)) = \fw$ being the most complicated of them to prove.
The proof is long because we have to analyze most of the 18 cases one by one. However, most of the arguments in these proofs are relatively easy once the definition of $\phi$ is set.
The main challenge in the proof was to come up with the right definition of $\phi$. 

\medskip
\noindent\emph{Well-definedness of  $\phi$.}
In \autoref{def:invophi}, in several cases, namely A.1, A.2, C.1.b, C.1.c, C.4.b, C.X, D.5, and D.X, the function  $\phi$  is defined recursively.
A priori it is not even clear why the syntactic value $\phi(\fw)$ is even well-defined in these cases.
It requires proof that in these cases the recursive argument is in the domain of $\phi$, in particular that it is also a proper barren labeled \wlkg.

Next we show that the syntactic value $\phi(\fw)$ for proper barren labeled \wlkgs $\fw$ is well-defined.
We remark that \Cref{lemma:well-definedphi} does not yet show that $\phi(\fw)$ is a proper barren labeled \wlkg; it will require more efforts to prove (see~\Cref{lem:rand_phi_proper,lem:presdigon,lem:nolabeleddigons}). 

\begin{lemma}\label{lemma:well-definedphi}
In case A.1 of \Cref{def:invophi} it holds that $(W^1[2,\ell_1], W^2, \ldots, W^p)$ is a proper barren labeled \wlkg, in cases A.2, C.1.c, and D.5 it holds that $(W^2, \ldots, W^p)$ is a proper barren labeled \wlkg, in cases C.1.b, and C.4.b it holds that $(W^1[b+1,\ell_1], W^2, \ldots, W^p)$ is a proper barren labeled \wlkg, and in cases C.X and D.X it holds that $(W^1[a+1, \ell_1], W^2, \ldots, W^p)$ is a proper barren labeled \wlkg.
\end{lemma}
\begin{proof}
In all cases, the labeled \wlkg used as the recursive argument is obtained from $\fw$ by removing either the walk $W^1$ or a prefix of $W^1$.
First we need to argue that the recursive argument is a labeled \wlkg.
For this, the only thing to argue is that (1) the recursive argument contains
at least one walk (i.e., $p\ge 2$ in cases A.2, C.1.c, and D.5) and that (2) all walks in the recursive argument are non-empty (i.e. $\ell_1 \ge 2$ in case A.1, $b < \ell_1$ in cases C.1.b and C.4.b, and $a<\ell_1$ in cases C.X and D.X).
The other properties of labeled \wlkgs are clearly satisfied when removing either $W^1$ or a prefix of $W^1$.

The above conditions are satisfied directly by definition in cases A.1 and C.1.b.
In case C.4.b, $b < \ell_1$ holds by the fact that (due to case C.2) the index $b$ is a digon in $W^1$.
In case C.X, recall that $a$ is the index of the second last occurrence of $v^1_1$ in $W^1$, so $a < \ell_1$.
In case D.X, we have that $a < \ell_1$ by case D.5.
For the remaining cases A.2, C.1.c, and D.5, observe the following.
If $p=1$ would hold, then $\fw=(W^1)$.
Then, since in all these three cases $v^1_{\ell_1} = v^1_1$ and $W^1[2,\ell_1]$ cannot contain labels (in A.2 trivially, in C.1.c by $r^1_1 = r^1_b = 0$ and \autoref{lem:phi_palindrome}, and in D.5 by cases D.1, D.2, and D.4 combined with \autoref{lem:phi_palindrome}),
it should hold that $(W^1[1,1])$ is a labeled \lkg consisting only of one walk with one vertex that would contradict the fact that $\fw$ is barren by \Cref{def:barren}.

It is clear by definition of a proper labeled \wlkg that removing a walk or a prefix of a walk maintains that the \wlkg is proper.
To complete the proof, it remains to show case by case that the labeled \wlkgs used as recursive arguments are barren.

In all of the cases the proofs will follow the same template: For the sake of contradiction we suppose that the labeled \wlkg $\fw'$ used as a recursive argument is not barren, then consider the labeled \lkg $\fw''$ that witnesses that $\fw'$ is not barren, and then use $\fw''$ to construct a labeled \lkg that shows that $\fw$ is not barren, obtaining a contradiction.
We spell out these steps in detail for the case A.1, and in less detail for subsequent cases.

\smallskip\noindent\textsl{Case A.1.}
For the sake of contradiction, suppose that $\fw' = (W^1[2,\ell_1], W^2, \ldots, W^p)$ is not barren.
Then by the definition of barren, there exists a labeled \lkg $\fw''$ with $\svs(\fw'') = \svs(\fw')$, $\ev(\fw'') = \ev(\fw')$, $\col(\fw'') = \col(\fw')$, length $\le \ell-1$, and edges $E(\fw'') \subseteq E(\fw')$.
By the assumptions of case A.1, the labeled \wlkg $\fw'$ does not contain $v^1_1$, so the labeled \lkg $\fw''$ cannot contain $v^1_1$ because the edge property $E(\fw'') \subseteq E(\fw')$ ensures that $v^1_1$ cannot occur in a walk of length more than one, and the start vertex property $\svs(\fw'') = \svs(\fw')$ ensures that $v^1_1$ cannot occur in a walk of length one.
By the start vertex property, it holds that the first vertex of the first walk in $\fw''$ is $v^1_{2}$.
Therefore, $W^1[1,1] \concm \fw''$ is a labeled \lkg.
Because $v^1_1 v^1_2 \in E(\fw)$ and $E(\fw'') \subseteq E(\fw') \subseteq E(\fw)$, we have that $E(W^1[1,1] \concm \fw'') \subseteq E(\fw)$.
Also, observe that because of $\col(\fw'') = \col(\fw')$, it holds that $\col(W^1[1,1] \concm \fw'') = \col(\fw)$.
Similarly, we observe that $\svs(W^1[1,1] \concm \fw'') = \svs(\fw)$, $\ev(W^1[1,1] \concm \fw'') = \ev(\fw)$, and the length of $\fw''$ is at most $\ell$.
Therefore, $W^1[1,1] \concm \fw''$ is a labeled \lkg that according to \Cref{def:barren} contradicts the fact that $\fw$ is barren.

\smallskip\noindent\textsl{Case A.2.}
Again, suppose that $\fw' = (W^2, \ldots, W^p)$ is not barren, and consider the witness $\fw''$.
Because $v^1_1$ does not occur in $\fw''$, it holds that $W^1 \sqcup \fw''$ is a labeled \lkg that contradicts that $\fw$ is barren.

\smallskip\noindent\textsl{Case C.1.b.}
Suppose that $\fw' = (W^1[b+1, \ell_1], W^2, \ldots, W^p)$ is not barren and consider the witness $\fw''$.
As, by definition of $b$ in case C, $v^1_1$ occurs in $\fw$ only in the subwalk $W^1[1,b]$, it cannot occur in $\fw''$.
Therefore $W^1[1,1] \concm \fw''$ is a labeled \lkg.
It is easy to observe that $\svs(W^1[1,1] \concm \fw'') = \svs(\fw)$, $\ev(W^1[1,1] \concm \fw'') = \ev(\fw)$, and $E(W^1[1,1] \concm \fw'') \subseteq E(\fw)$.
Also, \autoref{lem:phi_palindrome} implies that $\col(W^1[1,b]) = \emptyset$ and therefore $\col(W^1[1,1] \concm \fw'') = \col(\fw)$, therefore contradicting that $\fw$ is barren.

\smallskip\noindent\textsl{Case C.1.c.}
This case is similar as the previous, in particular, \autoref{lem:phi_palindrome} implies that $\col(W^1) = \emptyset$.
Therefore, if we assume that $\fw' = (W^2, \ldots, W^p)$ is not barren and we take the labeled \lkg $\fw''$ that witnesses that $\fw'$ is not barren, we can construct a labeled \lkg $W^1[1,1] \sqcup \fw''$ that contradicts the fact that $\fw$ is barren.

\smallskip\noindent\textsl{Case C.4.b.}
Assume that $\fw' = (W^1[b+1, \ell_1], W^2, \ldots, W^p)$ is not barren and take the labeled \lkg $\fw''$ that witnesses that $\fw'$ is not barren.
As $v^1_1$ occurs in $\fw$ only in the subwalk $W^1[1,b]$, it cannot occur in $\fw''$.
Therefore $W^1[1,1] \concm \fw''$ is a labeled \lkg.
Note that in this case $W^1[2,a-1]$ is a palindrome, the index $a$ of the walk $W^1$ is not labeled because the index $1$ is labeled and $\fw$ is proper, and $W^1[a+1,b-1]$ is a palindrome, and the index $b$ of $W^1$ is not labeled.
Therefore, by \autoref{lem:phi_palindrome}, $\col(W^1[1, b]) = \{v^1_1\}$, and therefore $\col(W^1[1,1] \concm \fw'') = \col(\fw)$, and therefore we contradict the fact that $\fw$ is barren.

\smallskip\noindent\textsl{Case C.X.}
In this case, the argument is less apparent because  $v^1_1$ indeed occurs in $W^1[a+1,\ell_1]$.
Again, we start by assuming that $\fw' = (W^1[a+1, \ell_1], W^2, \ldots, W^p)$ is not barren, and take the labeled \lkg $\fw''$ that witnesses that $\fw'$ is not barren.
Now, note that $v^1_1$ occurs in $\fw'$ only as a single digon in $W^1[a+1, \ell_1]$.
Therefore, $v^1_1$ cannot occur as a starting or ending vertex in $\fw''$.
Also, there is only one edge in $E(\fw')$ incident to $v^1_1$, which then prevents $v^1_1$ occuring at any position in $\fw''$, because any position containing $v^1_1$ would have to a digon, but $\fw''$ is labeled \lkg and thus does not contain digons.
Therefore, we construct a labeled \lkg $W^1[1,1] \concm \fw''$ and use the fact that $W^1[2,a-1]$ is palindrome with \autoref{lem:phi_palindrome} to conclude that $\col(W^1[1,1] \concm \fw'') = \col(\fw)$, and to finally observe that $W^1[1,1] \concm \fw''$ satisfies also all the other needed properties to contradict the fact that $\fw$ is barren.

\smallskip\noindent\textsl{Case D.5.}
Note that here we have that $r^1_{1} \neq 0$, $v^1_1 = v^1_{\ell_1}$, and $W^1[2,\ell_1-1]$ is a palindrome.
The arguments are similar to case C.X:
The vertex $v^1_1$ does occur in the \wlkg $(W^2, \ldots, W^p)$, but it occurs in it only a single time, which is a digon in the walk $W^i$.
So suppose that $\fw' = (W^2, \ldots, W^p)$ is not barren, and consider the witness $\fw''$.
By similar arguments as in case C.X, $v^1_1$ cannot occur in $\fw''$.
Therefore, we construct a labeled \lkg $W^1[1,1] \sqcup \fw''$ and use the fact that $W^1[2,a-1]$ is palindrome with \autoref{lem:phi_palindrome} to conclude that $\col(W^1[1,1] \sqcup \fw'') = \col(\fw)$, and to finally observe that $W^1[1,1] \sqcup \fw''$ satisfies also all the other needed properties to contradict the fact that $\fw$ is barren.

\smallskip\noindent\textsl{Case D.X.}
Note that here again, we have that $r^1_{1} \neq 0$ and $W^1[2,a-1]$ is a palindrome, where $a$ is the last occurrence of $v^1_1$ in $W^1$.
The arguments are similar to case C.X:
The vertex $v^1_1$ does occur in the \wlkg $(W^1[a+1,\ell_1], W^2, \ldots, W^p)$, but it occurs in it only a single time, which is a digon in the walk $W^i$.
So again suppose that $\fw' = (W^1[a+1, \ell_1], W^2, \ldots, W^p)$ is not barren, and consider the witness $\fw''$.
Again by arguments of C.X we have that $v^1_1$ cannot occur in $\fw''$.
Therefore, we again construct a labeled \lkg $W^1[1,1] \concm \fw''$ and use the fact that $W^1[2,a-1]$ is palindrome with \autoref{lem:phi_palindrome} to conclude that $\col(W^1[1,1] \concm \fw'') = \col(\fw)$, and to finally observe that $W^1[1,1] \concm \fw''$ satisfies also all the other needed properties to contradict that $\fw$ is barren.
\end{proof}

The next three lemmas establish that $\phi(\fw)$ is a proper barren labeled \wlkg.
In addition, \Cref{lem:rand_phi_proper} shows that $\phi$ satisfies the properties $f(\phi(\fw)) = f(\fw)$, $\ev(\phi(\fw)) = \ev(\fw)$, and $\svs(\phi(\fw)) = \svs(\fw)$.

\begin{lemma}
\label{lem:rand_phi_proper}
Let $\fw$ be a proper barren labeled \wlkg.
It holds that $\phi(\fw)$ is a labeled \wlkg, $\svs(\phi(\fw)) = \svs(\fw)$, $\ev(\phi(\fw)) = \ev(\fw)$, and $f(\phi(\fw)) = f(\fw)$.
\end{lemma}
\begin{proof}
We prove the lemma by induction on the length of the \wlkg $\fw$.
Here, all of the cases should be easy to verify, so the arguments we provide will be terse.

Cases A.1 works directly by induction, in particular, we can use the induction assumptions that 
\begin{itemize}
\item $\phi(W^1[2,\ell_1], W^2, \ldots, W^p)$ is a labeled \wlkg,
\item $\svs(\phi(W^1[2,\ell_1], W^2, \ldots, W^p)) = \svs((W^1[2,\ell], W^2, \ldots, W^p))$,
\item $\ev(\phi(W^1[2,\ell_1], W^2, \ldots, W^p)) = \ev((W^1[2,\ell], W^2, \ldots, W^p))$, and 
\item $f(\phi(W^1[2,\ell_1], W^2, \ldots, W^p)) = f((W^1[2,\ell], W^2, \ldots, W^p))$,
\end{itemize}
to prove the same properties for $\fw$.
In particular, we use the start vertex property to ensure that the first vertex of the first walk of $\phi(W^1[2,\ell_1], W^2, \ldots, W^p)$ is $v^1_2$, and therefore the first edge of the first walk of $W^1[1,1] \concm \phi(W^1[2,\ell_1], W^2, \ldots, W^p)$ is $v^1_1 v^1_2$.

Case A.2 works similarly to A.1.
Case B works by the property that $v^i_a = v^j_b$, in particular observing that if both of the suffixes are non-empty, the suffix swap operation preserves the set of ending vertices $\ev(\fw)$.
Case C.1.a works because $v^1_1 = v^1_b$, and cases C.1.b and C.1.c work by similar induction as A.1.
Case C.2 works because $v^1_1 = v^1_b$, in particular, even though the index of the label in the walk changes, the vertex variable or the label-color pair variable do not change because the vertex does not change.
Case C.3 works because $v^1_1 = v^1_a$.
Case C.4.a works because $v^1_a = v^1_b$ and case C.4.b works by similar induction as A.1.
Case C.X works again by induction.
Case D.1 works because $v^1_1 = v^i_b$.
Case D.2 works because $v^1_1 = v^i_b$, by the same argument as case C.2.

In case D.3, all other conditions work directly by $v^1_1 = v^i_c$, but we should pay attention to the ending vertices condition $\ev(\phi(\fw)) = \ev(\fw)$, because it can happen that one of the suffixes is empty.
As observed already in the definition, observe that at most one of the suffixes can be empty because $v^1_1 = v^i_c$ and $W^1$ and $W^i$ have different ending vertices because $\fw$ is proper.
First, if $\ell_1 = 1$, then the ending vertex of $W^i$ becomes $v^i_c = v^1_1 = v^1_{\ell_1}$, and the ending vertex of $W^1$ becomes $v^i_{\ell_i}$, so the condition holds.
Second, if $\ell_i = c$, then the ending vertex of $W^1$ becomes $v^1_1 = v^i_c = v^i_{\ell_i}$, and the ending vertex of $W^i$ becomes $v^1_{\ell_1}$, so the condition holds.

Case D.4 works because $v^1_1 = v^1_a$.
Case D.5 works by induction.
Case D.6 works because $v^1_{a+1} = v^i_{b+1}$ and the suffixes are guaranteed to be non-empty.
Case D.X works by induction.
\end{proof}

Note that because $f(\fw)$ and $\svs(\fw)$ determine $R(\fw)$, $E(\fw)$, and the length of $\fw$ uniquely, \Cref{lem:rand_phi_proper} implies that $\phi(\fw)$ is a barren \wlkg (because $\fw$ is barren).
It remains to prove that $\phi(\fw)$ is proper, and to prove that $\phi(\fw)$ is proper the only remaining thing to prove is that $\phi(\fw)$ does not contain labeled digons (\Cref{lem:nolabeleddigons}).
In particular, the property $\ev(\phi(\fw)) = \ev(\fw)$ guarantees that the ending vertices of $\phi(\fw)$ are distinct, and $f(\phi(\fw)) = f(\fw)$, which implies $\col(\phi(\fw)) = \col(\fw)$ guarantees that the labeled vertices of $\phi(\fw)$ have different colors (because $\fw$ is proper).

We will make use of the following lemma that follows directly from \autoref{lem:rand_phi_proper}.

\begin{lemma}
\label{lem:presdigon}
If a vertex occurs exactly once in $\fw$ and this occurrence is a digon, then this vertex also occurs exactly once in $\phi(\fw)$ and this occurrence is also a digon with the same adjacent vertices. 
\end{lemma}
\begin{proof}
Suppose that a vertex $v$ occurs exactly once in $\fw$ and this occurrence is a digon.
Therefore, it cannot be a starting or ending vertex in $\fw$.
By \autoref{lem:rand_phi_proper},
it holds that
$\svs(\phi(\fw)) = \svs(\fw)$ and $\ev(\phi(\fw)) = \ev(\fw)$ and therefore $v$ cannot be a starting or ending vertex neither in $\phi(\fw)$.
Then, since by \autoref{lem:rand_phi_proper}, we have that $f(\phi(\fw)) = f(\fw)$, implying $E(\phi(\fw)) = E(\fw)$, the vertex $v$ must have the exactly same adjacent vertices in $\phi(\fw)$ as in $\fw$.
\end{proof}

Then we prove that $\phi(\fw)$ has no labeled digons.

\begin{lemma}\label{lem:nolabeleddigons}
Let $\fw$ be a proper barren labeled \wlkg.
The labeled \wlkg $\phi(\fw)$ has no labeled digons.
\end{lemma}
\begin{proof}
We prove the lemma by induction on the length of the \wlkg $\fw$.

\smallskip\noindent\textsl{Case A.1.}
In this case, as $\phi(W^1[2,\ell_1], W^2, \ldots, W^p)$ has no labeled digons by induction, the only potential place for a labeled digon could be index 2 of $W^1$.
However, because $v^1_1$ occurs only once in $\fw$ and therefore does not occur in $\phi(W^1[2,\ell_1], W^2, \ldots, W^p)$, we have that the index 2 of $W^1$ cannot become a digon.

\smallskip\noindent\textsl{Case A.2.}
Trivially by induction.

\smallskip\noindent\textsl{Case B.}
Here, by the definition of $v^i_a$ and $v^j_b$ in case B, both  $v^i_a$ and $v^j_b$ are unlabeled and $v^i_a = v^j_b$ holds, so if $\fw \leftrightarrow^{i,j}_{a,b}$ would contain a labeled digon then also $\fw$ would.

\smallskip\noindent\textsl{Case C.1.a.}
Here, the indices $1$ and $b$ of $W^1$ are not labeled so they cannot become labeled digons.
For indices in $[2,b-1]$, note that if $i \in [2,b-1]$ would be a labeled digon in $W^1\overleftarrow{[2,b-1]}$, then $b+1-i$ would have been a labeled digon in $W^1$.

\smallskip\noindent\textsl{Case C.1.b.}
Potential places for labeled digons here are incides $b$ and $b+1$ at $W^1$.
However, $b$ is not labeled so no labeled digon can be at $b$, and because $v^1_b$ does not occur in $(W^1[b+1,\ell_1], W^2, \ldots, W^p)$, it cannot occur in $\phi(W^1[b+1,\ell_1], W^2, \ldots, W^p)$ and therefore $b+1$ cannot become labeled digon.

\smallskip\noindent\textsl{Case C.1.c.}
Trivially by induction.

\smallskip\noindent\textsl{Case C.2.}
The index 1 of $W^1$ is not digon by definition of digon, and the index $b$ is not digon by definition of this case, so no labeled digons are created.

\smallskip\noindent\textsl{Case C.3.}
Here, the index 1 cannot be a digon by definition, and the index $a$ of $W^1$ has $r^1_a = 0$ by case C.2, so they cannot become labeled digons.
For indices in $[2,a-1]$, the same argument as in case C.1.a applies.

\smallskip\noindent\textsl{Case C.4.a.}
Neither index $a$ nor $b$ is labeled so they cannot become labeled digons, and for indices in $[a+1,b-1]$ the same argument as in case C.1.a applies.

\smallskip\noindent\textsl{Case C.4.b.}
Same argument as in C.1.b, and using that by case C.2, it holds that $r^1_{b} = 0$.

\smallskip\noindent\textsl{Case C.X.}
Here, the index $a$ of $W^1$ cannot become a labeled digon because it is not labeled.
For the index $a+1$ the argument is more complicated:
First note that the vertex $v^1_a = v^1_b (= v^1_1)$ occurs only once in $(W^1[a+1,\ell_1], W^2, \ldots, W^p)$ (as $v^1_b$).
Also, by case C.2, $b$ is a digon in $W^1$.
Therefore, by \Cref{lem:presdigon}, the vertex $v^1_b = v^1_a$ occurs in $\phi(W^1[a+1,\ell_1], W^2, \ldots, W^p)$ only once, and this occurrence is a digon adjacent with vertex $v^1_{b-1}=v^1_{b+1}$.
Because by \Cref{lem:rand_phi_proper} $\phi$ preserves the starting vertices, and by case C.4 it holds that $v^1_{a+1} \neq v^1_{b-1}$, it holds that the occurrence of $v^1_b$ in $\phi(W^1[a+1,\ell_1], W^2, \ldots, W^p)$ cannot be as the second vertex of the first walk.
Therefore, the index $a+1$ of the first walk cannot be a labeled digon in $W^1[1,a] \concm \phi(W^1[a+1,\ell_1], W^2, \ldots, W^p)$.

\smallskip\noindent\textsl{Case D.1.}
As $r^1_{1} = r^i_b = 0$ and $v^1_1 = v^i_b$, it holds that if $\fw \leftrightarrow^{1,i}_{1,b}$ would contain a labeled digon then also $\fw$ would.

\smallskip\noindent\textsl{Case D.2.}
The index 1 of $W^1$ cannot become a digon by definition of digon.
The index $b$ of $W^i$ cannot become a digon by definition of this case.

\smallskip\noindent\textsl{Case D.3.}
Because $v^1_1 = v^i_a$, in this case if $\fw \leftrightarrow^{1,i}_{2,c+1}$ would have a labeled digon at index 2 of $W^1$ or at index $c+1$ of $W^i$, then this same labeled digon would have existed also in $\fw$.
Then, no labeled digon can be created to index 1 of $W^1$ by definition of digon or to index $c$ of $W^i$ because $r^1_1 \neq 0$ and therefore $r^i_c = 0$ by case D.2.

\smallskip\noindent\textsl{Case D.4.}
The index $1$ of $W^1$ cannot become a digon by definition, and the index $a$ cannot become a labeled digon because it is not labeled because the index $1$ is labeled.
For indices in $[2,a-1]$, the same argument as in case C.1.a applies.

\smallskip\noindent\textsl{Case D.5.}
Trivially by induction.

\smallskip\noindent\textsl{Case D.6.}
In this case, by cases D.1 and D.2 we have that $r^1_{1}\neq 0$ and $ r^i_b = 0$.
The former implies that $r^1_a = 0$.
Therefore, the index $a$ of $W^1$ nor the index $b$ of $W^i$ cannot become labeled digons because they cannot become labeled.
Then, if the index $a+1$ of $W^1$ would be a labeled digon in $\fw \leftrightarrow^{1,i}_{a+1,b+1}$, then the index $b+1$ of $W^i$ would have been a labeled digon in $\fw$ because $v^1_a = v^i_b$ (and symmetrically for $b+1$ of $W^i$).

\smallskip\noindent\textsl{Case D.X.}
Here, a similar argument as in C.X works: By \Cref{lem:presdigon} we know that $v^i_b$ occurs only as a digon surrounded by $v^i_{b-1} = v^i_{b+1}$ in $\phi(W^1[a+1,\ell_1], W^2, \ldots, W^p)$.
Therefore, because $\phi$ maintains the starting vertices and $v^1_{a+1} \neq v^i_{b+1}$, it holds that the index $a+1$ of $W^1$ cannot be a digon in $W^1[1,a] \concm \phi(W^1[a+1, \ell_1], W^2, \ldots, W^p)$.
The index $a$ of $W^1$ cannot become a labeled digon because it is not labeled.
\end{proof}

\medskip
\noindent\emph{The function $\phi$ is an involution.}
Now we have shown that $\phi$ is a function $\phi : \bls \rightarrow \bls$, where $\bls$ is the set of all barren proper labeled \wlkgs, and that $f(\phi(\fw)) = f(\fw)$, $\ev(\phi(\fw)) = \ev(\fw)$, and $\svs(\phi(\fw)) = \svs(\fw)$ hold.
Next we show that $\phi$ is an involution on $\bls$, i.e., $\phi(\phi(\fw)) = \fw$ holds.

\begin{lemma}
For any proper barren labeled \wlkg $\fw$ it holds that $\phi(\phi(\fw)) = \fw$.
\end{lemma}
\begin{proof}
We use induction on the length of \wlkg $\fw$.
The structure of the proof is to show that in all cases except C.X and D.X, the walk $\phi(\fw)$ goes to the same case of \autoref{def:invophi} as $\fw$.
Then, the cases C.X and D.X are treated together.

\smallskip\noindent\textsl{Case A.1.}
In both $\fw$ and $\phi(\fw)$ it holds that the $v^1_1$ occurs only once in the \wlkg and $\phi$ does not change the first vertex, so $\phi(\phi(\fw)) = \fw$ holds by induction.

\smallskip\noindent\textsl{Case A.2.}
In this case the first walk of $\fw$ is the same as the first walk of $\phi(\fw)$, so $\phi(\phi(\fw)) = \fw$ holds by induction.

\smallskip\noindent\textsl{Case B.}
In this case, the function $\phi$ preserves the set of walks in which $v^1_1$ occurs, and moreover preserves the walk in which $v^1_1$ occurs as labeled (if it occurs as labeled in any walk).
Therefore the \wlkg $\phi(\fw)$ also goes to case B, and in case B the indices $i,j$ selected for $\phi(\fw)$ are the same as selected for $\fw$.
The suffix swap operation also does not change the indices of the first occurrences of $v^1_1$ in $W^i$ and $W^j$, so the indices $a$ and $b$ selected are the same.
Then, the lemma follows by observing that $\fw \leftrightarrow^{i,j}_{a,b} \leftrightarrow^{i,j}_{a,b} = \fw$.

At this point, let us observe that $\fw$ goes to cases C.1-4 if and only if $\phi(\fw)$ goes to cases C.1-4.
This is because all of these cases maintain that $v^1_1$ occurs only in the walk $W^1$, but multiple times in the walk $W^1$.
In particular, in the recursive cases this is maintained by the fact that $v^1_1$ does not occur in the recursive argument.

We also observe that $\fw$ goes to cases D.1-6 if and only if $\phi(\fw)$ goes to cases D.1-6.
This is because all of these cases maintain that $v^1_1$ occurs in exactly two different walks.
In the cases D.1-4 and D.6 this is easy to observe since these are not recursive, and in the case D.5 this follows from the fact that the walk $W^1$ is not changed and that in this case $v^1_1$ occurs exactly once in $(W^2, \ldots, W^p)$.

\smallskip\noindent\textsl{Case C.1.}
Note that going to the cases C.1.a, C.1.b, and C.1.c depends only on the last occurrence $b$ of $v^1_1$, and the labels $r^1_{1}$ and $r^1_{b}$.
None of these cases change these, so $\phi(\phi(\fw)) = \fw$ holds in case C.1.a because $W^1\overleftarrow{\overleftarrow{[2,b-1]}} = W^1$ and reversing a subwalk does not change whether it is a palindrome, and by induction in cases C.1.b and C.1.c.

\smallskip\noindent\textsl{Case C.2.}
The order of the vertices in the walks and the fact that exactly one of $r^1_{1}$ and $r^1_{b}$ is labeled is maintained, so if $\fw$ goes to case C.2 then also $\phi(\fw)$ goes to C.2.
Then, $\phi(\phi(\fw)) = \fw$ because $\fw \frown^{1,1}_{1,b} \frown^{1,1}_{1,b} = \fw$.

\smallskip\noindent\textsl{Case C.3.}
It holds that $a < b$, so therefore reversing $W^1[2,a-1]$ does not change the fact that $b$ is a digon in $W^1$.
It also does not change the index $a$ of the second last occurrence of $v^1_1$, nor the index $b$ of the last occurrence of $v^1_1$, nor the fact that $r^1_{1} \neq 0$, nor the fact that $W^1[2,a-1]$ is not a palindrome.

\smallskip\noindent\textsl{Case C.4.a.}
Because $v^1_{a+1} = v^1_{b-1}$, reversing $W^1[a+1,b-1]$ does not change the fact that $b$ is a digon in $W^1$.
Reversing $W^1[a+1,b-1]$ also does not change the index $a$ of the second last occurrence of $v^1_1$, nor the index $b$ of the last occurrence of $v^1_1$, nor the fact that $r^1_{1} \neq 0$, nor the fact that $W^1[2,a-1]$ is a palindrome.

\smallskip\noindent\textsl{Case C.4.b.}
Going to the case C.4.b depends only on the subwalk $W^1[1,b]$ and on the vertex with index $b+1$ in $W^1$ (whether the index $b$ is a digon).
Clearly, $\phi$ does not change the subwalk $W^1[1,b]$ in this case.
The vertex with index $b+1$ is not changed because by \autoref{lem:rand_phi_proper} the starting vertices are preserved by $\phi(W^1[b+1,\ell_1], W^2, \ldots, W^p)$, so the lemma holds by induction.

Before moving to the cases C.X and D.X, we handle the cases D.1-D.5.

\smallskip\noindent\textsl{Case D.1.}
In this case, the operation $\fw \leftrightarrow^{1,i}_{1,b}$ does not change the two walks in which $v^1_1$ occurs, nor it changes the fact that the first occurrence of $v^1_1$ in $W^i$ is at index $b$, nor that $r^1_1 = r^i_b = 0$.
The lemma follows from the fact that $\fw \leftrightarrow^{1,i}_{1,b} \leftrightarrow^{1,i}_{1,b} = \fw$.

\smallskip\noindent\textsl{Case D.2.}
This case does not change the vertices of the walks, so it is maintained that $v^1_1$ occurs only in $W^1$ and $W^i$.
It also does not change the fact that exactly one of $r^1_{1}$ and $r^i_b$ is labeled or the fact the index $b$ is not a digon in $W^i$, so the lemma follows from the fact that $\fw \frown^{1,i}_{1,b} \frown^{1,i}_{1,b} = \fw$.

\smallskip\noindent\textsl{Case D.3.}
This case does not change the index $b$ of the first occurrence of $v^1_1$ in $W^i$ or the index $c$ of the second occurrence of $v^1_1$ in $W^i$, and neither does it change the fact that $r^1_{1} \neq 0$.
Because $c>b+1$, it also does not change that the index $b$ is a digon in $W^i$.
The lemma follows from the fact that $\fw \leftrightarrow^{1,i}_{2,c+1} \leftrightarrow^{1,i}_{2,c+1} = \fw$.

\smallskip\noindent\textsl{Case D.4.}
This case does not change the index $b$ of the first occurrence of $v^1_1$ in $W^i$, nor that $r^1_{1} \neq 0$, nor that $b$ is a digon in $W^i$, nor that $v^1_1$ occurs only once in $W^i$.
It also does not change the index $a$ of the last occurrence of $v^1_1$ in $W^i$, or the fact that $W^1[2,a-1]$ is a palindrome so the lemma holds.

\smallskip\noindent\textsl{Case D.5.}
The vertex $v^1_1$ occurs exactly once in $(W^2, \ldots, W^p)$, so it is also maintained that $v^1_1$ occurs in exactly two walks.
The walk $W^1$ is not changed, so it is maintained that $r^1_{1} \neq 0$ and therefore $\phi(\fw)$ does not go to case D.1.
By \autoref{lem:presdigon} we have that $\phi(\fw)$ does not go to case D.2, and again because $v^1_1$ occurs only once in $(W^2, \ldots, W^p)$ we have that $\phi(\fw)$ does not go to case D.3.
As $W^1$ is not changed we have that $\phi(\fw)$ does not go to case D.4.
Then, as case D.5 does not change the walk $W^1$, it is maintained that $a = \ell_1$, so $\phi(\fw)$ goes to case D.5.

\smallskip\noindent\textsl{Case D.6.}
This case does not change that $r^1_{1} \neq 0$, so $\phi(\fw)$ does not go to case D.1.
It also does not change the index $b$ of the first occurrence of $v^1_1$ in $W^i$, and because $v^1_{a+1} = v^i_{b+1}$, it does not change that $b$ is a digon in $W^i$, and therefore $\phi(\fw)$ does not go to case D.2.
Because $a$ is the last occurrence of $v^1_1$ in $W^1$, it also does not change that $v^1_1$ occurs in $W^i$ only once, so $\phi(\fw)$ does not go to case D.3.
It also does not change the index $a$ of the last occurrence of $v^1_1$ in $W^1$ or the subwalk $W^1[2,a-1]$, so $\phi(\fw)$ does not go to cases D.4. or D.5.
Therefore, $\phi(\fw)$ goes to case D.6 with the same values of $a$, $b$, and $i$, so $\phi(\phi(\fw)) = \fw$ holds because $\fw \leftrightarrow^{1,i}_{a+1,b+1} \leftrightarrow^{1,i}_{a+1,b+1} = \fw$.

\smallskip\noindent\textsl{Case C.X.}
We aim to prove that if $\fw$ goes to case C.X, then $\phi(\fw)$ also goes to case C.X or to the case D.X, with the same value of the index $a$, and therefore $\phi(\phi(\fw)) = \fw$ will hold by induction, as in both cases $\phi$ is defined as $\phi(\fw) = W^1[1,a] \concm \phi(W^1[a+1,\ell_1], W^2, \ldots, W^p)$.
First, it is maintained that $v^1_1$ occurs more than once, but in at most two walks, because $v^1_1$ occurs only once in $(W^1[a+1,\ell_1], W^2, \ldots, W^p)$.
Therefore, $\phi(\fw)$ does not go to case A or B.

Suppose that $v^1_1$ occurs in $\phi(\fw)$ only in the walk $W^1$, i.e., goes to case C.
We will show that $\phi(\fw)$ goes to case C.X.
It is maintained that $r^1_{1} \neq 0$, so $\phi(\fw)$ does not go to case C.1.
Then, by \autoref{lem:presdigon} it is maintained that the last occurrence of $v^1_1$ must be a digon so $\phi(\fw)$ does not go to case C.2.
Also, this case does not change the index $a$ of the second last occurrence of $v^1_1$ nor the walk $W^1[2,a-1]$, so it is maintained that $W^1[2,a-1]$ is a palindrome and therefore $\phi(\fw)$ does not go to case C.3.
Then, to argue that $\phi(\fw)$ does not go to case C.4, observe that because $\phi$ maintains the starting vertex, the vertex at the position $a+1$ is maintained.
The vertices around the digon at the last occurrence of $v^1_1$ are maintained by \autoref{lem:presdigon}, so therefore if $\fw$ does not go to case C.4 then also $\phi(\fw)$ does not go to case C.4.
Therefore $\phi(\fw)$ goes to case C.X, and as the walk $W^1[1,a]$ is maintained, it goes to this case with the same value of $a$, so the lemma holds by induction.

Then, suppose that $v^1_1$ occurs in $\phi(\fw)$ in two walks $W^1$ and $W^i$, i.e., goes to case D.
We will show that $\phi(\fw)$ goes to case D.X.
It is maintained that $r^1_{1} \neq 0$, so $\phi(\fw)$ does not go to case D.1.
Then, by \autoref{lem:presdigon} it must be that $v^1_1$ occurs in $W^i$ only once and as a digon, and therefore $\phi(\fw)$ does not go to case D.2 or D.3.
Now, it will hold that the index $a$ of the last occurrence of $v^1_1$ in the walk $W^1$ of $\phi(\fw)$ is the same as the index $a$ of the last occurrence of $v^1_1$ in the walk $W^1$ of $\fw$.
Therefore, the subwalk $W^1[1, a]$ will be the same in $\fw$ and $\phi(\fw)$, and therefore $\phi(\fw)$ will not go to case D.4 because $\fw$ did not go to case C.3.
Then, because $\phi$ cannot turn a non-empty walk into an empty walk, it is maintained that the length of $W^1$ is more than $a$, so $\phi(\fw)$ cannot go to case D.5.
Then, for case D.6 we again note that $\phi$ maintains the vertex at position $a+1$, and that the digon around the occurrence of $v^1_1$ outside of $W^1[1,a]$ is maintained by \autoref{lem:presdigon}.
Therefore, $\phi(\fw)$ goes to case D.X with the same value of $a$, so the lemma holds by induction.

\smallskip\noindent\textsl{Case D.X.}
We will show that if $\fw$ goes to case D.X, then $\phi(\fw)$ also goes to case D.X or to the case C.X, with the same value of the index $a$, and therefore $\phi(\phi(\fw)) = \fw$ will hold by induction, as in both cases $\phi$ is defined as $\phi(\fw) = W^1[1,a] \concm \phi(W^1[a+1,\ell_1], W^2, \ldots, W^p)$.
First, it is maintained that $v^1_1$ occurs more than once, so therefore $\phi(\fw)$ does not go to case A.
Then, as $v^1_1$ occurs only once in $(W^1[a+1,\ell_1], W^2, \ldots, W^p)$, it can occur in at most two walks in $\phi(\fw)$, so $\phi(\fw)$ cannot go to case B.

Suppose that $v^1_1$ occurs in $\phi(\fw)$ only in the walk $W^1$, i.e., goes to case C.
We will show that $\phi(\fw)$ goes to case C.X.
It is maintained that $r^1_{1} \neq 0$, so $\phi(\fw)$ does not go to case C.1.
Then, because $W^1[1,a]$ contains all other occurrences of $v^1_1$ in $\fw$ except the occurrence in $W^i$, it must be that now the last occurrence of $v^1_1$ in $W^1$ of $\phi(\fw)$ corresponds to the occurrence of $v^1_1$ in $W^i$ of $\fw$, in particular, by \autoref{lem:presdigon} the last occurrence of $v^1_1$ in $W^1$ of $\phi(\fw)$ must be a digon, and therefore $\phi(\fw)$ does not go to case C.2.
By the same reasoning, it also must be that the the index $a$ of the last occurrence of $v^1_1$ in $W^1$ of $\fw$ is the same as the index $a$ of the second last occurrence of $v^1_1$ in $W^1$ of $\phi(\fw)$, and therefore the subwalk $W^1[1,a]$ of $\phi(\fw)$ in case C is the same as the subwalk $W^1[1,a]$ of $\fw$ in case D.
Then, it follows that because $\fw$ did not go to case D.4, $\phi(\fw)$ does not go to case C.3.
Then, as the case D.X does not change the vertex at the index $a+1$ of $W^1$, it holds that the vertex at the index $a+1$ of $W^1$ is the same in $\fw$ and $\phi(\fw)$.
Also, by \autoref{lem:presdigon} the vertices around the digon of the last occurrence of $v^1_1$ are the same in $\fw$ and $\phi(\fw)$, so $\phi(\fw)$ does not go to case C.4 because $\fw$ did not go to case D.5.
Therefore, $\phi(\fw)$ must go to case C.X, and we already reasoned that the index $a$ is the same as for $\fw$ in the case D.X, so the lemma holds by induction.

Suppose that $v^1_1$ occurs in $\phi(\fw)$ in two different walks, i.e., goes to case D.
We will show that $\phi(\fw)$ goes to case D.X.
First, it is maintained that $r^1_{1} \neq 0$, so $\phi(\fw)$ does not go to case D.1.
Then, we note that the walk $W^i$ that contains the other occurrence of $v^1_1$ may be different for $\phi(\fw)$ than $\fw$.
However, it is maintained that $v^1_1$ occurs only once outside of $W^1$, and by \autoref{lem:presdigon} that the other occurrence is a digon and the vertices around this digon are maintained.
Therefore, $\phi(\fw)$ does not go to case D.2, nor to the case D.3.
Now, the index $a$ of the last occurrence of $v^1_1$ in $W^1$ will be the same for $\phi(\fw)$ and $\fw$ because $\phi(\fw)$ does not change the subwalk $W^1[1,a]$.
Therefore, it is maintained that $W^1[2,a-1]$ is a palindrome, and therefore $\phi(\fw)$ does not go to case D.4.
Then, because $\phi$ cannot turn a non-empty walk into an empty walk, it is maintained that the length of $W^1$ is more than $a$, so $\phi(\fw)$ cannot go to case D.5.
Then, by the start vertex property of $\phi$, the vertex at index $a+1$ of $W^1$ is also maintained, and by \autoref{lem:presdigon} the vertices around the digon of the other occurrence of $v^1_1$ are maintained, so $\phi(\fw)$ does not go to case D.6.
Therefore, $\phi(\fw)$ goes to the case D.X, with the same value of $a$, and therefore the lemma holds by induction.
\end{proof}

Finally, we show that $\phi$ is fixed-point-free.

\begin{lemma}
For any proper barren labeled \wlkg $\fw$ it holds that $\phi(\fw) \neq \fw$.
\end{lemma}
\begin{proof}
We prove this by induction on the length $\ell$ of the \wlkg.
In the recursive cases A.1, A.2, C.1.b, C.1.c, C.4.b, C.X, D.5, and D.X this holds directly by induction.
In cases B, D.1, D.3, and D.6, $\phi$ changes the suffixes of two walks, and at least one of the suffixes is non-empty.
Because $\fw$ is proper, the ending vertices of all walks in $\fw$ are different, so the ending vertex of at least one of the walks involved in the suffix swap is changed (in fact the ending vertices of both of the walks change, but it is not necessary for this proof).
In cases C.1.a, C.3, C.4.a, and D.4, $\phi$ reverses a non-palindromic subwalk so $\phi(\fw) \neq \fw$.
In cases C.2 and D.2 $\phi$ changes a label from one position to another, so $\phi(\fw) \neq \fw$.
\end{proof}

This completes the proof of \autoref{lem:main_phi_exists}.

 
\section{From colored graphs to frameworks\label{sec:extframeworks}}
In this section, we extend our results from weighted colored graphs to weighted frameworks, in particular, we prove \Cref{thm:weightedframeworkthm} (recall that \Cref{thm:weightedframeworkthm} implies \Cref{thm:main-frameworks}), and then discuss even further extensions to frameworks where the matroid is not necessarily represented over a finite field.

\subsection{Frameworks}
We recall definitions related to frameworks.

\medskip\noindent\emph{Matroids.} We refer to the textbook of Oxley~\cite{Oxley11} for the introduction to Matroid Theory.

\begin{definition}\label{def:matroid}
A pair $M=(V,\mathcal{I})$, where $V$ is a \emph{ground set} and $\mathcal{I}$ is a family of subsets of $V$, called \emph{independent sets of $M$}, is a \emph{matroid} if it satisfies the following conditions, called \emph{independence axioms}:
\begin{itemize}
\item[~{\em (I1)}]  $\emptyset\in \mathcal{I}$, 
\item[~{\em (I2)}]  if $X \subseteq Y $ and $Y\in \mathcal{I}$ then $X\in\mathcal{I}$, 
\item[~{\em (I3)}] if $X,Y  \in \mathcal{I}$  and $ |X| < |Y| $, then there is $v\in  Y \setminus X $  such that $X\cup\{v\} \in \mathcal{I}$.
\end{itemize}
\end{definition}
An inclusion maximal set of $\mathcal{I}$ is called a \emph{base}.
We use $V(M)$ and $\mathcal{I}(M)$ to denote the ground set and the family of independent sets of $M$, respectively.  

Let $M=(V,\mathcal{I})$ be a matroid. We use $2^V$ to denote the set of all subsets of $V$.
A function $r\colon 2^V\rightarrow \mathbb{Z}_{\geq 0}$  such that for every $X\subseteq V$,
\begin{equation*}
r(X)=\max\{|Y|\colon Y\subseteq X\text{ and }Y\in \mathcal{I}\}
\end{equation*}
is called the \emph{rank function} of $M$. The \emph{rank of $M$}, denoted $r(M)$, is $r(V)$; equivalently, the rank of $M$ is the size of any base of $M$.  
A matroid
$M'=(V,\mathcal{I}') $ is a \emph{$k$-truncation} of $M=(V,\mathcal{I})$ if for every $X\subseteq V$, $X\in \mathcal{I}'$ if and only if $X\in \mathcal{I}$ and $|X|\leq k$.

We work with several particular types of matroids.  A \emph{uniform} matroid is defined by the ground set  $V$ and its rank $r$; 
 every subset $S$ of $V$  of size at most $r$ is independent.  \emph{Partition} matroids are the matroids that can be written as disjoint sums of uniform matroids. \emph{Transversal} matroids arise from graphs. For a bipartite graph $G = (U \cup B, E)$ with all edges between $U$ and $B$, we can define a matroid $M=(V,\mathcal{I})$ such that a set $S \subseteq V$ is independent if there exists a matching in $G$ such that every vertex in $S$ is an endpoint of a matching edge.

\medskip\noindent\emph{Matroid representations.}
Let  $M=(V,\mathcal{I})$ be a matroid and let $\mathbb{F}$ be a field. An $r\times n$-matrix $A$ is a \emph{representation of $M$ over $\mathbb{F}$} if there is a bijective correspondence $f$ between $V$ and the set of columns of $A$ such that for every $X\subseteq V$, $X\in \mathcal{I}$ if and only if the set of columns $f(X)$ consists of linearly independent vectors of $\mathbb{F}^r$. Equivalently, $A$ is a representation of $M$ if $M$ is isomorphic to the \emph{column} matroid of $A$, that is, the matroid whose ground set is the set of columns of the matrix and the independence of a set of columns is defined as the linear independence.   
If $M$ has a  
such a representation, then $M$ is \emph{representable} over $\mathbb{F}$ and it is also said
$M$ is a \emph{linear} (or \emph{$\mathbb{F}$-linear}) matroid. 
We can assume that the number of rows $r=r(M)$ for a matrix representing $M$~\cite{Marx09}. 

Whenever we consider a linear matroid, it is assumed that its representation is given and the size of $M$ is $\|M\|=\|A\|$, that is, the bit-length of the representation matrix.
Notice that given a representation of a matroid, deciding whether a set is independent demands a polynomial number of field operations.
In particular, if the considered field is a finite or is the field of rationals, we can verify independence in time that is a polynomial in $\|M\|$.

\medskip\noindent\emph{Frameworks.}
A framework is a pair $(G, M)$, where $M = (V, \mathcal{I})$ is a matroid whose ground set is the set of vertices of $G$, i.e., $V(M) = V(G)$.
A weighted framework is a triple $(G, M, \we)$, where $(G,M)$ is a framework and $\we : V(G) \rightarrow \mathbb{Z}_{\ge 1}$ is a weight function.
An $(S,T)$-\lkg $\ps$ in a weighted framework $(G,M,\we)$ is $(k,w)$-ranked if $V(\ps)$ contains a set $X \subseteq V(\ps)$ with $X \in \mathcal{I}$, size $|X| = k$, and weight $\we(X) = w$.
When discussing algorithms for (weighted) frameworks, we explicitly specify how $M$ is represented.

\subsection{From colored graphs to frameworks\label{subsec:rand_rep}}
In this section we  prove \Cref{thm:weightedframeworkthm}.
We reduce the more general cases of matroids to the case of \Cref{thm:weightedmain}.

We start by giving our algorithm for the special case when the rank of $M$ is bounded by $k$, in particular when $M$ is represented as a $k \times n$ matrix.

\begin{lemma}
\label{lem:matroidfullrank}
There is a randomized algorithm, that given a weighted framework $(G,M, \we)$, where $G$ is an $n$-vertex graph and $M$ is represented as a $k \times n$ matrix over a finite field of order $q$, sets of vertices $S,T\subseteq V(G)$, and integers $p,k,w$, in time $2^{p+\Oh(k^2 \log q)} n^{\Oh(1)} w$ either finds a $(k,w)$-ranked $(S,T)$-\lkg of order $p$ and of minimum total length, or determines that $(G,M,\we)$ has no $(k,w)$-ranked $(S,T)$-\lkgs of order $p$.
\end{lemma}
\begin{proof}
The matrix has at most $q^k$ distinct column vectors so we can guess the $k$ column vectors forming the independent set $X$ of size $k$ that we are looking for with at most $q^{k^2}$ guesses.
By inserting $|S|$ new vertices with neighborhoods equal to $S$, all-zero column vector, and weight $1$, we can assume that the vectors of the starting vertices $S$ will never correspond to the guessed basis.
Then, we $k+1$-color the graph, assigning the color $k+1$ to the vertices of the set $S$ and other vertices whose column vectors are not in the guessed basis, and the colors $[k]$ to the other vertices according to which of the $k$ guessed column vectors they correspond to.
We also assign the weight of all vertices whose column vector is not in the guessed basis to be $1$.

Then, $(G,M,\we)$ has a $(k,w)$-ranked $(S,T)$-\lkg of order $p$ if and only if it has a $(k+1,w+1)$-colored $(S,T)$-\lkg of order $p$.
In particular, the extra color $k+1$ contributes weight one and one color more, and the selected set $X \subseteq V(\ps)$ without the extra color must correspond to an independent set of $M$.
Therefore, we get an algorithm with time complexity $q^{k^2} 2^{p+k} n^{\Oh(1)} w = 2^{p + \Oh(k^2 \log q)} n^{\Oh(1)} w$.
\end{proof}

By extending \Cref{lem:matroidfullrank} to matrices with a large number of rows by using randomized lossy truncation, we prove  \Cref{thm:weightedframeworkthm}, 
 which we restate here.
 
 \thmweightedframeworks*
\begin{proof}
Let $A$ be a $r \times n$ matrix representing $M$.
Our goal is to obtain a ``lossy'' representation of the $k$-truncation of $M$ as a $k \times n$ matrix over a field of order $\Oh(q+k^2)$.
In particular, a representation so that any independent set of $M$ of size $k$ is independent in the representation with probability $\ge 1/2$, and any dependent set of $M$ is dependent in the representation.
Then, we obtain the algorithm by applying the algorithm of \autoref{lem:matroidfullrank}.
Note that $2^{p + \Oh(k^2 \log (q+k^2))} n^{\Oh(1)} w = 2^{p + \Oh(k^2 \log (q+k))} n^{\Oh(1)} w$.

We use two techniques from~\cite{Marx09}, increasing the order of the field and truncation.
First, we make sure that the order of the field is at least $2k$ by choosing the least integer $i$ such that $q^i \ge 2k$, and going to the field of order $q^i$, as detailed in Proposition~3.2 of~\cite{Marx09}.
Now, we can assume that $A$ is over a field of order at least $2k$ and at most $\Oh(q+k^2)$.

Then, we truncate the matroid by multiplying the matrix $A$ by a random $k \times r$ matrix $R$, in particular we claim that the $k \times n$ matrix $B = RA$ is now the desired representation of the $k$-truncation of $M$.
The analysis here is the same as in Proposition~3.7 of~\cite{Marx09}, but with a smaller field.
In particular, let us consider a subset $U$ of the ground set of $M$, and let $A_0$ be the $r \times |U|$ submatrix of $A$ corresponding to $S$.
Now, in $B$, the $k \times |U|$ submatrix corresponding to $S$ will be the matrix $B_0 = R A_0$.
The rank of $B_0$ is at most the rank of $A_0$, so if $U$ is dependent in $M$ it will be dependent in the representation by $B$.
Then, assume that $U$ is an independent set of $M$ and $|U| = k$.
Now, $\det R A_0$ can be considered as a degree-$k$ non-zero polynomial whose variables are the $kr$ random entries of $R$.
Therefore, by \autoref{lem:schwartzzippel}, the probability that $\det R A_0 = 0$ is at most $k/2k$.
\end{proof}

With minor adjustments, \Cref{thm:weightedframeworkthm} can be adapted for frameworks with matroids
that are in general not representable over a finite field of small order.
For example, uniform matroids, and more generally transversal matroids, are representable over a finite field, but the field of representation must be large enough. 
We first show how \Cref{thm:maintheorem} can be applied in the case of transversal matroids.

\begin{theorem}\label{lemma:transv}
There is a randomized algorithm that given a weighted framework $(G,M,\we)$, where $G$ is an $n$-vertex graph and $M$ is a transversal matroid represented by the corresponding bipartite graph, sets of vertices $S,T\subseteq V(G)$, and integers $p,k,w$, in time $2^{p+\Oh(k^2 \log k)} n^{\Oh(1)} w$
either finds a $(k,w)$-ranked $(S,T)$-\lkg of order $p$ and of minimum total length, 
or determines that $(G,M,\we)$ has no $(k,w)$-ranked $(S,T)$-\lkgs of order $p$.
\end{theorem}
\begin{proof}
We will construct a representation of the transversal matroid as a linear matroid over a finite field of order $\Oh(k)$, so that any independent set of $M$ of size $k$ is independent in the representation with probability $\ge 1/2$, and any dependent set of $M$ is dependent in the representation.
This yields the algorithm by then using \Cref{thm:weightedframeworkthm}.

Our construction is the same as the construction of~\cite{Marx09}, except by using a smaller field.
We choose the least prime $p$ with $p \ge 2k$ and work in the field of order $p$.
Let the bipartition of the vertices of the bipartite graph be $(A,B)$.
We construct an $|B| \times |A|$ matrix, so that an entry of the matrix is a random element of the field if it corresponds to an edge, and zero otherwise.
Now, the determinant of a submatrix is guaranteed to be zero if there is no corresponding matching, so any dependent set of $M$ is dependent in the representation.
Otherwise, the determinant of a $k \times k$ submatrix can be seen as a non-zero degree-$k$ polynomial that was evaluated at a random point.
Therefore, as $p\ge2k$, by \autoref{lem:schwartzzippel}, the probability that it is non-zero is at least $1/2$.
\end{proof}

It is also possible to apply  \Cref{thm:weightedframeworkthm}  in the situation when $M$ is represented by 
an integer matrix over rationals with entries bounded by $n^{\Oh(k)}$.

\begin{theorem}\label{lemma:rationals}
There is a randomized algorithm that given a weighted framework $(G,M,\we)$, where $G$ is an $n$-vertex graph and $M$ is represented as an integer matrix over rationals with entries bounded by $n^{\Oh(k)}$, sets of vertices $S,T\subseteq V(G)$, and integers $p,k,w$, in time $2^{p+\Oh(k^2 \log k)} n^{\Oh(1)} w$ 
either finds a $(k,w)$-ranked $(S,T)$-\lkg of order $p$ and of minimum total length, 
or determines that $(G,M,\we)$ has no $(k,w)$-ranked $(S,T)$-\lkgs of order $p$.
\end{theorem}
\begin{proof}
Let $c$ be a constant so that the entries of the matrix are bounded by $n^{ck}$.
We pick a random prime $p$ among the first $2 \log_2 (k! n^{ck^2})$ primes, go to the finite field of order $p$ by taking every entry modulo $p$, and then apply the algorithm of \autoref{thm:maintheorem}.

We first analyze the time complexity and then the correctness.
By the prime number theorem, the prime $p$ is bounded by 
\[p = \Oh(\log k! n^{ck^2} \cdot \log \log k! n^{ck^2}) = \Oh(k^3 \log n \log \log n).\]
We can find such random prime in $n^{\Oh(1)}$ time by elementary methods.
Then, the time complexity by using \autoref{thm:maintheorem} will be $2^{\Oh(p + k^2 \log (k + k^3 \log n \log \log n)} n^{\Oh(1)}$.
Denote $t(n) = \log n \log \log n$ and consider two cases.
First, if $t(n) \le k^5$, then the time complexity is $2^{\Oh(p + k^2 \log (k^8)} n^{\Oh(1)} = 2^{\Oh(p + k^2 \log k)} n^{\Oh(1)}$.
Second, if $t(n) > k^5$, then the time complexity is $2^{\Oh(p+k^2 \log k^3 t(n))} n^{\Oh(1)} = 2^{\Oh(p+t(n)^{1/2} \log t(n))} n^{\Oh(1)} = 2^{\Oh(p)}\cdot2^{\Oh(\log^{1/2} n \log^{\Oh(1)} \log n)} =2^{\Oh(p)} n^{\Oh(1)}$.

Then, for the correctness we show that any dependent set of $M$ is dependent in the representation and any independent set of $M$ of size $k$ is independent in the representation with probability $\ge 1/2$.
Let $A$ be a square submatrix of the original representation and $A_p$ the corresponding submatrix in the presentation modulo $p$.
Now, $\det A_p = \det A \mod p$.
Therefore, all dependent sets stay dependent.
Then, assume that $A$ is a $k \times k$ submatrix corresponding to an independent set, i.e., $\det A \neq 0$.
Now, the independent set can change into dependent only if $\det A$ is divisible by $p$.
The value $\det A$ is bounded by $k! n^{ck^2}$, so there are at most $\log_2 (k! n^{ck^2})$ primes dividing it.
We chose $p$ randomly among the first $2 \log_2 (k! n^{ck^2})$ primes, so with probability $\ge 1/2$ the prime $p$ does not divide $\det A$.
\end{proof}
 

\newcommand{\newpref}{\widehat{P}}
\newcommand{\newsuf}{\widehat{Q}}

\section{Deterministic algorithm for longest $(S,T)$-\lkg}\label{sec:detKiril}

This section is dedicated to the proof of \Cref{thm:detmaintheorem}, which we restate next for convenience.

\detmaintheorem*

We start with showing the main combinatorial lemma behind the theorem. The lemma is illustrated in \Cref{fig:tokens}.

\begin{lemma}
    Let $G$ be a digraph and let $C_1$, \ldots, $C_q$ be disjoint sets in $V(G)$. For $s, t \in V(G)$ let $P_1$, \ldots, $P_q$ be internally-disjoint $(s, t)$-paths. For each $i \in [q]$, let $v_i \in V(P_i)$ be such that the suffix of $P_i$ starting from $v_i$ lies inside $C_i$, except for $t$. For $i \in [q]$, let $Q_i$ be a path from $v_i$ to $t$ with all internal vertices in $C_i$. Then there exist internally-disjoint $(s,t)$-paths $P_1'$, \ldots, $P_q'$ such that $P_i'$ is either (i) $P_i$ or (ii) a composition of a prefix of $P_i$ not containing any vertices of $P_i$ beyond $v_i$, and a suffix of $Q_j$ for some $j \in [q]$, and there is at least one path of type (ii) among $P_1'$, \ldots, $P_q'$.
    \label{lemma:tokens}
\end{lemma}

\begin{proof}
    For each $i \in [q]$, denote the subpath of $P_i$ from $s$ to $v_i$ by $P_i^\rightarrow$, and from $v_i$ to $t$ by $P_i^\leftarrow$.
    First, assume there exists $i \in [q]$ such that $Q_i$ does not share a common internal vertex with any $P_j^\rightarrow$, $j \in [q]$. In this case, the solution is immediate: for each $j \ne i$, set $P_j' = P_j$, and set $P_i' = P_i^\rightarrow \circ Q_i$. The path $P_i^\rightarrow$ does not intersect any other $P_j'$ internally since $P_1$, \ldots, $P_q$ are internally-disjoint, and by the assumption $Q_i$ internally intersects neither $P_i^\rightarrow$ nor any other $P_j'$. So for the remaining part of the proof we assume that for each $Q_i$ there exists $j \in [q]$ such that $Q_i$ and $P_j^\rightarrow$ share a common internal vertex.

    We now show the statement by analyzing a certain token sliding game. Intuitively, we put a token on each of the paths $P_1$, \ldots, $P_q$, originally on the place of the first intersection between $P_i$ and some $Q_j$ (see \Cref{fig:tokens_b}). Then we slide the tokens further along the paths according to certain rules, until the tokens reach a state where no rules can be applied (\Cref{fig:tokens_c,fig:tokens_d}). Our goal is to show that in this case we obtain the  desired paths $P_1'$, \ldots, $P_q'$.

    More formally, we define a \emph{state} $S$ as a tuple $(t_1, \ldots, t_q)$, where $t_i \in V(P_i)$ for each $i \in [q]$.
    The original state $S^1 = (t_1^1, \ldots, t_q^1)$ is defined as follows: for each $i \in [q]$, $t_i^1$ is the closest to $s$ vertex along $P_i$ that belongs to $Q_j$ for some $j \in [q]$. The game then proceeds iteratively, constructing the state $S^{h + 1}$ from $S^h$ for each $h$ starting from $h = 1$ by applying one of the following rules. For $j \in [q]$, we shall refer to a path $Q_j$ as \emph{active} if $t_j^h \ne t$.

    \begin{description}
        \item[Clear] Let $t^h_i = v_i$ for some $i \in [q]$. Then set $t^{h + 1}_i = t$, and for each $j \in [q]$ such that $t^h_j \in V(Q_i) \setminus \{t\}$, set $t^h_j$ to be the next vertex along the path $P_j$ that belongs to an active $Q_{j'}$ for some $j' \in [q]$, here $Q_i$ is not considered active. For all remaining $j \in [q]$, set $t^{h + 1}_j = t^h_j$.
        \item[Push] Let $i$ and $i'$, $i \ne i' \in [q]$, be such that both $t^h_i$ and $t^h_{i'}$ belong to $Q_j - \{t\}$ for some $j \in [q]$; additionally, let $t^h_i$ be the farthest of two from $t$ along $Q_j$. Set $t_i^h$ to be the next vertex along the path $P_i$ that belongs to an active $Q_{j'}$ for some $j' \in [q]$. For all $i'' \in [q]$, $i'' \ne i$, set $t^{h + 1}_{i''} = t^h_{i''}$.
    \end{description}

    As long as there is a possibility, a \textbf{Clear} rule is applied; \textbf{Push} is only applied if no \textbf{Clear} is available. If there are several options for applying the same rule, ties are breaking arbitrarily. We observe that every application of each rule moves at least one of the state vertices further along its respective path, and these vertices are never moved back. Thus, after a finite number of steps we reach a state where neither of the rules is applicable. Denote this state by $S^T_i$, our goal is to show the following.

    \begin{claim}\label{claim:token_left}
        Let $S^T_i = (t^T_1, \ldots, t^T_q)$. There is $i \in [q]$ such that $t^T_i \ne t$.
    \end{claim}

    Before showing the proof of \Cref{claim:token_left} we make the following simple observation.

    \begin{claim}\label{claim:end_tokens}
        For a state $S^h$, $h \in [T]$, for each $j \in [q]$, $t_j^h \in V(P_j)$, and $t_j^h$ is either $t$ or belongs to an active $Q_i$ for some $i \in [q]$. Moreover, all of $t_j^h$ that are not $t$, are distinct.
    \end{claim}
    \begin{proof}
        By construction, the first part of the statement holds for the starting state $S^1$. Both rules either move vertices to an active $Q_i$ further along its path, or directly to $t$. The second part follows immediately from the fact that the paths $P_j$ are internally-disjoint.
    \end{proof}

    We first explain how \Cref{claim:token_left} implies the claim in the lemma. By \Cref{claim:token_left}, there is at least one $i \in [q]$ such that $t^T_i \ne t$; let $I \subset [q]$ be the set of all indices with this property. By \Cref{claim:end_tokens}, each vertex in $\{t^T_i\}_{i \in I}$ lies on an active $Q_j$ for some $j \in [q]$, and these vertices are all distinct. Moreover, since the rule \textbf{Push} is not applicable, no two of these vertices share the same $Q_j$. Let $\pi: I \to [q]$ be the injection that maps $i \in I$ to the index $j$ such that $t^T_i \in V(Q_j)$. We construct the desired family of paths as follows: for $i \in [q] \setminus I$, let $P_i'$ be $P_i$, and for $j \in I$, let $P_j'$ be a concatenation of the subpath of $P_i$ from $s$ to $t^T_j$ (denoted $\newpref_i$), and the subpath of $Q_{\pi(i)}$ from $t^T_j$ to $t$ (denoted $\newsuf_i$). Observe also that for each $i \in I$, $t^T_i$ is not $v_i$ since the rule \textbf{Clear} is not applicable. Moreover, since $I$ is non-empty, there is at least one path of type (ii) in the constructed family. It only remains to show that the paths $\{P'_i\}_{i \in [q]}$ are internally-disjoint.

    For $i \in [q] \setminus I$, denote $\newpref_i = P^\rightarrow_i$ and $\newsuf_i = P^\leftarrow_i$. For $i \ne i' \in [q]$, $s$ is the only intersection between $\newpref_i$ and $\newpref_{i'}$ since $\newpref_i$ and $\newpref_{i'}$ are proper prefixes of $P_i$ and $P_{i'}$ respectively, and these paths are internally-disjoint by the assumption of the lemma. Observe that for each $i \in [q]$, the vertices of $\newsuf_i$ except $t$ lie in the set $C_j$, for some $j \in [q]$, and this correspondence between $\{\newsuf_i\}_{i \in [q]}$ and $\{C_j\}_{j \in [q]}$ is a bijection defined by $\pi$ on $I$ and by the identity permutation on $[q] \setminus I$. Since the sets $\{C_j\}_{j \in [q]}$ are disjoint, for any $i \ne i' \in [q]$, we get that the paths $\newsuf_i$ and $\newsuf_{i'}$ share the only common vertex $t$.

    It remains to verify that for each $i \neq i' \in [q]$, $\newpref_i$ shares no common vertices with $\newsuf_{i'}$. For $i' \in [q] \setminus I$, $\newsuf_{i'} = P^\leftarrow_{i'}$, which is a suffix of $P_{i'}$, and this path cannot intersect $\newpref_i$ which is a prefix of $P_i$; thus in the following we assume $i' \in I$. Assume the contrary, then there is a vertex $u$ on $\newpref_i$ that belongs to $Q_{\pi(i')}$, and is located on $Q_{\pi(i')}$ closer to $t$ than $t^T_{i'}$, which is the starting vertex of $\newsuf_{i'}$. Observe that the rule \textbf{Clear} has never been applied to $Q_{\pi(i')}$, otherwise $Q_{\pi(i')}$ would not be active in $S^T$. Thus, there is a state $S^h$ where $t^h_i = u$, since any application of the rules to $t^{h'}_i$ for any $h' \in [T]$ either leaves this vertex in place or moves it to the next vertex of $P_i$ belonging to an active $Q_j$ for some $j \in [q]$. Now observe that no application of the rule \textbf{Push} makes the closest vertex to $t$ on $Q_{\pi(i')} - \{t\}$ among $\{t^h_j\}_{j \in [q]}$ farther, by definition of \textbf{Push}. However, we get that $t^h_i$ is closer to $t$ on $Q_{\pi(i')}$ than $t^T_i$, which is the only vertex of $\{t^T_j\}_{j \in [q]}$ on $Q_{\pi(i')} - \{t\}$. This is a contradiction to the assumption that $\newpref_i$ and $\newsuf_{i'}$ intersect.

    \begin{proof}[Proof of \Cref{claim:token_left}]
        Assume the contrary, that $t^T_i = t$ for each $i \in [q]$. Since the paths $\{P^\leftarrow_i\}_{i \in [q]}$ are non-empty, $T > 1$.
        Thus, the state $S^{T - 1}$ is defined and $S^T$ is obtained by applying a rule to $S^{T - 1}$. First, observe that this rule could not have been \textbf{Push}, as it assumes there exist distinct $t^{T - 1}_i$ and $t^{T - 1}_{i'}$ on $Q_j - \{t\}$ for some $i, i', j \in [q]$, and only moves $t^{T - 1}_i$ away while keeping $t^T_{i'}$ on $Q_j - \{t\}$. Therefore, \textbf{Clear} has been applied to $S^{T - 1}$, replacing $t^{T - 1}_i = v_j$ by $t^T_i = t$, for some $i, j \in [q]$. Now, we claim that there is another $i' \in [q]$, $i' \ne i$ such that $t^{T - 1}_{i'} \in V(Q_j) \setminus \{t\}$. Indeed, by the starting assumption of the proof, there exists $i'$ such that $P^\rightarrow_{i'}$ and $Q_j$ share an internal vertex $u$. Since $Q_j$ is active until the last step, and since the application of any rule moves $t^h_{i'}$ to the next vertex of $P_{i'}$ intersecting some active $Q_{j'}$, there exists a step $h \in [T - 1]$ where $t^h_{i'} = u$. Before the step $T - 1$, only the rule \textbf{Push} could have been applied to $Q_j$, and an application of this rule never makes the intersection $\{t^h_{i''}\}_{i'' \in [q]} \cap (Q_j \setminus \{v_j, t\})$ empty. Therefore there exists $i'' \in [q]$ such that $t^{T - 1}_{i''} \in Q_j \setminus \{v_j, t\}$. This contradicts the assumption that $t^T_{i''} = t$ since the application of \textbf{Clear} to $Q_j$ moves $t^{T - 1}_{i''}$ to the next vertex on $P_{i''}$ on an active $Q_{j'}$; this will not take $t^{T}_{i''}$ farther than $v_{i''}$ along $P_{i''}$.
    \end{proof}

\end{proof}

Before we move to the proof of the main theorem, we note that the basic idea of random separation is to exploit random colorings of the vertex set. We, on the other hand, are first and foremost looking for a deterministic algorithm; the standard approach would be to enumerate a sufficiently ``expressive'' set of colorings, instead of trying a pre-set number of random colorings.
 Unfortunately, the existing results on derandomization of random separation algorithms cannot be applied directly, as normally random separation is considered for constant number of sets; most often two.
 Thus in the next lemma we directly construct a suitable family of functions by using the standard tool of perfect hash families, given by the classical result of Naor, Schulman, and Srinivasan~\cite{NaorSS95} (we refer to~\cite[Chapter 5]{cygan2015parameterized} for the detailed introduction to the concept). For integers $n$ and $k$, an $(n, k)$-\emph{perfect hash family} $\mathcal{F}$ is a family of functions from $[n]$ to $[k]$ such that for each set $S \subset [n]$ of size $k$ there exists $f \in \mathcal{F}$ that acts on $S$ injectively. We are now ready to state our derandomization lemma.

 \begin{lemma}
     For an $n$-element set $U$ and $q$ integers $k_1$, \ldots, $k_q$, $\sum_{i = 1}^q k_i= \ell$ there exists a family of functions $\mathcal{F}$ of size $q^{\Oh(\ell)} \log n$ mapping $U$ to $\{1, \ldots, q\}$ with the following property. For any disjoint sets $A_1$, \ldots, $A_q \subset U$ with $|A_i| = k_i$ for $i \in [q]$, there exists a function $f \in \mathcal{F}$ such that $f(x) = i$ if $x \in A_i$.
     Moreover, $\mathcal{F}$ can be computed in time $q^{\Oh(\ell)} n \log n$.
     \label{lemma:derandomize}
 \end{lemma}

 \begin{proof}
     First, construct an $(n, \ell)$-perfect hash family $\mathcal{H}$ of size $e^\ell \ell^{\Oh(\log \ell)} \log n$ in time $e^\ell \ell^{\Oh(\log \ell)} n \log n$ by the result of Naor, Schulman, and Srinivasan~\cite{NaorSS95}.
     For every $h \in \mathcal{H}$ and a partition $[t] = I_1 \cup \ldots \cup I_q$ such that $|I_i| = k_i$ for $i \in [q]$, add a function $f_{I_1, \ldots, I_q}^h$ to $\mathcal{F}$. The function acts as follows: for any $x \in U$, $f_{I_1, \ldots, I_q}^h(x) = i$ if $h(x) \in I_i$.

     We now show that $\mathcal{F}$ defined above satisfies the conditions of the lemma. Fix the subsets $A_1$, \ldots, $A_q$ of $U$, denote $A = A_1 \cup \ldots \cup A_q$, $|A| = \ell$. By the definition of an $(n, \ell)$-perfect hash family, there exists $h \in \mathcal{H}$ such that the images $h(x)$ are distinct for all $x \in A$. For each $i \in [q]$, define $I_i$ to be the set of indices that $h$ assigns to $A_i$. By definition, $f_{I_1, \ldots, I_q}^h(x) = i$ if $h(x) \in I_i$, and $h(x) \in I_i$ if and only if $x \in A_i$. It only remains to bound the number of partitions $I_1$, \ldots, $I_q$.
     \begin{claim}\label{claim:partitions}
         The number of partitions of $[\ell]$ into disjoint subsets $I_1$, \ldots, $I_q$ with $|I_i| = k_i$ for $i \in [q]$ is $\Oh(\ell \cdot q^\ell)$.
     \end{claim}
     \begin{proof}
         The number of partitions is equal to the multinomial coefficient
         \[\binom{\ell}{k_1, k_2, \ldots, k_q} = \frac{\ell!}{k_1! k_2! \cdots k_q!} \le e \sqrt{\ell} \frac{\left(\frac{\ell}{e}\right)^t}{\left(\frac{k_1}{e}\right)^{k_1}  \cdot \left(\frac{k_2}{e}\right)^{k_2} \cdots \left(\frac{k_q}{e}\right)^{k_q}} = \Oh\left(\ell \cdot \frac{\ell^\ell}{k_1^{k_1} \cdot k_2^{k_2} \cdots k_q^{k_q} } \right),\]
         by Stirling's formula. We now argue that $k_1^{k_1} \cdot k_2^{k_2} \cdots k_q^{k_q} \ge \left(\frac{\ell}{q} \right)^\ell$, which immediately implies that the desired number of partitions is bounded by $\Oh\left(\ell \cdot \frac{\ell^\ell}{\left(\frac{\ell}{q} \right)^\ell} \right) = \Oh(\ell \cdot q^\ell)$. For that, we observe that the function $f(x) = x \log x$ is strictly convex on $x \ge 1$, since $(x \log x)' = (1 + \log x)$, thus $x \log x + y \log y \ge 2 \cdot \left(\frac{x + y}{2} \cdot \log\frac{x + y}{2}\right)$, and $x^x \cdot y^y \ge \left(\frac{x + y}{2}\right)^{2 \cdot \frac{x + y}{2}}$ for any $x, y \ge 1$, where the equality only holds if $x = y$. Now, consider the function $h(x_1, x_2, \dots, x_q) = x_1^{x_1} \cdot x_2^{x_2} \cdots x_q^{x_q}$ defined on the polytope $K \subset \mathbb{R}^q$ bounded by $x_1 \ge 1$, \ldots, $x_q \ge 1$, $\sum_{i = 1}^q x_i = \ell$. Since $h$ is continuous on $K$ and $K$ is compact, $h$ attains its minimum in $K$. Assume that $h$ achieves its minimum on $x_1$, \ldots, $x_q \in K$ with $x_i \ne x_j$ for some $i, j \in [q]$. Then by the above, $h(x_1, \ldots, x_i, \ldots, x_j, \ldots, x_q) > h(x_1, \ldots,  \frac{x_i + x_j}{2}, \ldots, \frac{x_i + x_j}{2}, \ldots, x_q)$; the tuple on the right-hand side still belongs to $K$. Thus, such a $x_1$, \ldots, $x_q$ cannot achieve the minimum, and $h$ is minimized at the only point with equal coordinates, $(\ell/q, \ldots, \ell/q) \in K$. Since $(k_1, \ldots, k_q) \in K$, the claim is done.
     \end{proof}

     The bound on the size of $\mathcal{F}$ now follows directly from \Cref{claim:partitions}: $|\mathcal{F}| = |\mathcal{H}| \cdot \Oh(\ell \cdot q^\ell) = q^{\Oh(\ell)} \cdot n^{\Oh(1)}$.
 \end{proof}

 Finally, with \Cref{lemma:tokens} and \Cref{lemma:derandomize} at hand, we move to the proof of \Cref{thm:detmaintheorem} itself.

\begin{proof}[Proof of \Cref{thm:detmaintheorem}]
    First, we observe that finding an $(S, T)$-\lkg of order $p$ and total length at least $k$ is equivalent to finding an $(s, t)$-\lkg of order $p$ and total length at least $(k + 2)$, where moreover $s \ne t$ and $s$ is not adjacent to $t$.
    Indeed, consider the digraph $G'$ that is a copy of $G$ with two new vertices $s$ and $t$, where $N_{G'}^+(s) = S$ and $N_{G'}^-(t) = T$. Then, any directed $(s, t)$-\lkg of order $p$ and length $k + 2$ in $G'$ induces a directed $(S, T)$-\lkg of order $p$ and length $k$ in $G$ by removing $s$ and $t$, and vice versa. Thus for the rest of the proof we assume that the task is to find an $(s, t)$-\lkg of order $p$ and total size at least $k$, $s \ne t$, and $s$ is not adjacent to $t$. We now describe two separate subroutines of our algorithm, tailored for different cases of the maximum length of the path in the target $(s, t)$-\lkg.
    The \textbf{short case} succeeds if there is an $(s, t)$-\lkg where all paths have less than $2k$ internal vertices, and the \textbf{main case} succeeds otherwise (the proof of correctness follows after the description of the algorithm). 

    \textbf{Short case.} For each $i \in [p]$, we branch over the number of internal vertices $k_i$ of the $i$-th path in the target linkage, $1 \le k_i < 2k$. If $\sum_{i = 1}^p k_i < k - 2$, we disregard the choice of $\{k_i\}_{i = 1}^p$ and proceed to the next branch. Otherwise, consider a function family $\mathcal{F}$ given by an invocation of \Cref{lemma:derandomize} with $q = p$ and the current values of $k_1$, \ldots, $k_p$. Branch over the choice of $f \in \mathcal{F}$ and denote by $C_1, C_2, \ldots, C_p$ the vertices colored by the respective colors via $f$. For each $i \in [p]$, we use a deterministic algorithm for finding a directed $(s, t)$-path with exactly $k_i$ internal vertices in the graph $G[C_i \cup \{s, t\}]$ in time $2^{Oh(k_i)} \cdot n^{\Oh(1)}$. The fastest-known such algorithm is the algorithm of Zehavi~\cite{Zehavi15} running in time $\Oh(2.597^{k_i}) \cdot n^{\Oh(1)}$\footnote{While the result in~\cite{Zehavi15} is stated for finding an arbitrary path of certain length, it could be easily adjusted to finding an $(s, t)$-path.} If for some choice of $\{k_i\}_{i = 1}^p$ and $f$ the desired collection of paths is found, the algorithm returns it. If no branch succeeds, the algorithm reports a no-instance.

    We now argue for correctness of the algorithm above. Since the paths are internally-disjoint by construction and $\sum_{i = 1}^p k_i \ge k - 2$, if the algorithm returns a collection of paths, they clearly form a solution. In the other direction, fix a solution induced by directed $(s, t)$-paths $P_1^*$, \ldots, $P_p^*$, where for each $i$ the $i$-th path contains exactly $k_i < 2k$ internal vertices, and consider the respective branch of the algorithm above. If for each $i \in [p]$ the set $C_i$ contains the internal vertices of the $i$-th path, the algorithm succeeds, as $G[C_i \cup \{s, t\}]$ contains an $(s, t)$-path with exactly $k_i$ internal vertices. Denote by $A_i$ the set of internal vertices of $P_i^*$, for each $i \in [p]$,  \Cref{lemma:derandomize} guarantees that there exists $f \in \mathcal{F}$ that colors each $A_i$ in color $i$, concluding the proof of correctness in this case.

\begin{algorithm}[h]
\For{$q'=1, \ldots, p$}
{
    Invoke \Cref{lemma:derandomize} with $q = p + 1$, $k_1 = \cdots = k_{q'} = k$, $k_{q' + 1} = \cdots = k_p = 2k$, and $k_{p + 1} = k$, to obtain the function family $\mathcal{F}$\;
\ForEach{$f \in \mathcal{F}$}
{
    Denote by $C_0, C_1, \ldots, C_p$ the vertices colored by the respective colors via $f$\;\label{alg:det_main_step:start}

\For{$i=1, \ldots, p$}
{
\ForEach{$v_i\in C_i$ at distance $k$ from $t$ in $G_i=G[C_i\cup\{t\}]$}
{
Find a shortest $(v_i,t)$-path $Q_i$ in $G_i$\;\label{alg:det_main_step:shortest}
Find paths $P_1$, $P_2$, \ldots, $P_p$ in $G-(V(Q_i)\setminus\{v_i,t\})$, where $P_i$ is an $(s, v_i)$-path and for each $j \ne i$, $P_j$ is an $(s, t)$-path, such that no two paths share a vertex except for $s$ and $t$, and $P_i$ does not contain $t$\;\label{alg:det_main_step:flow}
\If{such paths $P_1$, $P_2$, \ldots, $P_p$ exist}
{
    Set $P_i=P_i \circ Q_i$\;
    \Return{the paths $P_1$, $P_2$, \ldots, $P_p$}\; 
}
}
}\label{alg:det_main_step:end}
}
}
\caption{Main case of the algorithm in \Cref{thm:detmaintheorem}.}\label{alg:det_main_step}
\end{algorithm}

    \textbf{Main case.} 
    The basic procedure is given in \Cref{alg:det_main_step}. Observe that the task in \Cref{alg:det_main_step:flow} can be easily reduced to an instance of network flow, thus the whole procedure given in \Crefrange{alg:det_main_step:start}{alg:det_main_step:end} runs in polynomial time. If no iteration returns a collection of paths, the algorithm reports a no-instance.

    It is easy to observe that if \Cref{alg:det_main_step} returns a collection of paths, then these paths constitute a solution to the given instance. Indeed, by construction, $P_1$, \ldots, $P_p$ are internally-disjoint $(s, t)$-paths. The length of $Q_i$ is exactly $k$, since the path $P_i$ contains $Q_i$ as a subpath, the $(s, t)$-\lkg given by the paths $P_1$, \ldots, $P_p$ is of order $p$ and length at least $k$. It remains to verify that if there exists an $(s, t)$-\lkg of order $p$ where there is a path with at least $2k$ internal vertices, then our algorithm successfully returns a collection of paths for some choice of $f$. Denote the $p$-many $(s, t)$-paths that form a solution of minimum total length by $P_1^*$, \ldots, $P_p^*$, assuming that the paths are ordered from longest to shortest. In particular, $|V(P_1^*) \setminus \{s, t\}| \ge 2k$. Denote by $q'$ the maximum index such that $|V(P_{q'}^*) \setminus \{s, t\}| \ge 2k$, $1 \le q' \le p$.

    We say that the coloring $C_0$, $C_1$, \ldots, $C_p$ \emph{agrees} with the paths $P_1^*$, \ldots, $P_p^*$ if the following holds:
    \begin{enumerate}[label=(\roman*)]
        \item for each $i \in [q']$, the last $k$ internal vertices of $P_i^*$ belong to $C_i$,
        \item the first $k$ vertices of $P_1^*$ belong to $C_0$,
        \item for each $j \in [p] \setminus [q']$, all internal vertices of $P_j^*$ belong to $C_j$.
    \end{enumerate}
    For $i \in [q']$, denote by $A_i$ the last $k$ internal vertices of $P_i^*$; by $A_{p + 1}$ the first $k$ vertices of $P_1^*$; for $j \in [p] \setminus [q']$, denote by $A_j$ all internal vertices of $P_j^*$ plus arbitrary vertices not yet part of any $A_i$ so that $|A_j| = 2k$. By \Cref{lemma:derandomize}, there exists $f \in \mathcal{F}$ that induces a coloring $C_1$, \ldots, $C_{p + 1}$ with $A_i \subset C_i$ for each $i \in [p + 1]$, meaning that this coloring agrees with the solution $P_1^*$, \ldots, $P_p^*$.
    In the remainder of the proof, we argue that for this choice of $f$ \Cref{alg:det_main_step} outputs a solution.

    For $i \in [p]$, let $v_i$ be the vertex of $P_i^*$ at distance exactly $k$ from $t$ along the path. Since the coloring agrees with the solution, the last $k$ internal vertices of $P_i^*$ belong to $C_i$, including $v_i$. If $v_i$ is at distance $k$ from $t$ in $G_i$, denote by $Q_i$ the shortest $(v_i, t)$-path in $G_i$ that the algorithm finds on \Cref{alg:det_main_step:shortest}; otherwise denote by $Q_i$ an arbitrary shortest $(v_i, t)$-path in $G_i$. We now apply \Cref{lemma:tokens} to the paths $P_1^*$, \ldots, $P_q^*$ in the graph $G$ with selected disjoint vertex subsets $C_1$, \ldots, $C_q$, with selected vertices $v_1$, \ldots, $v_q$and paths $Q_1$, \ldots, $Q_q$. By the lemma, there exist internally-disjoint $(s,t)$-paths $P_1'$, \ldots, $P_q'$, such that for each $i \in [q]$, $P_i'$ is either $P_i^*$ or a concatenation of a prefix $\newpref_i$ of $P_1^*$ not extending beyond $v_i$, and a suffix $\newsuf_i$ of $Q_j$ for some $j \in [q]$; moreover, at least one of the paths is of the second type. 

    First, we claim that the paths $P_1'$, \ldots, $P_q'$, $P_{q + 1}^*$, \ldots, $P_p^*$ together form an $(\{s\}, \{t\})$-linkage of length at least $k$. Indeed, the paths $P_1^*$, \ldots, $P_p^*$ are internally-disjoint from the beginning; thus for each $q < i < j \le p$, $P_i^*$ and $P_j^*$ do not share common internal vertices. Moreover, for each $i \in [q]$, $j \in [p] \setminus [q]$, $P_i'$ and $P_j^*$ are immediately internally-disjoint in case $P_i' = P_i^*$. In case $P_i'$ is a concatenation of $\newpref_i$ and $\newsuf_i$, the path $\newpref_i$ again does not share a vertex with $P_j^*$ as a prefix of $P_i^*$, except for $s$. The suffix $\newsuf_i - \{t\}$ on the other hand is fully contained in some $C_{i'}$, $i' \in [q]$, while $P_j^* - \{s, t\}$ is contained in $C_j$ by the property (iii) of the coloring. Finally, for the length observe that the path $P_1'$ contains the first $k$ internal vertices of the path $P_1^*$ since they belong to the set $C_0$ disjoint from $C_i$ for any $i > 0$.

    Consider now a path $P_i'$ that is not $P_i^*$ but a concatenation, $\newpref_i \circ \newsuf_i$.
    Since the paths $P_1^*$, \ldots, $P_p^*$ come from a solution of smallest total length, it cannot be that $P_i'$ is shorter than $P_i^*$. On the other hand, the length of $\newsuf_i$ is at most $k$, which is the length of the suffix of $P_i^*$ from $v_i$ to $t$, and $\newpref_i$ is at most as long as the prefix of $P_i^*$ from $s$ to $v_i$. Thus it has to be that $\newpref_i$ is exactly the prefix of $P_i^*$ from $s$ to $v_i$, and $\newsuf_i$ is $Q_i$; additionally, $Q_i$ is then of length exactly $k$, thus $v_i$ is at distance $k$ from $t$ in $G_i$, and $Q_i$ is the shortest path chosen by the algorithm in \Cref{alg:det_main_step:shortest}.

    We now claim that the existence of $(s, t)$-\lkg $P_1'$, \ldots, $P_q'$, $P_{q + 1}$, \ldots, $P_p$ implies that \Cref{alg:det_main_step:flow} is executed successfully for the respective choice of $v_i$. Indeed, denote $P_i = \newpref_i$, for each $i' \in [q]$, $i' \ne i$, denote $P_{i'} = P_i'$, and for each $j \in [p] \setminus [q]$, denote $P_j = P_j^*$. This collection of paths is of form required by the conditions in \Cref{alg:det_main_step:flow}, and thus the algorithm of \Cref{alg:det_main_step:flow} returns a suitable collection of paths as well (not necessarily the same). This concludes the proof of correctness for \textbf{Main case}.

    Finally, observe that the running time of both cases is dominated by the invocation of \Cref{lemma:derandomize} with $q = p + 1$ and $\ell \le 2k (p + 1)$, resulting in the total running time of $p^{\Oh(kp)} \cdot n^{\Oh(1)}$.
\end{proof}


\section{Conclusion}\label{sec:concl} 
We conclude with several concrete open questions. The first question is about derandomizing \Cref{thm:maintheorem}, even for the case when $p=1$.
The algorithm in  \Cref{thm:maintheorem} is based on DeMillo-Lipton-Schwartz-Zippel lemma  for  polynomial identity testing, and therefore we do not expect to derandomize it using similar techniques~\cite{kabanets2004derandomizing,DBLP:journals/mst/MontoyaM13}.
However, similarly to \Cref{thm:detmaintheorem}, we do not exclude that other methods could result in (maybe slower) deterministic algorithms.
We are not aware of any \emph{deterministic} and \classFPT in $k$ algorithm for \ProblemName{Maximum Colored Path}.

The second question is about the \textsc{Disjoint Paths} problem. Here for a given set of pairs of terminal vertices $(s_1, t_1), \dots, (s_r, t_r)$, the problem is to decide whether there are vertex-disjoint $(s_i,t_i)$-paths, $i\in [r]$. The problem is FPT parameterized by $r$ by the seminal algorithm of Robertson and Seymour  \cite{RobertsonS95b}. A natural extension of this problem would be on colorful graphs, where we want the disjoint paths to collect at least $k$ colors. We do not know whether the colored variant of the problem is FPT parameterized by $k$ even for  \textsc{$2$-Disjoint Paths}, that is for $r=2$.

The third question concerns extending \Cref{thm:main-frameworks}, where we demand matroid  $M$ to be represented as a matrix  over a finite field of order $q$. 
The natural question here  is whether there is an FPT algorithm for $k$-ranked $(s,t)$-path (and more generally, for $k$-ranked $(S,T)$-\lkg of order $p$) in frameworks $(G,M)$, where  $M$ is a linear matroid represented as a matrix over rationals.
We also ask what is the complexity of this problem when $M$ is given by an independence oracle.
As was shown by Jensen and Korte~\cite{JensenK82}, various matroid problems have unconditional complexity lower bounds asserting that they admit no algorithms where the number of oracle calls is bounded by a polynomial on the size of the matroid ground set. 
For example, this concerns the classical \textsc{Matroid Parity} problem that can be solved in polynomial time on linear matroids as it was shown by Lov\'{a}sz (see, e.g.,~\cite{LovaszPlummerbook876}).
It is natural to ask whether such a lower bound can be shown for $k$-ranked $(s,t)$-path.

The last concrete question is about \probLongSTP and \probLongCycle.
Our algorithm implies the first $2^k n^{\Oh(1)}$ time algorithms for these problems, and the dependency on $k$ in the time complexity of our algorithm is unlikely to be improved in the general colored case.
However, it remains an interesting open problem whether \probLongSTP or \probLongCycle could be solved in time $(2-\varepsilon)^k n^{\Oh(1)}$ for some $\varepsilon > 0$, especially keeping in mind that \probkPath admits an $1.66^k n^{\Oh(1)}$ time algorithm~\cite{BjorklundHKK17}.

\bibliographystyle{siam}
\bibliography{Frameworks}

\end{document}

\end{document}